%% file: main.tex
\documentclass[aps,pra,10pt,twocolumn,notitlepage,showpacs,showkeys,preprintnumbers,amsmath,amssymb,nofootinbib]{revtex4-2}

\input{preamble}

\input{figures.tex}
\newcommand{\bsu}[1]{\mathbf{#1}}
\usepackage{setspace}

\makeatletter
\def\l@subsubsection#1#2{}
\makeatother

\newcommand{\manualtocentry}[2]{%
  \noindent%
  \hyperref[#1]{\ref*{#1}.\quad #2}%
  \nobreak{\dotfill}\nobreak%
  \pageref{#1}%
  \\[0.1cm]
}

\newcommand{\manualtocsubentry}[2]{%
  \phantom{~}\hspace{3em}%
  \hyperref[#1]{\ref*{#1}.\quad #2}%
  \nobreak{\dotfill}\nobreak%
  \pageref{#1}%
  \\[0.1cm]
}

\newcommand{\altmanualtocentry}[2]{%
  \noindent%
  \hyperref[#1]{#2}%
  \nobreak{\dotfill}\nobreak%
  \pageref{#1}%
  \\[0.1cm]
}
\newcommand{\manualtocentryapp}[2]{%
  \noindent%
  \hyperref[#1]{\ref*{#1}.\quad #2}%
  \nobreak{\dotfill}\nobreak%
  \pageref{#1}%
  \\[0.1cm]
}

\newcommand{\manualtocsubentryapp}[2]{%
  \phantom{~}\hspace{3em}%
  \hyperref[#1]{\ref*{#1}.\quad #2}%
  \nobreak{\dotfill}\nobreak%
  \pageref{#1}%
  \\[0.1cm]
}

\begin{document}
\title{On the Complexity of the Succinct State Local Hamiltonian Problem}
\author{Gabriel Waite}\email{gabriel.waite@student.uts.edu.au}
\author{Karl Lin}
\affiliation{Centre for Quantum Computation and Communication Technology,%
Centre for Quantum Software and Information, School of Computer Science,%
Faculty of Engineering and Information Technology,%
University of Technology Sydney, NSW 2007, Australia}
\begin{abstract}
    We study the computational complexity of the \textsc{Local Hamiltonian} problem under the promise that its ground state is succinctly represented.
    We show that the \textsc{Succinct State $2$-Local Hamiltonian} problem, for qubit Hamiltonians, is (promise) \textbf{MA}-{complete}.
    The approach combines a systematic characterisation of succinct quantum states, defined through arithmetic over specific number fields, with a refined reduction that lowers the locality of Feynman-Kitaev circuit-Hamiltonians from $6$ to $2$, without increasing particle dimension.
    This reveals a complexity phase transition, parameterised by locality, and extends the scope of previously known \textbf{MA}-complete problem instances.
    Our results further clarify how succinctness behaves under circuit-based constructions, and progresses toward a better understanding of the boundary between efficiently describable and efficiently verifiable quantum systems.
\end{abstract}
\maketitle

\begin{center}
\manualtocentry{sec:introduction}{Introduction}
    \manualtocsubentry{sec:summary-of-results}{Summary of Results}
    \manualtocsubentry{sec:related-work}{Related Work}
\manualtocentry{sec:preliminaries}{Preliminaries}
    \manualtocsubentry{sec:binary-representations}{Binary Representations}
    \manualtocsubentry{sec:number-classes}{Binary Number Classes}
    \manualtocsubentry{sec:complexity}{Computational Complexity}
\manualtocentry{sec:succinct-states}{Succinct States}
    \manualtocsubentry{sec:subset-states}{Properties of Subset States}
    \manualtocsubentry{sec:operations-subset-states}{Operations with Subset States}
    \manualtocsubentry{sec:operations-hybrid-subset-states}{Operations with Hybrid Subset States}
    \manualtocsubentry{sec:properties-general-succinct-states}{Properties of General Succinct States}
    \manualtocsubentry{sec:multi-alphabet-query-access}{Multi-Alphabet Query Access}
\manualtocentry{sec:LHP-succinct-ground-states}{The Succinct State Local Hamiltonian Problem}
    \manualtocsubentry{sec:containment}{Class Containment}
    \manualtocsubentry{sec:extension-MA-containment}{Extension of MA Containment}
    \manualtocsubentry{sec:hardness}{Class Hardness}
\manualtocentry{sec:locality-reduction}{Locality Reduction}
    \manualtocsubentry{sec:3-local-reduction}{Reduction to 3-Local Hamiltonians}
    \manualtocsubentry{sec:2-local-reduction}{Reduction to 2-Local Hamiltonians}
\manualtocentry{sec:conclusion}{Conclusion}
\altmanualtocentry{sec:declarations}{Declarations}
\altmanualtocentry{sec:refs}{References}
\altmanualtocentry{app:toc}{Appendices}
\end{center}

\newpage
\input{body.tex}

\bibliographystyle{apsrev4-2}
\bibliography{ref}\label{sec:refs}

\onecolumngrid
\appendix
\input{appendix_0.tex}
\newpage
\addtocontents{toc}{\protect\setcounter{tocdepth}{0}}
\section*{Table of Contents}\label{app:toc}

\begin{center}
\begin{minipage}{0.9\textwidth}
\manualtocentryapp{app:number-form}{Binary Number Class Structure}
\manualtocentryapp{app:proofs}{Proof of Main Text Results}
\manualtocentryapp{app:stoquastic-hamiltonians-easy-witness}{Local Stoquastic Hamiltonians with Easy Witness Ground States}
\manualtocentryapp{app:toffoli-gate-decomposition}{Toffoli Gate Decomposition}
\manualtocentryapp{app:pre-idled-quantum-verifier}{Pre-idled Quantum Verifier Scenario}
\manualtocentryapp{app:local-hamiltonians-spatially-sparse-graphs}{Local Hamiltonians on Spatially Sparse Graphs}
\manualtocentryapp{app:2local}{Proof of 2-local Hamiltonian Construction}
    \manualtocsubentryapp{app:regular-interval-structured-circuits}{Regular Interval Structured Toffoli-Equivalent Circuits}
    \manualtocsubentryapp{app:2-local-hamiltonian-construction}{2-Local Hamiltonian Construction}
    \manualtocsubentryapp{app:tilde_H}{Proof of Proposition~\ref{prop:tilde_H}}
\end{minipage}
\end{center}

\input{appendix_a.tex}
\input{appendix_b.tex}
\input{appendix_c.tex}
\input{appendix_d.tex}
\input{appendix_e.tex}
\input{appendix_f.tex}
\input{appendix_g.tex}
\end{document}

%% file: preamble.tex
\usepackage[margin=0.75in]{geometry}
\usepackage{natbib}
\usepackage{newtxtext}
\usepackage{physics}
\usepackage{amsmath}
\usepackage{amsfonts}
\usepackage{amssymb}
\usepackage{mathtools}
\usepackage{xfrac}
\usepackage{stmaryrd}

\usepackage{graphicx}
\usepackage{tikz} 
\definecolor{myOrange}{RGB}{245, 118, 26}
\definecolor{indigo}{rgb}{0.29, 0.0, 0.51}
\usetikzlibrary{patterns, 
decorations.pathmorphing, 
decorations.pathreplacing, 
decorations.shapes, 
shapes,
arrows,
positioning,
calc,
automata,
positioning}
\usepackage{soul}
\usepackage{xcolor}
\usepackage[LGR,T1]{fontenc}
\usepackage{url}
\usepackage{microtype}
\usepackage{hyperref}
\hypersetup{
    colorlinks,
    linkcolor=[RGB]{245, 118, 26},
    citecolor=[RGB]{245, 118, 26},
    urlcolor=[RGB]{245, 118, 26}
}
\usepackage[capitalise]{cleveref}
\usepackage{upgreek} 

\usepackage{algorithm2e}
\usepackage{algpseudocode} 
\RestyleAlgo{ruled}
\SetKwComment{Comment}{/* }{ */}
\SetKwInOut{Input}{input}\SetKwInOut{Output}{output}

\usepackage{amsthm}
\usepackage{thmtools} 
\usepackage{thm-restate} 
\newtheorem{theorem}{Theorem}
\newtheorem*{theorem*}{Theorem}

\newtheorem{conjecture}{Conjecture}
\newtheorem{corollary}{Corollary}
\newtheorem{lemma}{Lemma}

\newtheorem*{result*}{Result}
\newtheorem{proposition}{Proposition}
\theoremstyle{definition}
\newtheorem{definition}{Definition}
\newtheorem*{definition*}{Definition}
\theoremstyle{remark}
\newtheorem{remark}{Remark}
\AtEndEnvironment{remark}{\null\hfill\ensuremath{\diamond}}
\newtheorem{assumption}{Assumption}

\usepackage{paralist} 
\usepackage{multirow}
\usepackage{enumitem} 

\newcommand{\refcite}[1]{Ref.~\cite{#1}}

\newcommand{\poly}[1]{{\rm poly}(#1)}

\newcommand{\bin}[1]{{\tt #1}}
\newcommand{\conc}[2]{{#1}\mathbin\Vert{#2}}
\newcommand{\B}{\{0,1\}}
\newcommand{\re}{\mathfrak{R}}
\newcommand{\im}{\mathfrak{I}}
\newcommand{\bs}[1]{\boldsymbol{#1}}
\newcommand{\uu}{\mathbin\Vert}
\newcommand{\bsr}{\textbf{\textsl{D}}}
\newcommand{\dbbrckt}[1]{\llbracket #1 \rrbracket}
\newcommand{\mcb}[2]{\colorbox{#1}{$\displaystyle #2$}}
\newcommand{\rcb}[3]{
\begingroup
    \color{#1}
    \underbrace{\color{black} \mcb{#1}{\tt #2}}_{#3}
\endgroup
}
\newcommand{\cl}[1]{\textnormal{{\bf #1}}}
\newcommand{\clw}[2]{{\bf #1}-{\rm #2}}
\newcommand{\clsb}[2]{\textnormal{{\bf #1}\textsubscript{#2}}}

\newcommand{\Gate}[1]{{\fontfamily{cmr}\selectfont\textsc{#1}}}
\renewcommand{\sc}[1]{\textnormal{\textsc{#1}}}
\renewcommand{\ket}[1]{\lvert #1 \rangle}
\renewcommand{\bra}[1]{\langle #1 \rvert}
\renewcommand{\mel}[3]{\langle #1 \rvert #2 \lvert #3 \rangle}
\newcommand{\ind}{\boldsymbol{1}}

%% file: figures.tex
\tikzset{
    hierarchy/.pic = {
        \draw[pattern=north west lines, pattern color=gray, fill opacity=0.35, very thick] (0,3) ellipse (4 and 6.5) node[opacity=1, yshift=4cm] {$(\mathbb{C}^2)^{\otimes n}$};
        \draw[thick, draw=myOrange, dashed, fill=myOrange, fill opacity=0.2,scale=0.75] (0,3.5) ellipse (5 and 6.25) node[opacity=1,yshift=3cm] {$\mathbb{C}_p\dbbrckt{\omega}$};
        \draw[thick, draw=myOrange, fill=myOrange, fill opacity=0.2,scale=0.75] (0,3) ellipse (4.75 and 5.5) node[opacity=1,yshift=2.25cm] {$\mathbb{C}_p$};
        \draw[thick, draw=myOrange, fill=myOrange, fill opacity=0.2,scale=0.75] (0,2.25) ellipse (3.25 and 4) node[opacity=1, yshift=1.75cm] {$\mathbb{Q}_p$};
        \draw[thick, draw=myOrange, dashed, fill=myOrange, fill opacity=0.2,scale=0.75] (0,1.5) ellipse (3.35 and 3.1) node[opacity=1, yshift=1.25cm] {$\mathbb{Q}_p^+\dbbrckt{\sqrt{\cdot}_1}$};
        \draw[thick, draw=myOrange, fill=myOrange, fill opacity=0.2,scale=0.75] (0,1) circle (2.5) node[opacity=1, yshift=0.75cm] {$\mathbb{Q}^+_p$};
        \draw[thick, draw=myOrange, fill=myOrange, fill opacity=0.2,scale=0.75] (0,0) circle (1.25) node[opacity=1] {$\mathbb{N}_p$};
        \draw[fill=black] (0.45,0.45) circle (0.05);
        \draw[fill=black] (1.55,2.3) circle (0.05);
        \draw[fill=black] (1.1,4) circle (0.05);
        \draw[fill=black] (2.3,4.6) circle (0.05);
        \draw[fill=black] (1.8,6.2) circle (0.05);
        \node[draw] (A) at (7-0.75,1) {Subset states $\ket{S}$};
        \node[draw] (B) at (7-0.75,3) {\shortstack{Real succinct\\ states}};
        \node[draw] (C) at (7-0.75,7) {\shortstack{Complex succinct\\ states}};
        \node[draw, gray, opacity=0.8] (D) at (8-0.75,4.5) {MPS};

        \draw[thick,-latex, gray, opacity=0.6] (D.west) -- (2.3,4.6);
        \draw[thick,-latex, gray, opacity=0.6] (D.west) -- (1.1,4);
        \draw[thick,-latex] (A.west) -- (0.45,0.45);
        \draw[thick,-latex] (A.west) -- (1.55,2.3);
        \draw[thick,-latex] (B.west) -- (1.1,4);
        \draw[thick,-latex] (C.west) -- (2.3,4.6);
        \draw[thick,-latex] (C.west) -- (1.8,6.2);
    },
    flow/.pic = {
\begin{scope}[
            node distance=10mm and 10mm,
            box/.style = {draw=none, minimum height=5mm, inner xsep=3mm, inner sep=5pt, align=center, thick},
            scale=0.9,transform shape
            ]
            \coordinate (origa) at (0,0);
            \coordinate (origb) at (7.5,0);
            \coordinate (origc) at (-7.5,0);
            \coordinate (origd) at (3.5,-9);
            \coordinate (orige) at (-2.5,-9);
            
            \node[fill=gray!25,label={[draw=black,inner sep=2pt,fill=white]below:Ref.~\cite{jiang2025local}}] (a0) [box, below=of origa] {\sc{$6$-Loc. Ham.} \\[0.1cm] Complex Succ. G.S.\\[0.1cm] \clw{MA}{complete}};
            \node[fill=gray!25,label={[draw=black,inner sep=2pt,fill=white]below:Ref.~\cite{jiang2025local}}] (a1) [box, below=of a0] {\sc{$6$-Loc. Real Ham.} \\[0.1cm] Real Succ. G.S.\\[0.1cm] \clw{MA}{complete}};
            \node[fill=gray!25,label={[draw=black,inner sep=2pt,fill=white]below:Ref.~\cite{liu2021stoqma}}] (a2) [box, below=of a1] {\sc{$6$-Loc. Stoq. Ham.} \\[0.1cm] Real Succ. G.S.\\[0.1cm] \clw{MA}{complete}};
            
            \node[label={[draw=black,inner sep=2pt,fill=white]below:Refs.~\cite{bravyi2006complexity, jiang2025local} + This Work}] at (7.5,0.5) {\cl{MA}};
            \node[fill=gray!25,] (b0) [box, below=of origb] {\sc{$k$-Loc. Ham.} \\ Complex Succ. G.S.};
            \node[fill=gray!25,] (b1) [box, below=of b0] {\sc{$(k+1)$-Loc. Ham.} \\[0.1cm] Real Succ. G.S.};
            \node[fill=gray!25,] (b2) [box, below=of b1] {\sc{Stoq. Ham.} \\[0.1cm] Real Succ. G.S.};
            \node[label={[draw=black,inner sep=2pt,fill=white]below:Refs.~\cite{bravyi2006complexity, jiang2025local} + This Work}] at (-7.5,0.5) {\clw{MA}{hard}};
            \node[fill=gray!25,] (c0) [box, below=of origc] {\sc{$6$-Loc. Ham.} \\ Complex Succ. G.S.}; 
            \node[fill=gray!25,] (c1) [box, below=of c0] {\sc{$6$-Loc. Real Ham.} \\ Real Succ. G.S.}; 
            \node[fill=gray!25] (c2) [box, below=of c1] {\sc{$6$-Loc. Stoq. Ham.} \\ Real Succ. G.S.}; 
            \node[fill=myOrange!35,label={[draw=black,inner sep=2pt,fill=white]above: This Work}] (c3) [box, below=of c2, yshift=-1cm] {\sc{$4$-Loc. Ham.} \\ Complex Succ. G.S.}; 
            \node[fill=myOrange!35,label={[draw=black,inner sep=2pt,fill=white]above: This Work}] (c4) [box, below=of c3] {\sc{$4$-Loc. Real Ham.} \\ Real Succ. G.S.}; 
            \node[fill=myOrange!35,label={[draw=black,inner sep=2pt,fill=white]above: This Work}] (c5) [box, below=of c4] {\sc{$4$-Loc. Stoq. Ham.} \\ Real Succ. G.S.}; 
            \node[fill=myOrange!35,label={[draw=black,inner sep=2pt,fill=white]above: This Work}] (c6) [box, below=of c5] {\sc{$3$-Loc. Ham.} \\ Complex Succ. G.S.}; 
            \node[fill=myOrange!35,label={[draw=black,inner sep=2pt,fill=white]above: This Work}] (c7) [box, below=of c6] {\sc{$2$-Loc. Ham.} \\ Complex Succ. G.S.}; 
            \node[label={[draw=black,inner sep=2pt,fill=white]below: This Work}] at (-7.5,-18.5) {\clw{MA}{hard}};
            \node[fill=myOrange!35,label={[draw=black,inner sep=2pt,fill=white]below: This Work}] (d0) [box, below=of origd] {\sc{$3$-Loc. Ham.} \\ Complex Succ. G.S. \\ \clw{MA}{complete}};
            \node[fill=myOrange!35,label={[draw=black,inner sep=2pt,fill=white]below: This Work}] (d1) [box, below=of d0] {\sc{$2$-Loc. Ham.} \\ Complex Succ. G.S. \\ \clw{MA}{complete}};
            \node[fill=myOrange!35,label={[draw=black,inner sep=2pt,fill=white]below: This Work}] (e0) [box, below=of orige] {\sc{$4$-Loc. Ham.} \\ Complex Succ. G.S. \\ \clw{MA}{complete}};
            \node[fill=myOrange!35,label={[draw=black,inner sep=2pt,fill=white]below: This Work}] (e1) [box, below=of e0] {\sc{$4$-Loc. Real Ham.} \\ Real Succ. G.S. \\ \clw{MA}{complete}};
            \node[fill=myOrange!35,label={[draw=black,inner sep=2pt,fill=white]below: This Work}] (e2) [box, below=of e1] {\sc{$4$-Loc. Stoq. Ham.} \\ Real Succ. G.S. \\ \clw{MA}{complete}};

            \begin{pgfonlayer}{background layer}
                \draw[-latex, very thick, gray] (a0) -- (a1);
                \draw[-latex, very thick, gray] (a1) -- (a2);
                \draw[thick,pattern=north west lines, pattern color=myOrange!35] (5,-0.5) rectangle (10,-7);
                \draw[-latex, very thick, gray] (b0) -- (b1);
                \draw[-latex, very thick, gray] (b1) -- (b2);
                \draw[thick,pattern=north west lines, pattern color=myOrange!35] (-5,-0.5) rectangle (-10,-7);
                \draw[thick,pattern=north west lines, pattern color=myOrange!35] (-5,-7.5) rectangle (-10,-18);
                \draw[-latex, very thick, gray] (c1) -- (c0);
                \draw[-latex, very thick, gray] (c2) -- (c1);
        
                \draw[-latex, thick, dashed, gray, opacity=0.6] (b0.west) -| (3.75,-1.5) |- (a0.east);
                \draw[-latex, thick, dashed, gray, opacity=0.6] (c0.east) -| (-3.75,-1.5) |- (a0.west);
                \draw[-latex, thick, dash dot, gray, opacity=0.6] (c1.east) -| (-4.5,-4.5) |- (a1.west);
                \draw[-latex, thick, dash dot, gray, opacity=0.6] (b1.west) -| (4.5,-4.5) |- (a1.east);
                \draw[-latex, thick, dash pattern={on 7pt off 2pt on 1pt off 3pt}, gray, opacity=0.6] (b2.west) -| (5.25,-7.25) |- (a2.east);
                \draw[-latex, thick, dash pattern={on 7pt off 2pt on 1pt off 3pt}, gray, opacity=0.6] (c2.east) -| (-5.25,-7.25) |- (a2.west);
        
                \draw[-latex, very thick] (a0) -- ++(0,-1.5) -| (e0.north);
                \draw[-latex, very thick] (a0) -- ++(0,-1.5) -| (d0.north);
                \draw[-latex, very thick] (e0) -- (e1);
                \draw[-latex, very thick] (e1) -- (e2);
        
                \draw[-latex, very thick, dashed] (c3) -- (c2);
                \draw[-latex, very thick] (c4) -- (c3);
                \draw[-latex, very thick] (c5) -- (c4);
        
                \draw[-latex, thick, dashed, gray, opacity=0.6] (c3.east) -| (-5.05,-10) |- (e0.west);
                \draw[-latex, thick, dash dot, gray, opacity=0.6] (c4.east) -| (-5.15,-12) |- (e1.west);
                \draw[-latex, thick, dash pattern={on 7pt off 2pt on 1pt off 3pt}, gray, opacity=0.6] (c5.east) -| (-5.25,-14) |- (e2.west);
        
                \draw[-latex, thick, dashed, gray, opacity=0.6] (b0.east) -| (3,-9.5) -| (0,-9.5) |- (e0.east);
                \draw[-latex, thick, dash dot, gray, opacity=0.6] (b1.east) -| (3.75,-12) -| (1,-12) |- (e1.east);
                \draw[-latex, thick, dash pattern={on 7pt off 2pt on 1pt off 3pt}, gray, opacity=0.6] (b2.east) -| (4.5,-15) -| (2,-15) |- (e2.east);
        
                \draw[-latex, thick, dotted, gray, opacity=0.6] (c6.south) |- (-5,-16) -| (-4.5,-17.25) -- (7,-17.25) -- (7,-12) |- (d0.east);
                \draw[-latex, thick, dotted, gray, opacity=0.6] (c7.south) |- (-5,-18.25) -|  (d1.south)++(0,-0.45);
        
                \draw[-latex, thick, dotted, gray, opacity=0.6] (b0.east) -| (9.75,-10) -| (7,-10) |-(d0.east);
                \draw[-latex, thick, dotted, gray, opacity=0.6] (b0.east) -| (9.75,-12) -| (6,-12) |-(d1.east);
            \end{pgfonlayer}
        \end{scope}
    }

}

%% file: body.tex
\section{Introduction}\label{sec:introduction}
Understanding when quantum computational problems admit classical verification remains a central challenge in complexity theory.
The \sc{Local Hamiltonian} problem, estimating ground-state energies of quantum many-body systems, is typically \clw{QMA}{complete}, requiring quantum witnesses and verification~\cite{kitaev2002classical}.
Complexity classifications for finding the energy of extremal product states~\cite{kallaugher2024complexity}, classifying free-fermion Hamiltonians~\cite{elman2021free,chapman2023unified}, and inclusion of auxiliary information~\cite{richter2007two,bravyi2015monte,stroeks2022spectral,weggemans2024guidable} reveal vastly different complexity landscapes depending on the structure imposed, demonstrating that the boundary between classical and quantum verification is sensitive to the structural assumptions placed on the problem.

We investigate this boundary through local Hamiltonians whose ground states are succinct: efficiently and exactly describable by classical algorithms that compute amplitudes (see \cref{def:succinct-state} for a formal definition).
Realising such access in practice is extremely demanding; computing amplitudes for general quantum states is \clw{GapP}{hard}~\cite{fortnow1999complexity,fenner1999determining}.
Nevertheless, when provided with it, a striking complexity boundary emerges.
We consider the \sc{Succinct State Local Hamiltonian} problem, where the task is to estimate the ground-state energy given the promise that the ground state is succinctly represented (see \cref{def:LHP-succinct-state} for a formal statement).
Prior work established \clw{MA}{completeness} for $6$-local Hamiltonians~\cite{liu2021stoqma,jiang2025local}.
Since physical Hamiltonians typically exhibit low-locality interactions, it is natural to ask whether this complexity classification persists at lower locality and indicates a more fundamental computational boundary.

In this work, we establish \clw{MA}{completeness} for $2$-local Hamiltonians acting on qubits, aligning with broader efforts to determine how structural constraints influence computational difficulty~\cite{kempe2006complexity,oliveira2008complexity,schuch2009computational,piddock2017complexity}.
The fact that hardness remains at this lower locality for qubit systems suggests that succinctly described ground states maintain computational challenges even when interactions become more physically realistic.
We achieve this via direct clock-based reductions with careful gate decompositions that maintain succinctness throughout, bypassing perturbative gadgets entirely; we describe these techniques in detail below.
Our results clarify a sharp complexity phase transition for the problem since the $1$-local case is solvable in \cl{NC} (Nick's Class~\cite{aaronson2005complexity}).

The formal study of the \sc{Succinct State Local Hamiltonian} problem is recent.
Liu~\cite{liu2021stoqma} demonstrated that the \sc{Real Succinct State $6$-Local Stoquastic Hamiltonian} problem is \clw{MA}{complete} through an analysis of \cl{eStoqMA}, which extends \cl{StoqMA} with \emph{easy witnesses} --- uniform superposition states with an efficient procedure to verify membership in the support.
Jiang~\cite{jiang2025local} removed the stoquastic restriction by employing fixed-node quantum Monte Carlo techniques~\cite{tenHaaf1995proof}, requiring knowledge of the ground state's amplitudes in the computational basis~\cite{bravyi2023rapidly}, and using the Feynman-Kitaev construction~\cite{kitaev2002classical,bravyi2006complexity}.
Both results required $6$-local interactions.
Lowering locality while preserving succinct structure is challenging: standard perturbative gadget techniques~\cite{kempe2006complexity,oliveira2008complexity} do not preserve succinct structure, as simulator Hamiltonians reproduce low-energy subspaces but not the amplitudes of the original states.
Our approach instead employs direct clock-based reductions with careful gate decompositions that maintain succinctness throughout, bypassing the need for perturbative gadgets entirely.

To achieve our result, we examine the representation of succinct states, constructing complex structures from simpler cases through a systematic, compositional approach.
We introduce new succinct-state encodings based on algebraic representations of complex numbers, extending beyond the rational succinct states of Ref.~\cite{jiang2025local}.
This allows us to characterise the succinct structure of Feynman-Kitaev circuit-Hamiltonian ground states~\cite{kitaev2002classical} and bring new classes of Hamiltonians into the \clw{MA}{complete} regime.
The rational-amplitude framework of Ref.~\cite{jiang2025local} does not apply to general circuits or to the $2$-local reduction we perform; our approach overcomes this limitation.
Moreover, we prove that each required component can be computed efficiently by classical algorithms while preserving the succinct structure throughout, thus providing a faithful polynomial-time reduction.

We propose two complementary approaches for reducing locality to three.
The first builds on the clock construction of Kempe and Regev~\cite{kempe2003local}, adding penalties for illegal clock states and decoupling two clock qubits from propagation terms.
The second decomposes reversible circuits into \emph{structured} Clifford$+T$ circuits over the gate set $\{\Gate{Cnot}, \Gate{Had}, T\}$; this suffices to reduce locality to three while extending \clw{MA}{hardness} to general Hamiltonians.
Combined with our analysis of succinct states, this yields a polynomial-time (Karp) reduction from arbitrary \clw{MA}{verification circuits} to $3$-local Hamiltonians with succinct ground states.

To obtain $2$-local Hamiltonians, we introduce a new gate set $\mathcal{R} = \{\Gate{C}Z, \Gate{Had}, T, Z\}$ and prove that the structured circuits defined previously can be faithfully represented using gates from $\mathcal{R}$ such that each $\Gate{C}Z$ gate occurs at a fixed regular interval, surrounded by a prescribed number of single-qubit gates.
This regular-interval structure is essential for the $2$-local clock Hamiltonian construction of Kempe, Kitaev, and Regev~\cite{kempe2006complexity}: it ensures that each $\Gate{C}Z$ gate is surrounded by a fixed number of single-qubit timesteps, allowing propagation of $\Gate{C}Z$ gates to be encoded as a sum of $2$-local operators without introducing multi-qubit gate terms into the Hamiltonian.
We prove that \clsb{MA}{q} circuits can be efficiently transformed into this regular-interval structured circuit family while preserving the computed unitary and the completeness and soundness parameters of the original circuits.
Applying the construction of Ref.~\cite{kempe2006complexity} to these circuits and verifying that the resulting ground state remains succinct then yields a polynomial-time reduction from arbitrary \clw{MA}{verification circuits} to $2$-local Hamiltonians with succinct ground states, establishing \clw{MA}{completeness} of the \sc{Succinct State $2$-Local Hamiltonian} problem.
We expect our techniques to have broader applicability, particularly for succinct state preparation and verification algorithms.

\subsection{Summary of Results}\label{sec:summary-of-results}

Our main technical contributions divide naturally into two parts: 
\begin{inparaenum}[1.]
    \item the characterisation of succinct quantum states, and 
    \item the classification of computational complexity arising from their representations.
\end{inparaenum}
More focused discussions of our results and techniques follow in the sequel.
For comments on the theoretical and practical aspects of the \sc{Succinct State Local Hamiltonian} problem as well as remarks on related works, see \cref{sec:LHP-succinct-ground-states}.

\subsubsection{Characterisation of Succinct States}
We begin by formalising the notion of a \emph{succinct state}. 
This requires a careful treatment of the number systems over which amplitudes are represented. 
Throughout, we consider sets of numbers, denoted generically by $\mathbb{S}_{p(n)}$, that can be represented \emph{exactly} using $p(n)$ bits, where $p(n)$ is polynomial in the system size $n$.

Specifically, we work with the following natural families: the naturals $\mathbb{N}_{p(n)}$, the rationals $\mathbb{Q}_{p(n)}$, the non-negative rationals $\mathbb{Q}^+_{p(n)}$, and the complex rationals $\mathbb{C}_{p(n)}$, whose real and imaginary parts are rational.
We also extend these to include algebraically encoded numbers and related fields; see \cref{sec:preliminaries} for formal definitions.

Capturing \emph{algebraic encodings} is crucial for our analysis. 
This permits the inclusion of irrational quantities, such as $1/\sqrt{2}$, while remaining exactly representable through finite binary descriptions.
Elements of $\mathbb{Q}_{p(n)}$ already require a compound encoding containing numerator, denominator, and sign bits.
We can equivalently view such binary representations as finite classical algorithms computing amplitudes exactly.

Formally, let $\mathbb{A}^{(\#)}_q$ denote the set of $q$-bit binary encodings for an algebraic characteristic $\#$ (e.g., a sign, a root, or a phase factor).
We then define the family of algebraically encoded number sets as
\begin{equation*}
    \mathbb{S}_{p}\dbbrckt{\#_q} \coloneqq \mathbb{A}^{(\#)}_q \times \mathbb{S}_{p}.
\end{equation*}
Throughout, $\mathbb{S}$ may denote any of $\{\mathbb{N}, \mathbb{Q}, \mathbb{C}\}$.

\subparagraph{Definition of succinct states.}
A quantum state $\ket{\psi}$ is said to be \emph{$\mathbb{S}$-succinct} if all of its amplitudes belong to $\mathbb{S}_{p(n)}$ and there exists an efficient classical query algorithm $\mathcal{Q}_\psi$ that computes these amplitudes exactly (see~\cref{def:succinct-state}). 
We occasionally denote such a state by the triple $(\ket{\psi},\mathbb{S}_{p(n)},\mathcal{Q}_\psi)$.

This framework lets us ask structural questions: when we combine two succinct states, does the result remain succinct?
For example, under tensor product, addition, or projection, does succinctness persist, and how do the corresponding query models combine?
We show that these properties follow from simple polynomial-time operations on binary encodings --- that is, efficient classical post-processing of amplitudes.

\begin{restatable*}[]{lemma}{alma}
    \label{lma:tensor-product-C-succinct}
    Consider two $\mathbb{S}$-succinct states $(\ket{\psi},\mathbb{S}_{p(n)},\mathcal{Q}_{\psi})$ and $(\ket{\phi}, \mathbb{S}_{q(m)},\mathcal{Q}_{\phi})$, where $p(n)$ and $q(m)$ are polynomial functions of $n$ and $m$ respectively.
    Then the tensor product $\ket{\psi}\ket{\phi}$ is an $\mathbb{S}_{2r(s)+1}$-succinct state with an efficient classical query algorithm $\mathcal{Q}_{\psi\phi}$.
    Here $s = \max\{n,m\}$ and $r(s) = \max\{p(s),q(s)\}$.
\end{restatable*}

\begin{restatable*}[]{lemma}{blma}
    \label{lma:complex-succinct-state-to-real-succinct-state}
    Let $\ket{\phi}$ be a $\mathbb{C}_{p(n)}$-succinct state with classical query algorithm $\mathcal{Q}_{\phi}$ such that each amplitude $\alpha(j) = R(j) + {\rm i} I(j)$, where $R(j), I(j) \in \mathbb{Q}_{p(n)}$.
    Then 
    \begin{align*}
        \ket{\phi} &= \sum_{j \in \B^n} R(j)\ket{j} + {\rm i}\sum_{j \in \B^n} I(j)\ket{j}
        = \ket{\phi_{R}} + {\rm i} \ket{\phi_I}.
    \end{align*}
    Define two orthogonal states $\ket{\varphi_1} = \ket{\phi_R}\ket{0} + \ket{\phi_I}\ket{1}$ and $\ket{\varphi_2} = \ket{\phi_R}\ket{0} - \ket{\phi_I}\ket{1}$.
    Then $\ket{\varphi_1}$ and $\ket{\varphi_2}$ are $\mathbb{Q}_{p(n)}$-succinct states with efficient query algorithms $\mathcal{Q}_{\varphi_1}$ and $\mathcal{Q}_{\varphi_2}$ respectively.
\end{restatable*}

\subparagraph{Subset states and reversible circuits.}
To connect succinctness with computational models, we next study \emph{subset states}, which are uniform superpositions over subsets $S \subseteq \B^n$ denoted by $\ket{S}$. 
We show that the history state of a classically reversible circuit is a subset state. 
Such states naturally fall into two categories: 
\begin{inparaenum}
    \item those for which membership in $S$ can be verified by a query algorithm, and 
    \item those where the uniform amplitude value itself can be queried directly.
\end{inparaenum}

\begin{restatable*}[]{lemma}{clma}
    \label{lma:classical-gates-subset state}
    Let $\bsr$ be a bit string of size $O(\poly{n})$ describing a circuit consisting of ${K = O(\poly{n})}$ classically reversible gates $\{R_k\}_{k\in [K]}$.
    For a subset state $\ket{S}$ on $S \subseteq \B^n$, define
    \begin{align*}
        \ket{A_k} &\coloneqq R_k\ket{S}, &
        \ket{B_k} &\coloneqq R_k \cdots R_1\ket{S}.
    \end{align*}
    Then $\ket{A_k}$ and $\ket{B_k}$ are $\mathbb{N}_1$-succinct states with efficient classical query algorithms $\mathcal{Q}_{A_k}$ and $\mathcal{Q}_{B_k}$ respectively.
\end{restatable*}

\begin{restatable*}[]{remark}{armk}
    \label{rmk:subset-more-general}
    Allowing for algebraic encodings, any subset state on $S \subseteq \B^n$ can also be viewed as a 
    $\mathbb{Q}^+_{\log_2{\abs{S}}} \dbbrckt{\sqrt{\cdot}_1}$-succinct state, i.e., $\mathbb{A}_1^{(\sqrt{\cdot})}\times\mathbb{Q}^+_{\log_2{\abs{S}}}$.
    \null\hfill\ensuremath{\diamond}
\end{restatable*}

\subparagraph{Beyond classical reversibility.}
We finally extend this to include non-classical gates such as the Hadamard and $T$. 
This raises the question: can a classical query algorithm still describe the amplitudes after such gates act?
We show that for a restricted but expressive class of circuits, the answer is \emph{yes}.
When the full circuit can be encoded by a polynomial-length bit string, there exists a classical query algorithm capable of producing amplitudes exactly at certain points in the computation.

\begin{restatable*}[]{lemma}{dlma}
    \label{cor:CRG-T-T-dagger-Hadamard-gate-seq-subset state}
    Let $\bsr$ encode a circuit consisting of $O(\poly{n})$ classically reversible gates, $O(n)$ $T$ gates, $O(n)$ $T^\dagger$ gates, and $O(1)$ Hadamard gates.
    Let $K$ denote the total number of gates, and write $\{U_k\}_{k\in [K]}$ for the ordered sequence.
    For each $k \in [K]$, define the state
    \begin{equation*}
        \ket{H_k} = U_k \cdots U_1\ket{S}.
    \end{equation*}
    Then $\ket{H_k}$ is an $\mathbb{N}_p \dbbrckt{\frac{1}{\sqrt{2}}_p}$-succinct state for $p = O(1)$.
\end{restatable*}

\begin{restatable*}[]{lemma}{elma}
    \label{lma:history-state-ex-2}
    The superposition state
    \begin{equation}\label{eq:history-state-ex-2}
        \ket{\eta} = \frac{1}{\sqrt{\abs{K}}} \sum_{k=1}^{K} \ket{H_k}\ket{k},
    \end{equation}
    is a $\mathbb{C}_{r(n)}\dbbrckt{\sqrt{\cdot}_1}$-succinct state, where $r(n) = \poly{n}$, with the efficient classical (query) algorithm $\mathcal{Q}_{\eta}$.
\end{restatable*}

Our framework characterises how fixed query-access models can be extended to accommodate additional quantum states, and thus may find applications beyond the scope of this work.  
In particular, access to a transformed state $U\ket{\psi}$ through $\mathcal{Q}_\psi$ is possible whenever the unitary $U$ has bounded spread of quantum gates~\cite{vandennest2011simulating}.  
However, this structural restriction prevents efficient querying of the state amplitudes in arbitrary bases.  
By contrast, in models that allow direct query access to arbitrary quantum states, one could request the query oracle for $U\ket{\psi}$ itself, without being confined to computational-basis queries.

\subsubsection{Locality Reduction Techniques}
We first review standard reductions between Hamiltonian classes that preserve essential spectral and structural properties relevant to our succinct-state framework.  
The reduction from a complex $k$-local Hamiltonian to a real Hamiltonian proceeds by observing that the imaginary unit ${\rm i}$ is isomorphic to the real matrix 
$\big(\begin{smallmatrix}0 & -1\\1 & 0\end{smallmatrix}\big) \equiv -{\rm i}Y$.  
Splitting $H$ into real and imaginary parts thus yields a $(k+1)$-local real Hamiltonian acting on $n+1$ qubits, whose spectrum forms a $2$-multiset of the original.  
This reduction preserves query access, ground-state structure, and ground energy to within polynomial resources: queries to the real Hamiltonian can be implemented through informed queries to the complex one, and the correspondence extends to the ground state itself.  
We also analyse the reduction from real to stoquastic Hamiltonians via the fixed-node quantum Monte Carlo method~\cite{tenHaaf1995proof}, and show that it preserves all structural features relevant to succinct encodings.

\subparagraph{Succinctness and verification protocols.}  
Since succinct states can be represented algebraically, we argue that the \cl{MA} protocol of \refcite{jiang2025local} is robust against families of algebraically encoded succinct states.  
This follows naturally from the structure of the history state and the decomposition circuits used therein.
More specifically, our characterisation of succinct states has demonstrated polynomial-time classical algorithms for expressing amplitudes, which suffice for the verification steps of the protocol.

\begin{restatable*}[]{corollary}{acorol}
    The \cl{MA} protocol of \textnormal{Ref.~\cite{jiang2025local}} is robust under the inclusion of $\mathbb{C}_{p(n)}\dbbrckt{\sqrt{\cdot}}$-succinct states.
\end{restatable*}

\subparagraph{Normalisation of the history state.}  
We revisit the \clw{MA}{hardness} proof to ensure the history state is correctly normalised and that its amplitudes can be expressed exactly using a polynomial number of bits, fitting our notion of succinctness.  
Alternative normalisation methods, such as padding with identity gates, can enforce rational amplitudes, but for broader applicability we allow algebraic encodings.

\begin{restatable*}[]{proposition}{aprop}
\label{prop:history-state-solution-assumption}
    The history state $\ket{\eta(x,\chi)}$ associated with the Feynman-Kitaev clock construction for \clsb{MA}{q} circuits $\mathcal{V}$ is a subset state on 
    \begin{equation*}
        \mathcal{S} \coloneqq \bigcup_{k=0}^{K} \bigg( \big(\prod_{j=k}^{0} R_k \circ S \big) \times \{1^k\,0^{K-k}\}\bigg),
    \end{equation*}
    where 
    \begin{align*}
        S &= \{x\} \times \{\chi\} \times \{0\}^m \times \{0,1\}^p,\\
        \mathcal{V} &= \{R_K,\dots,R_1,R_0\};
    \end{align*}
    with $x$ an $n$-bit string, $\chi$ a $w$-bit string, and $R_0 = I$.  
    Hence, $\ket{\eta(x,\chi)}$ is a $\mathbb{Q}^+_{q(n)}\dbbrckt{{\sqrt{\cdot}}_1}$-succinct state, where $q(n) = \log_2(2^p(K+1))$.
\end{restatable*}

\subparagraph{Reduction hierarchy.}  
Our first locality reduction applies the clock-construction arguments of Kempe and Regev~\cite{kempe2003local}, adapted to \clsb{MA}{q} circuits.

\begin{restatable*}[]{theorem}{athrm}
\label{thrm:4-local-stoquastic-hamiltonian-with-Qp-succinct-ground-states}
    The \sc{$\mathbb{Q}^+_{p(n)}\dbbrckt{\sqrt{\cdot}_1}$-Succinct State $4$-Local Stoquastic Hamiltonian} problem is \clw{MA}{complete}.
\end{restatable*}

The second reduction refines this to a locality of three by decomposing classically reversible circuits into Clifford+$T$ form.  
We define \emph{structured Toffoli-equivalent circuits}, built from $\{\Gate{Cnot},\Gate{Had},T\}$ gates under strict structural constraints, leading to a new promise class \clsb{StMA}{q} (see \cref{def:StMAq}).

\begin{restatable*}[]{lemma}{flma}
    $\clsb{StMA}{q} = \clsb{MA}{q}$.
\end{restatable*}

The standard \clw{MA}{hardness} arguments extend directly, implying that these structured circuits yield $3$-local Hamiltonians.

\begin{restatable*}[]{theorem}{bthrm}
\label{thrm:main-result-3l}
    The \sc{$\mathbb{C}_{p(n)}\dbbrckt{\sqrt{\cdot}}$-Succinct State $3$-Local Hamiltonian} problem is \clw{MA}{complete}.
\end{restatable*}

As a corollary, the \sc{Succinct State Local Hamiltonian} problem remains \clw{MA}{complete} even for Hamiltonians defined on spatially sparse graphs (see Appendix~\ref{app:local-hamiltonians-spatially-sparse-graphs}).  
If a perturbative reduction preserving succinctness were discovered, our result would serve as a starting point for geometric reductions~\cite{oliveira2008complexity}.  

Our final reduction to $2$-local Hamiltonians builds on the structured circuit decomposition approach.
Starting from structured Toffoli-equivalent circuits (STECs) over the gate set $\mathcal{G} = \{\Gate{Cnot}, \Gate{Had}, T\}$, we pass to a new family of circuits over the gate set $\mathcal{R} = \{\Gate{C}Z, \Gate{Had}, T, Z\}$ in which each $\Gate{C}Z$ gate occurs at a fixed regular interval $\ell$ in the circuit, meaning that consecutive $\Gate{C}Z$ gates are separated by exactly $\ell$ single-qubit gates.
This regularity condition is necessary for applying the $2$-local clock Hamiltonian construction of Ref.~\cite{kempe2006complexity}.
Without it, the propagation term for a $\Gate{C}Z$ gate would act simultaneously on two computational qubits and a clock qubit, making it $3$-local; the regular-interval padding distributes this interaction across adjacent single-qubit timesteps so that each resulting Hamiltonian term remains $2$-local, with the action of the $\Gate{C}Z$ gate emerging only upon restriction to the nullspace of the appropriate Hamiltonian penalty terms.

We call circuits satisfying this condition \emph{regular-interval structured Toffoli-equivalent circuits} (RI-STECs) and the corresponding promise class \clsb{RIStMA}{q}.
The gate set $\mathcal{R}$ generates unitaries with amplitudes in $\mathbb{Q}({\rm i}, \sqrt{2})$, the same field as for STECs, and the succinctness of the resulting ground states follows from the equivalences $\Gate{C}Z \leftrightarrow \{\Gate{Cnot}, \Gate{Had}\}$ and $Z \leftrightarrow T$-type gates, together with the succinct state arguments established for the $3$-local case.

We first conclude that RI-STECs are computationally equivalent to \cl{MA}\textsubscript{q} circuits.

\begin{restatable*}[]{corollary}{ricor}\label{cor:RIStMA}
    $\clsb{RIStMA}{q} = \clsb{MA}{q}$.
\end{restatable*}

Applying the direct $2$-local clock Hamiltonian construction of Ref.~\cite{kempe2006complexity} to RI-STEC verification circuits then yields our main result.

\begin{restatable*}[]{theorem}{cthrm}
\label{thrm:main-result-2l}
    The \sc{$\mathbb{C}_{p(n)}\dbbrckt{\sqrt{\cdot}}$-Succinct State $2$-Local Hamiltonian} problem is \clw{MA}{complete}.
\end{restatable*}

Together with \cref{thrm:main-result-3l}, this shows that the \clw{MA}{completeness} of the \sc{Succinct State Local Hamiltonian} problem is robust to the locality of the Hamiltonian: the problem remains \clw{MA}{complete} for any locality $k \geq 2$.
This proves a sharp boundary in the problem's complexity phase diagram, parameterised by the locality $k$, since the $1$-local case is exactly solvable in polynomial time by diagonalising each term independently.

Beyond $2$-local, further reductions appear difficult with current techniques.
Perturbative gadgets do not preserve succinct ground states: the simulator Hamiltonian $\widetilde{H} = H + V$ that reproduces a target Hamiltonian $H_{\mathrm{targ}}$ need not share its succinct ground space, since gadget methods preserve low-energy subspaces but not explicit state representations.
Consequently, the existing \clw{MA}{hardness} results for $2$-local stoquastic Hamiltonians~\cite{bravyi2006complexity} do not directly extend to the succinct setting via this approach, and new techniques would be required to handle the stoquastic case or consider geometric constraints.

\subsection{Related Work}\label{sec:related-work}
Recent prior work has explored the complexity of deciding the ground-state energy of local Hamiltonians to inverse-polynomial precision under different settings.
Notably,~\citet{stroeks2022spectral} demonstrated that there exists polynomial-time classical and quantum algorithms for estimating the ground-state energy of local Hamiltonians with polynomial-gapped eigenvalues and given an input state with specific classical access.
For example, it was demonstrated that it is possible to classically learn a constant number of eigenstates for a stoquastic local Hamiltonian when the specific set of eigenvalues are well-separated and there exists an input state with efficient classical sample access, and at least inverse-polynomial overlap with a constant number of eigenstates.
Similar algorithms in Ref.~\cite{stroeks2022spectral} were studied for slightly more general settings with the main global assumption requiring the existence of a state that has non-negligible overlap with at most a polynomial number of eigenstates.
The classical results require sample access and the ability to compute amplitude ratios.
Unfortunately, identifying a state with such properties is no easy task.
The setting of the problem considered in this work bears resemblance to the work of~\citet{stroeks2022spectral} in that we also consider the existence of a state that permits efficient classical computation of amplitude ratios and a goal of deciding the ground-state energy of a local Hamiltonian.
However, the main difference is that we assume the ground state has the classical access properties, rather than some ``guiding state'' that is not necessarily the ground state.
Furthermore, our problem uses a query access model which is not necessarily amenable to the classical sample access model, and we make no assumption on the separation between eigenvalues.
Therefore, the results of~\citet{stroeks2022spectral} do not directly apply to our setting, else $\cl{MA} \subseteq \cl{BPP}$, and our findings are not in contradiction with theirs.

The \sc{Guidable Local Hamiltonian} problem \cite{bravyi2015monte,deshpande2022importance,gharibian2023dequantizing, weggemans2024guidable} relates closely to the problem we study.
This guidable variant assumes there exists some guiding state having overlap with the ground state; the state is not given as input to the problem.
It follows from the results of~\citet{gharibian2023dequantizing} that if we relax the condition of the ground state being succinct and instead assume the existence of a guiding state with a succinct representation, allowing for perfect sampling-access and constant overlap, then the problem of deciding the ground-state energy to constant precision is the class \cl{MA}.
Further conclusions from~\citet{weggemans2024guidable} suggest that when the guiding state is succinctly represented but has overlap, at most, inverse-polynomially close to unity, the problem is \clw{QCMA}{hard}.
It then follows that the assumption of the ground state being succinctly represented is strong, especially when resolving the ground-state energy to inverse-polynomial precision.
In fact, these results may help shed light on the resolution of Conjecture~\ref{conj:relaxed-guiding-state} (defined in \cref{sec:conclusion}).

\section{Preliminaries}\label{sec:preliminaries}
We assume familiarity with the basic concepts and conventions of quantum computing~\cite{nielsen2010quantum} and complexity theory~\cite{watrous2008quantum, kitaev2002classical}; for surveys on quantum Hamiltonian complexity see Refs.~\cite{gharibian2015quantum, hamiltonianjungle2023}.
Several proofs of results in the main body are deferred to Appendix~\ref{app:proofs}, with the proof of \cref{thrm:main-result-2l} given in Appendix~\ref{app:2local}.

Let $\omega$ represent the primitive $8$-th root of unity, i.e., $\omega = {\rm e}^{2\pi{\rm i}/8}$.
Note that any subscript on $\omega$ does not refer to another root of unity.

For a generic state $\ket{\psi} \in (\mathbb{C}^2)^{\otimes n}$ expressed as a superposition state in the computational basis, we denote the amplitude of a computational basis state $\ket{j}$ as $\braket{j}{\psi} \eqqcolon \alpha(j)$.
We denote the support of a state $\ket{\psi} \in (\mathbb{C}^2)^{\otimes n}$ as ${\text{supp}(\ket{\psi}) \coloneqq \{j \in \B^n : \alpha(j)\neq 0\}}$.

For an $n$-qubit normalised state $\ket{\psi}$, $\mathcal{Q}_\psi : \B^n \to \mathbb{S}$ denotes a map from $n$-bit strings $j$ to an algorithm $\mathcal{S}$, encoding a complex number corresponding to the computational basis amplitude $\alpha(j)$.
We refer to $\mathcal{Q}_\psi$ as a \emph{query algorithm} for the state $\ket{\psi}$, where ``to query'' implies the ability to request a specific computational basis amplitude.
Unless otherwise specified, we assume the cost of querying $\mathcal{Q}_\psi$ is $O(1)$.

For a set of $m$ quantum gates $\mathcal{G} = \{g_1, \cdots, g_m\}$, denote $\mathbb{F}_j$ as the field for which the entries of the gate $g_j$ are defined.
Let $\mathbb{F}_\mathcal{G}$ be the smallest field containing the entries of any unitary $U$ produced by a polynomial-length sequence of gates from $\mathcal{G}$.

\subsection{Binary Representations}\label{sec:binary-representations}
A positional number system represents numerical values using a base $b$ and digits $\{x_j\}$,
\begin{equation*}
    x = \sum_j x_j \, b^j.
\end{equation*}
Throughout this work, we use the binary positional system with base $2$ and digits $\{0,1\}$.
For a real number $d$, we denote its binary representation by ${\rm bin}(d)$, and its bit length by $|{\rm bin}(d)|$.
An \emph{exact binary representation} is one that uses a finite number of bits to represent the value exactly in binary.
Rational numbers admit such finite representations, whereas irrational or non-terminating fractions require infinitely many bits.

We write $\B^n$ for the set of all $n$-bit strings.
For $x \in \B^n$, $x[i]$ denotes its $i$-th bit ($1 \leq i \leq n$), and for a subset $S \subseteq \B^n$ we define the indicator function
\begin{equation*}
    \ind_{S}(x) =
    \begin{cases}
        1 & \text{if } x \in S,\\
        0 & \text{otherwise.}
    \end{cases}
\end{equation*}

The concatenation of two bit strings $x \in \B^n$ and $y \in \B^m$ is
\begin{equation*}
    \conc{x}{y} \coloneqq (x[1], \ldots, x[n], y[1], \ldots, y[m]) \in \B^{n+m},
\end{equation*}
and for two subsets $X \subseteq \B^n$ and $Y \subseteq \B^m$,
\begin{equation*}
    X \times Y \coloneqq \{ \conc{x}{y} : x \in X, y \in Y \}.
\end{equation*}

An \emph{algebraic encoding} refers to such a concatenated bit string whose substrings represent distinct components of a numerical or algebraic quantity.
For example, if $x$ and $y$ are binary encodings of the real and imaginary parts of a complex number $z$, then $\conc{x}{y}$ encodes $z$.
Formally, for an algebraic characteristic $\#$, we define
\begin{equation}\label{eq:alge-enc-def}
    \mathbb{A}^{(\#)}_{p_{\#}} \coloneqq \{\alpha \in \B^{p_\#}\},
\end{equation}
as the set of exact binary representations of $\#$ using $p_{\#}$ bits.
Characteristics of interest include the sign of a number, powers of the imaginary unit, and frequencies associated with specific irrational numbers.

\subsection{Binary Number Classes}\label{sec:number-classes}
We define the notation 
\begin{equation}\label{eq:Np-def}
    \mathbb{N}_p \coloneqq \{{\rm bin}(n) : n\in\mathbb{N}, ~ n \leq 2^p\},
\end{equation}
for the set of all natural numbers exactly representable in $p$ bits, i.e., unsigned integers, and
\begin{equation}\label{Qplusp-def}
    \begin{split}
        \mathbb{Q}^+_p &\coloneqq \{{\rm bin}(q): q\in\mathbb{Q}^+, ~ q=\frac{n}{m}, ~ n,m \in \mathbb{N}_p, ~ m\neq 0\} \\
        &\subset \mathbb{N}_p \times \mathbb{N}_p,
    \end{split}
\end{equation}
for the set of all positive (unsigned) rational numbers algebraically encoded as a numerator and denominator, each exactly represented in $p$ bits.
Note that the denominator is never $0$.
We also define
\begin{equation}\label{eq:Qp-def}
    \mathbb{Q}_p \coloneqq \mathbb{A}_{1}^{({\rm sgn})} \times\mathbb{A}_{1}^{({\rm sgn})} \times \mathbb{Q}_p^+,
\end{equation}
as the set of all (signed) rational numbers, algebraically encoded with a sign bit for both the numerator and denominator, each exactly represented in $p$ bits.
From this point on, we assume the first sign bit is for the numerator and the second is for the denominator.
Finally, we have 
\begin{equation}\label{eq:Cp-def}
    \mathbb{C}_p \coloneqq \mathbb{Q}_p \times \mathbb{Q}_p,
\end{equation}
as the set of all complex numbers, algebraically encoded, with a real and imaginary part exactly represented as rational numbers in $p$ bits.
From this point on, we assume the first half of the bits in the algebraic encoding are for the real part and the other half are for the imaginary part.
Ideas similar to these have been explored in the context of gate sets for specific classes of problems~\cite{giles2013exact,nam2020approximate,rudolph2024towards}.
We additionally note that the set of numbers $\mathbb{C}_p$ is closely related to the Gaussian rationals $\mathbb{Q}({\rm i})$; we however favour our notation to separate from algebraic properties of fields like $\mathbb{Q}({\rm i})$, such as closure --- a property not important for our purposes.

It is not hard to show that for elements in each set above, the length of the binary strings are $p$, $2p$, $2p+2$, and $4p+4$ respectively.
Note that the asymptotic length of each string is $\Theta(p)$.
Additionally, it is easy to see that
\begin{align*}
    \forall\, n \in \mathbb{N}_p,&~0 \leq n\leq 2^p, \\
    \forall\, q \in \mathbb{Q}_p^+,&~2^{-p} \leq q \leq 2^p, \\
    \forall\, q \in \mathbb{Q}_p,&~2^{-p} \leq \abs{q} \leq 2^p, \\
    \forall\, z \in \mathbb{C}_p,&~2^{-p} \leq \abs{\re(z)},\abs{\im(z)} \leq 2^p,\\
    &\implies~2^{-p} \leq \abs{z} \leq 2^{p+\frac{1}{2}} .
\end{align*}
The set containment ${\mathbb{N}_p \subset \mathbb{Q}_p^+ \subset \mathbb{Q}_p \subset \mathbb{C}_p}$ follows trivially.
See Appendix~\ref{app:number-form} for examples on the explicit form of these binary strings.
We acknowledge our encodings are not optimal in terms of space and for hardware implementations.
Our results will follow for more practical encodings, such as using the two's complement representation for signed integers or representing algebraic numbers as roots of polynomials in $\mathbb{Z}[x]$.
However, for simplicity and ease of explanation, we will use the encodings defined above.

As a final remark, we introduce two more general sets of algebraic encoded numbers:
\begin{equation}\label{eq:complex-algebraic-set}
    \mathbb{C}_p\llbracket\#_{q_\#}\rrbracket \coloneqq \mathbb{A}_{q_{\#}}^{(\#)} \times \mathbb{C}_p.
\end{equation}
The action of the algebraic characteristic $\#$ may be distributed in different ways across the real and imaginary parts of the complex number.
We typically assume the action is distributed locally across the parts, i.e., for some $\alpha \in \mathbb{A}^{(\#)}$ we have $\alpha\circ a + {\rm i}\alpha\circ b$.
Dependent on the quantity $\#$, it is not necessarily true that $\alpha\circ(a + {\rm i} b) = \alpha\circ a + {\rm i}\alpha\circ b$.

For a more general scenario, we define
\begin{equation*}
    \mathbb{S}_p^{\star}\llbracket \#_{q_{\#}} \rrbracket \coloneqq \mathbb{A}_{q_{\#}}^{(\#)} \times \mathbb{S}_p^{\star};
\end{equation*}
$\mathbb{S}_p^{\star}$ is a set of numbers for a specific characteristic $\star$.
We use $\mathbb{S} \in \{\mathbb{N}, \mathbb{Q}, \mathbb{C}\}$ to denote a general set.

\subsection{Computational Complexity}\label{sec:complexity}
Complexity classes considered in this work refer to the \emph{promise problem} variants (unless explicitly specified otherwise), rather than language classes.
We drop all ``promise'' prefixes, for example \cl{promiseMA} is simply \cl{MA}.
A notion of ``hard'' or ``complete'' problems is appropriate under standard Karp reducibility.

We use the circuit model to define complexity classes and denote $\bsr$ as the representation of circuit from the uniformity condition we impose on the circuit families.
That is, $\bsr$ encodes a circuit $C$ specifying: 
\begin{inparaenum}[(1)]
    \item the sequence of gates in $C$ and the register they act on, 
    \item the initialisation of the input register and 
    \item the categorisation of the input and output registers.
\end{inparaenum}
For brevity, we define $\mathbb{K} = \{\ket{0},\ket{1}\}$; thus a state $\ket{\chi} \in \mathbb{K}^{\otimes w}$ implies $\chi \in \B^w$.

\begin{definition}[Semi-Classical Verification Circuit~\cite{waite2025complexity}]
    A semi-classical verification circuit is a tuple ${F_n = (n,w,m,p,U)}$ where $n$ is the number of input qubits, $w$ is the number of proof qubits, $m$ is the number of ancillae initialised in the $\ket{0}$ state and $p$ is the number of ancillae initialised in the $\ket{+}$ state.
    The circuit $U$ is a quantum circuit on $M\coloneqq n + w + m + p$ qubits, comprised of $K = O(\poly{n})$ gates from the set $\{X, \Gate{Cnot}, \Gate{Toffoli}\}$.
    The acceptance probability of a semi-classical verification circuit $F_n$, given some input string $x\in \varSigma^n$ and a proof state $\ket{\chi} \in \mathbb{K}^{\otimes w}$ is defined as:
    \begin{equation*}
        \Pr\left[F_n(x,\ket{\chi})\right]= \bra{\phi}U^\dagger \Pi_{\text{out}} U \ket{\phi},
    \end{equation*}
    where $\ket{\phi} = \ket{x,\chi,0^{m},+^{p}}$ and $\Pi_{\text{out}} = \ketbra{1}_1$ is a projector onto the output qubit.
\end{definition}

Note that $w,m,p = O(\poly{n})$.

\begin{definition}[\clsb{MA}{q}~\cite{bravyi2006complexity}]
    A promise problem ${L = (L_{\textsc{yes}}, L_{\textsc{no}})}$ belongs to the class \clsb{MA}{q} if there exists a polynomial-time generated stoquastic circuit family $\mathcal{F} = \{F_n : n \in \mathbb{N}\}$, where each semi-classical circuit $F_n$ acts on $n + w+m + p$ input qubits and produces one output qubit, such that:
    \begin{itemize}
        \item[] \textbf{Completeness}: For all $x\in L_{\textsc{yes}}$, $\exists \ket{\chi}\in\mathbb{K}^{\otimes w}$, such that, $ \Pr\left[F_{|x|}(x,\ket{\chi})=\mathtt{1} \right] \geq 2/3$
        \item[] \textbf{Soundness}: For all $x\in L_{\textsc{no}}$, $\forall\ket{\chi}\in\mathbb{K}^{\otimes w}$, then, $ \Pr\left[F_{|x|}(x,\ket{\chi})=\mathtt{1} \right] \leq 1/3$
    \end{itemize}
\end{definition}

Without loss of generality we always assume $p$ and $m$ are even for \clsb{MA}{q}.
It was shown by Bravyi \emph{et al}.~\cite{bravyi2006complexity} that $\clsb{MA}{q} = \cl{MA}$, and by Liu~\cite{liu2021stoqma} that ${\cl{eStoqMA}=\cl{MA}}$.
Furthermore, \cl{MA} admits amplification to perfect completeness \cite{zachos1987probabilistic}, i.e., $\clsb{MA}{$1$} = \cl{MA}$.

A sequence of $K = O(\poly{n})$ classically reversible gates $\{R_j\}_{j \in [K]}$ can be expressed in a $O(\poly{n})$ sized tuple (bit string) $\bsr$ such that, the gate parameters\footnote{Gate type, control and target qubits, index in sequence} are encoded in $\bsr$.
Similarly, specific quantum circuits can also be encoded this way.
For example, a quantum circuit comprised of $\{X, \Gate{Cnot}, \Gate{Toffoli},T\}$ can be encoded in a $O(\poly{n})$ sized bit string $\bsr$.
For a quantum circuit comprised of arbitrary phase gates, the encoding is more complex and may require a more sophisticated encoding scheme.\footnote{Likely one would have to specify the phase up to some precision.}

The gates $X$, $\Gate{Cnot}$, $\Gate{Toffoli}$ have entries over $\mathbb{F}_2$, where as the gates $T$ and $T^\dagger$ have entries over $\mathbb{Q}({\rm i},\sqrt{2})$.
When the proof state $\ket{\chi}$ is a computational basis state, the amplitudes of the superposition state during the evolution of an \cl{MA} circuit can be expressed as $a/2^{p/2}$ for some $a\in\mathbb{N}$.

\section{Succinct States}\label{sec:succinct-states}
To exactly specify a generic quantum state would require an infinite amount of classical information, due to the continuous nature of its amplitudes.
A more practical approach is to describe quantum states approximately, up to a certain precision in a chosen norm, such as the trace norm or the $\ell^2$ norm.
Even then, the number of bits required for such approximate descriptions can scale exponentially with the number of qubits.
Nevertheless, there exists a subclass of quantum states that can be \emph{exactly} specified using only a polynomial number of bits.
Our focus lies on a particular family of such states, which we refer to as \emph{succinct states}.
In addition to exact descriptions, slightly larger families allow for faithful \emph{approximations} using a polynomial number of bits.

The states we consider are equipped with query access $\mathcal{Q}_\psi(x)$ that provides an efficient classical algorithm to compute the amplitudes $\braket{x}{\psi}$ exactly (modulo a scaling factor).
This access model captures a powerful form of classical control over a quantum state, going beyond simple preparation arguments.
Indeed, even when a quantum state $\ket{\psi}$ is efficiently preparable --- that is, when there exists a polynomial-size quantum circuit $U$ such that $\ket{\psi} = U\ket{0^n}$ --- this does not imply that its amplitudes are classically tractable.
In fact, computing the amplitude $\mel{x}{U}{0^n}$ exactly is known to be \clw{GapP}{hard}~\cite{fortnow1999complexity, fenner1999determining}, and approximating the probability $\abs{\mel{x}{U}{0^n}}^2$ to relative error is \clw{\#P}{hard}.
Thus, the ability to classically query the amplitudes of a quantum state is highly non-trivial and bypasses complexity results that hold even for efficiently preparable states.

Key examples of succinct states include: product states, semi-classical subset states~\cite{cade2023improved}, weight-$k$ states~\cite{bremner2025parameterized} and tensor network representations, particularly matrix product states (MPS).
Among these, MPS are well-recognised and well-studied in the fields of complexity theory, many-body physics and quantum chemistry, making them a strong candidate for ideal succinct states.
It is well-known that MPS are described using a set of tensors, expressed as
\begin{equation*}
    \ket{\Psi} \coloneqq \sum_{\underline{\sigma} \in \Omega} {\rm Tr}\big[\prod_{v} A_v^{({\underline{\sigma}_v})}\big] \ket{\underline{\sigma}},
\end{equation*}
where $\underline{\sigma} $ represents a configuration of the $d$-dimensional system, and $A_v^{({\underline{\sigma}_v})}$ are tensors of size $\chi \times \chi$. 
The parameter $\chi$, known as the \emph{bond dimension}, quantifies the entanglement in the state. 
When $\chi$ scales polynomially with the system size, the state can be efficiently described, classically, via its tensors --- the classical space complexity is $O(n \chi^2 d)$.

Tensor network states, especially MPS and PEPS, are particularly useful for representing ground states of certain local Hamiltonians that obey area laws.
Systems with area laws have ground states that exhibit an entanglement entropy which scales with the boundary area of a bipartition on the system.
In one dimension, this connection is well-established: ground states of gapped Hamiltonians can be efficiently represented as MPS~\cite{hastings2004locality,verstraete2006matrix}.
In higher dimensions, the situation is more complicated. 
While frustration-free gapped Hamiltonians have been shown to obey area laws~\cite{anshu2022area}, there are counterexamples~\cite{ge2015area}, and whether area laws hold generally remains an open question~\cite{eisert2010colloquium,huang2020local}.
This poses challenges for applying tensor networks broadly in higher dimensions.

It is important to note that many tensor network methods assume approximate representations of ground states.
Our focus here is different: we define succinct states as a broader class that aims to \emph{exactly} encode ground states.
This exact correspondence offers a stricter framework and distinguishes our approach from methods relying on approximate descriptions. 
While MPS and PEPS are dense in ${\rm SU}(2)$, and thus provide natural candidates for approximate succinct representations, their utility in exactly encoding states remains an open question.

In \cref{sec:conclusion}, we conjecture how approximate succinct representations extend the range of Hamiltonians whose ground states can be efficiently described, potentially generalising this problem.
We now introduce formal definitions of succinct states.
Succinct states naturally arise in various forms, due to the fact quantum states have complex amplitudes.
Here, we focus on states that admit an \emph{exact representation} within a fixed number of bits, and provide a rigorous framework for describing such states.
As detailed in \cref{sec:preliminaries}, there are four main families of algorithms to encode numerical values exactly in a polynomial number of bits --- these will define the succinct states we consider.
We formally define one representative family, with analogous definitions applying to the others.

\begin{definition}[$\mathbb{C}_{p(n)}$-succinct state]\label{def:succinct-state}
    A normalised $n$-qubit state $\ket{\psi} = \sum_{j \in \B^n} \alpha(j) \ket{j}$, where $\alpha(j) \in \mathbb{C}$, is a $\mathbb{C}_{p(n)}$-succinct state if there exists an efficient classical (query) algorithm $\mathcal{Q}_\psi$ that, given an $n$-bit string $x$, outputs the exact binary representation of
    \begin{equation*}
        \mathcal{Q}_\psi(x) = c_\psi \cdot \alpha(x),
    \end{equation*}
    for some constant $0 < c_\psi \leq 2^{p(n)}$.
\end{definition}

The definition implies the value $c_\psi \cdot \alpha(x) = a + {\rm i}b$ is represented specifically in the form shown in \cref{eq:z-bin-breakdown}.
This heavily restricts the types of states that fall within this definition, not to mention the requirement for the efficient classical (query) algorithm.
It is clear from the definition that the classical algorithm can provide the amplitude of an \emph{un-normalised} version of the state $\ket{\psi}$.
In the case where $c_\psi = 1$, the algorithm outputs the exact amplitude of the state.

Another important aspect to note is that the (scaled) amplitudes are algebraically encoded; the output of the classical algorithm does not approximate the amplitude but rather provides a numerator-denominator pair.
Basic number theory shows that with this representation, the amplitudes cannot be arbitrary irrational numbers (even with the scaling factor).
This is a crucial point we will discuss further.

\begin{figure}[!ht]
    \centering
    \begin{tikzpicture}
        \pic[scale=0.7]{hierarchy};
    \end{tikzpicture}
    \caption{A hierarchy of succinct states. Not to scale.}
    \label{fig:succinct-states}
\end{figure}

As we discuss them later, we introduce the following definition.

\begin{definition}[$\mathbb{C}_{p(n)}\dbbrckt{\omega}$-succinct states]
    A normalised state $\ket{\psi} = \sum_{j\in\B^n} \alpha(j) \ket{j}$, where $\alpha(j) \in \mathbb{C}$, is a $\mathbb{C}_{p(n)}\dbbrckt{\omega}$-succinct state if there exists an efficient classical (query) algorithm $\mathcal{Q}_\psi$ that, given an $n$-bit string $x$, outputs the exact binary representation of
    \begin{equation*}
        \mathcal{Q}_\psi(x) = c_\psi \cdot \omega^s \cdot \alpha(x),
    \end{equation*}
    for some constant $0< c_\psi \leq 2^{p(n)}$ and $s\in\{0,\dots,7\}$.
\end{definition}

Recall the set $\mathbb{C}_{p(n)}\dbbrckt{\omega}$ algebraically encodes the integer $s$ in the first three bits of the output string.
This definition is a generalisation of the previous one and too admits extensions to other analogously defined sets.
The following remarks discuss how roots of unity can be encoded and more generally, how algebraic numbers can be approximated.

\begin{remark}
    Cyclotomic polynomials $\Phi_n(x)$ are polynomials whose roots are the primitive $n$-th roots of unity.
    These polynomials can be represented as binary strings.
    For instance, the $8$-th root of unity, $\omega$, has the associated cyclotomic polynomial $\Phi_8(x) = x^4 + 1$.
    In binary form, this polynomial can be expressed as:
    \begin{equation*}
        \text{bin}(\Phi_8(x)) = (01\;00\;00\;01),
    \end{equation*}
    where each pair of bits represents the sign and value of the polynomial's coefficients.
    Alternatively, a $3$-bit register can be used to represent the integer $s$, with the convention that the three most significant bits correspond to the power of $\omega$.
    Either encoding method is acceptable; we employ the latter for simplicity.
\end{remark}

\begin{remark}
    A result of~\citet{kuperberg2015hard} states that all algebraic numbers admit an $\epsilon$-multiplicative approximation in time $\poly{n,\ln(1/\epsilon)}$.
    Moreover, every algebraic number has a \emph{fully polynomial-time exponential-approximation scheme} (\cl{FPTEAS}).
    All values considered in this work can be expressed as roots of polynomials with integer coefficients~\cite{milne2003fields}.
    However, for simplicity in explanation, we focus on our proposed encodings.
\end{remark}

\cref{fig:succinct-states} shows a hierarchy of succinct states. 

A useful property we can derive from the definition of succinct states is the ability to calculate the ratio of two amplitudes.
This is a key property that is used in the proof of the containment in \cl{MA}~\cite{jiang2025local}.

\begin{restatable}[]{proposition}{cprop}
    For a succinct state $\ket{\psi}$ with an exact (scaled) amplitude representation in $p(n)$ bits, given a tuple of two $n$-bit strings $(x,y)$, using two calls to the query algorithm $\mathcal{Q}_\psi$, we can obtain the exact binary representation of the amplitude ratio
    \begin{equation*}
        \mathcal{Q}'_\psi(x,y) = \frac{\alpha(x)}{\alpha(y)},
    \end{equation*}
    in $O(p(n))$ bits (for the appropriate set), provided $\alpha(y)\neq 0$.
\end{restatable}

Given the formal definitions above, we can consider what happens when we combine certain types of succinct states.
For example: ``Is the tensor product of two succinct states also succinct?'', ``What are some natural examples of succinct states?''; we will consider these questions and more in the following sections.
The main idea going forward is to consider --- \emph{given a succinct state, with the associated classical query algorithm, if we apply a gate and/or combine succinct states in some manner, can we still efficiently compute the amplitude of the resulting state (using the original query algorithm(s))?}.
Answering this question will be crucial for the \clw{MA}{hardness} proof of the problem.

\subsection{Properties of Subset States}\label{sec:subset-states}
The most natural succinct state we might consider is the subset state.

\begin{definition}[Subset state]\label{def:subset state}
    For any subset $S \subseteq \B^n$, the subset state on $S$ is defined as
    \begin{equation*}
        \ket{S} = \frac{1}{\sqrt{|S|}}\sum_{s\in S} \ket{s}.
    \end{equation*}
\end{definition}

Two easy propositions that follow are:

\begin{proposition}[$\ket{0}$-padding]\label{prop:0-padding}
    If $\ket{S}$ is a subset state on $S \subseteq \B^n$ then
    \begin{equation*}
        \ket{S} \left(\bigotimes_{j=1}^{m} \ket{0}\right)
    \end{equation*}
    is also a subset state on $S \times \{0\}^m \subset \B^{n}\times \B^{m}$.
\end{proposition}

\begin{proposition}[Subset state tensor product]\label{prop:subset-tensor}
    If $\ket{S}$ is a subset state on $S\subseteq \B^n$ and $\ket{T}$ is a subset state on $T\subseteq \B^m$ then
    \begin{equation*}
        \ket{S} \otimes \ket{T} = \frac{1}{\sqrt{|S||T|}}\sum_{s \in S, t \in T} \ket{s}\ket{t} = \frac{1}{\sqrt{|S||T|}}\sum_{r \in S \times T} \ket{r},
    \end{equation*}
    is a subset state on $S \times T \subseteq \B^{n}\times \B^{m}$.
\end{proposition}

\begin{restatable}[]{lemma}{lmasubsetsuccinct}
    \label{lma:subset-succinct}
    The subset state $\ket{S}$ is an $\mathbb{N}_1$-succinct state.
\end{restatable}

It is clear that the classical query algorithm for subset states is essentially a membership oracle.
This is because unless $c_S = 1$, the algorithm does not output further useful information.

\begin{remark}
    For a subset $S \subseteq \B^n$, such that $\abs{S}$ is a square number or an integer power of $2$, the subset state $\ket{S}$ is a $\mathbb{Q}^+_{\log_2{\abs{S}}}$-succinct state.
\end{remark}

Motivated by the idea of exactly representing the amplitude of a subset state, we can consider the scenario where the size of the subset is not a power of $2$ or a square number.

\armk

This is quite a powerful type of state since we can now directly output the value of the uniform amplitude.
The query algorithm can still call from the uniform distribution of the support set $S$, this time with the algebraic encoding of the square root of the size of the set.
Since all the amplitudes are the same, the first bit of the output string will always be $1$.

When adding specific algebraic encodings to the output of classical query algorithms, we open the door for more complicated states.
For example, if it is possible for the classical algorithm to output an algebraic encoding of the square root of a particular rational, then it wouldn't be too unjust to carry this idea forward to other states.
This logic plays equally with permitting the algebraic encoding of the sign bit.
We can then, for example, consider the set 
\begin{equation*}
    \mathbb{Q}_p\dbbrckt{\sqrt{\cdot}} = (\mathbb{A}_1^{({\rm sgn})} \times \mathbb{A}_1^{(\sqrt{\cdot})} \times \mathbb{N}_p)^2,
\end{equation*}
and the succinct state definitions that follow.
We will see a simple example of the utility of the square root indicator in the next section.

\subsection{Operations with Subset States}\label{sec:operations-subset-states}
We now consider a range of operations that can be performed with subset states.
The first operation we consider is the tensor product of two subset states.

\begin{restatable}[]{lemma}{tensorproductsubsetstates}
    \label{lma:tensor-product-subset state}
    The tensor product of two subset states $\ket{S}$ and $\ket{T}$, on $S \subseteq \B^n$ and $T\subseteq \B^m$ respectively, is an $\mathbb{N}_1$-succinct state.
\end{restatable}

This is not immediately obvious since it is not clear how the individual query algorithms can be combined to output an $\mathbb{N}_1$ value.
We now consider the action of a reversible classical gate on subset states.

\clma

Inspired by \cref{lma:classical-gates-subset state}, the actions of more general reversible circuits can be considered.
For example, the action of a Hadamard gate on a computational basis state is 
\begin{equation*}
    \Gate{Had}_q \ket{x} = \frac{1}{\sqrt{2}}\big(\ket{y} + (-1)^{x[q]}\ket{\bar{y}} \big),
\end{equation*}
where $y[j] = \bar{y}[j] = x[j]$ for any $j \neq q$ and then, $y[q] = 0$, $\bar{y}[q] = 1$.
Clearly, the effect on the amplitude after the application of a single Hadamard results in the computation of two subsequent amplitudes.
The addition can be efficiently computed using the appropriate calls to the query algorithm.
However, note there are two `problems': \begin{inparaenum}[(a)]
    \item there is a factor of $1/\sqrt{2}$ in the amplitude, and
    \item $k$ Hadamard gates requires $O(2^k)$ calls to the query algorithm.
\end{inparaenum}
To address the first problem, we can consider the algebraic encoding of $1/\sqrt{2}$ in the output string.
The second problem can be addressed by only allowing a constant number of Hadamard gates.
To build up more general ideas we start with the following lemma.

\begin{restatable}[]{lemma}{hadamardgatesubsetstate}
    \label{lma:hadamard-gate-subset state}
    Consider a subset state $\ket{S}$ on $S \subseteq \B^n$.
    Let $\ket{C_q} = \Gate{Had}_q \ket{S}$ for some $q \in [n]$.
    Then $\ket{C_q}$ is an $\mathbb{N}_2\dbbrckt{{1/\sqrt{2}}_1}$-succinct state.
\end{restatable}

\begin{remark}
    In the setting we have laid out, we cannot simply say that $c_{C_q} = c_S/\sqrt{2}$, since the purpose is to show that it is possible to query the amplitudes of the resultant state using the query algorithm for $\ket{S}$.
    The outcome is that we require an extra component to track powers of $1/\sqrt{2}$.
    However, this is just for this specific example --- looking at the $\ket{C_q}$ in isolation can produce different outcomes, depending on what should be shown.
    In a more general scenario, it may be possible for the pre-factor to `cancel out' out irrational values; yet this may not always be possible.
    Furthermore, binary encodings with a flag bit for irrational numbers can form a space, at least, twice as large.
\end{remark}

The action of a constant number of Hadamard gates on a subset state follows straightforwardly from \cref{lma:hadamard-gate-subset state}.
Let $\bs{q} = (q_1,\dots,q_k)$ be a tuple of $k$ integers such that $q_i \in [n]$ and $q_i \neq q_j$ for $i \neq j$.
We denote the length of the tuple as $\abs{\bs{q}} = k$.

\begin{corollary}\label{cor:hadamard-seq-subset state}
    Consider a subset state $\ket{S}$ on $S \subseteq \B^n$.
    Let $\ket{C_{\bs{q}}} = \prod_{q \in \bs{q}}\Gate{Had}_q \ket{S}$ for a tuple $\bs{q}$, such that ${\abs{\bs{q}} = O(1)}$.
    Then $\ket{C_{\bs{q}}}$ is an $\mathbb{N}_{p}\dbbrckt{\frac{1}{\sqrt{2}}_p}$-succinct state, where $p=\lceil \log_2(\abs{\bs{q}}) \rceil = O(1)$.
\end{corollary}

A classically reversible circuit intertwined with a constant number of Hadamard gates can also be studied.
Using \cref{lma:classical-gates-subset state} and \cref{cor:hadamard-seq-subset state} we arrive at the following corollary.

\begin{corollary}\label{cor:CRG-hadamard-subset state}
    Let $\bsr$ represent the bit string of $O(\poly{n})$ size representing the information for a set of $O(\poly{n})$ classically reversible gates and $O(1)$ Hadamard gates.
    Let $K$ denote the total number of gates and hence we have the sequence set $\{U_k\}_{k\in K}$.
    Define a state ${\ket{D_k} \coloneqq U_k \cdots U_1\ket{S}}$ for some $k \in [K]$.
    Then $\ket{D_k}$ is an $\mathbb{N}_{p}\dbbrckt{\frac{1}{\sqrt{2}}_p}$-succinct state, where $p=O(1)$.
\end{corollary}

Similar ideas to those of the Hadamard gate can be applied to the $T$ gate.
Thankfully the action of the $T$ gate is easy to characterise.

\begin{restatable}[]{lemma}{Tgatesubsetstate}
    \label{lma:T-gate-subset state}
    Consider a subset state $\ket{S}$ on $S \subseteq \B^n$.
    Let $\ket{E_q} = T_q \ket{S}$ for some $q \in [n]$.
    Then $\ket{E_q}$ is an $\mathbb{N}_{1}\dbbrckt{{\omega}_3}$-succinct state.
\end{restatable}

\begin{remark}\label{rmk:T-gate-succinct-equivalency}
    There exists an efficient classical algorithm that can map an element of $\mathbb{N}_{1}\dbbrckt{{\omega}_3}$ to an element of $\mathbb{C}_{1}\dbbrckt{\frac{1}{\sqrt{2}}_1}$.
    This is due to the fact that $\omega^s = (\frac{1}{\sqrt{2}} + {\rm i} \frac{1}{\sqrt{2}})^s$.
    Due to the cyclic nature of the powers of $\omega$, there are 8 distinct values that can be encoded in the output string.
    Specifically, 
    \begin{align*}
        s&=0 \mapsto 1, & s&=1 \mapsto\frac{1}{\sqrt{2}}(1 + {\rm i}), \\
        s&=2 \mapsto {\rm i}, & s&=3 \mapsto\frac{1}{\sqrt{2}}(-1 + {\rm i}), \\
        s&=4 \mapsto -1, & s&=5 \mapsto\frac{1}{\sqrt{2}}(-1 - {\rm i}), \\
        s&=6 \mapsto -{\rm i}, & s&=7 \mapsto\frac{1}{\sqrt{2}}(1 - {\rm i}).
    \end{align*}
    Hence, $\ket{E_q}$ is also a $\mathbb{C}_{1}\dbbrckt{\frac{1}{\sqrt{2}}_1}$-succinct state.
    Arguably, $\mathbb{C}_1$ is ``too much'' for the problem at hand, however, it is convenient to consider the result this way.
\end{remark}

At this point we note that even though we stated that subset states are $\mathbb{N}_1$-succinct states, considering simple variations has led to algebraic encodings nonetheless.
Specifically, we are now requiring the tracking of $1/\sqrt{2}$.
The natural argument to give is then: if we can track $1/\sqrt{2}$, then why not track the square root of any rational? This ultimately results in the subset state modifications having ``simpler'' representations.
For example, we expressed that $\mathbb{A}_1^{(\sqrt{\cdot})} \times \mathbb{Q}^+_{\log_2{\abs{S}}}$ was sufficient for subset states (see \cref{rmk:subset-more-general}).
Using the modifications above we can see (ignoring subscripts for now) that 
\begin{equation*}
    \mathbb{A}^{(1/\sqrt{2})} \times \mathbb{N} \subset (\mathbb{A}^{(\sqrt{\cdot})} \times \mathbb{Q}^+)^2,
\end{equation*}
i.e., the right-hand set can exactly express the amplitude of subset states and those subset states acted only by gate sequences discussed above.
What we mean here is that if we readily assume the subset state classical algorithm can keep track of a square root, then it is not unreasonable to assume the classical algorithm for the states discussed in the lemmas and corollaries above can also track subsequent square roots.
In fact, we have explicitly shown, making use of the circuit descriptor $\bsr$, that efficient classical algorithms exist to do this.
We therefore proceed under the influence of \cref{rmk:T-gate-succinct-equivalency}.

\begin{corollary}\label{cor:T-dagger-gate-subset state}
    Consider a subset state $\ket{S}$ on $S \subseteq \B^n$.
    Let $\ket{\bar{E}_q} = T^\dagger_q \ket{S}$ for some $q \in [n]$.
    Then $\ket{\bar{E}_q}$ is a $\mathbb{C}_{1}\dbbrckt{\frac{1}{\sqrt{2}}_1}$-succinct state.
\end{corollary}

\begin{restatable}[]{lemma}{Tgateseqsubsetstate}
    \label{lma:T-gate-seq-subset state}
    Consider a subset state $\ket{S}$ on $S \subseteq \B^n$.
    Let $\ket{F_{\bs{q}}} = \prod_{q \in \bs{q}}T_q \ket{S}$ for a tuple $\bs{q}$, such that $\abs{\bs{q}} \leq n$.
    Then $\ket{F_{\bs{q}}}$ is a $\mathbb{C}_{1}\dbbrckt{\frac{1}{\sqrt{2}}_1}$-succinct state.
\end{restatable}

It is straightforward to see how the action of $T^\dagger$ gates can be considered.

\begin{corollary}\label{cor:T-dagger-gate-seq-subset state}
    Consider a subset state $\ket{S}$ on $S \subseteq \B^n$.
    Let $\ket{\bar{F}_{\bs{q}}} = \prod_{q \in \bs{q}}T^\dagger_q \ket{S}$ for a tuple $\bs{q}$, such that $\abs{\bs{q}} \leq n$.
    Then $\ket{\bar{F}_{\bs{q}}}$ is a $\mathbb{C}_{1}\dbbrckt{\frac{1}{\sqrt{2}}_1}$-succinct state.
\end{corollary}

\begin{corollary}\label{cor:T-T-dagger-gate-seq-subset state}
    Consider a subset state $\ket{S}$ on $S \subseteq \B^n$.
    Let $\ket{G_{\bs{q}}} = \prod_{q \in \bs{q}}V_q \ket{S}$ for a tuple $\bs{q}$, such that $\abs{\bs{q}} \leq n$ and where $V \in \{T,T^\dagger\}$.
    Then $\ket{G_{\bs{q}}}$ is a $\mathbb{C}_{1}\dbbrckt{\frac{1}{\sqrt{2}}_1}$-succinct state.
\end{corollary}

As a result of \cref{cor:hadamard-seq-subset state,cor:CRG-hadamard-subset state,cor:T-T-dagger-gate-seq-subset state} we obtain the following result.

\dlma

\begin{remark}
    There exists an efficient classical algorithm that can map an element of $\mathbb{A}^{(1/\sqrt{2})} \times \mathbb{C}$ to an element of $\mathbb{A}^{(\sqrt{\cdot})} \times \mathbb{C}$
\end{remark}

\begin{remark}
    We note that the results in this section are similar to those considered by Van der Nest~\cite{vandennest2011simulating} concerning the estimation of amplitudes of quantum states evolved by sparse circuits.
    However, our focus is on the exact representations of the amplitudes of the states considered.
    Specifically, our aim is to characterise and define the type of succinct states that arise from the action of specific gate sets on subset states.
    Our results are complementary to those of Ref.~\cite{vandennest2011simulating}.
\end{remark}

\subsection{Operations with Hybrid Subset States}\label{sec:operations-hybrid-subset-states}
To conclude the analysis on subset states we present two lemmas that capture specific dynamics of reversible circuits acting on subset states.
Essentially, the superposition states we consider track the action of $K$ reversible gates that from a given gate sequence $\mathcal{V}$.
The first lemma uses the states $\ket{B_k}$ and the second lemma uses the states $\ket{H_k}$.
A preliminary idea required for the lemmas is an attribute that occurs from the action of classical gates on subsets.
Since each classical gate is a bijective map on the computational basis states, the cardinality of the set is invariant under the action of such gates.
Up to an overall phase, the $T$ and $T^\dagger$ gates are also bijective maps on the computational basis states.
We formalise this idea in the following remark.

\begin{remark}\label{rmk:classically-reversible-gates-size-invariance}
    A classically reversible gate $R$ is a bijective map on $n$-bit strings.
    Let $\mathcal{V} = \{R_j\}_{j \in [m]}$ denote a set of classically reversible gates.
    Given a subset $S \subseteq \B^n$, define the action $R \circ S \coloneqq \{R(s) : s \in S\}$, where $R(s)$ is the action of the gate $R$ on the bit string $s$.
    Then,  
    \begin{equation*}
        \abs{\prod_{j = m}^{1} R_{j} \circ S} = \abs{S},
    \end{equation*}
    i.e., the cardinality of the set $S$ is invariant under the action of the classically reversible gates.
\end{remark}

\begin{restatable}[]{lemma}{historystateexone}
    \label{lma:history-state-ex-1}
    The superposition state
    \begin{equation}\label{eq:history-state-ex-1}
        \ket{\eta} = \frac{1}{\sqrt{\abs{K}}} \sum_{k=1}^{K} \ket{B_k}\ket{k},
    \end{equation}
    is an $\mathbb{N}_{1}$-succinct state with the efficient classical (query) algorithm $\mathcal{Q}_\eta$.
\end{restatable}

\begin{remark}
    Recall from \cref{rmk:subset-more-general} that we may consider the state in \cref{eq:history-state-ex-1} as an $\mathbb{A}^{(\sqrt{\cdot})}_1 \times \mathbb{Q}^+_{\log_2{(\abs{S}\abs{K})}}$-succinct state.
    This requires a more powerful classical algorithm to output the amplitude of the state.
\end{remark}

Note that by \cref{prop:subset-tensor} and \cref{rmk:classically-reversible-gates-size-invariance} we can see $\ket{\eta}$ as in \cref{eq:history-state-ex-1} is a superposition of distinct subset states.
To see this, note that for any $k \in [K]$: 
\begin{equation*}
    B_k \eqqcolon \prod_{j=k}^{1} R_j \circ S ~\subseteq \B^n, ~~\text{with}~~ \abs{B_k} = \abs{S},
\end{equation*}
and $k \in \B^{\abs{{\rm bin}(K)}}$ is an indicator bit string.
Clearly, the subsets formed by $B_k \times \{k\}$ are orthogonal to each other.
For the purposes of bookkeeping, we can consider the following definitions.

\begin{definition}[Hybrid Subset State (HSS)]
    For any subset $S \subseteq \B^n$ and a set of consecutive integers $I = \{1,\dots,k\}$, define the hybrid subset state on $(S,I)$ as
    \begin{equation*}
        \ket{\mathcal{M}_{S,I}} = \frac{1}{\sqrt{k}}\sum_{j=1}^{k} \ket{S_j},
    \end{equation*}
    where $S_k = S \times \{{\rm bin}(k)\}$ for each $k \in I$.
\end{definition}

\begin{definition}[Classically Encoded Hybrid Subset State (CEHSS)]
    For any subset $S \subseteq \B^n$, set of consecutive integers $I = \{1,\dots,k\}$ and set of classically reversible gates $\mathcal{V} = \{R_k\}_{k\in I}$, define the classically encoded hybrid subset state on $(S,I,\mathcal{V})$ as
    \begin{equation*}
        \ket{\mathcal{C}_{S,I,\mathcal{V}}} = \frac{1}{\sqrt{k}}\sum_{j=1}^{k} \ket{\hat{S}_j},
    \end{equation*}
    where
    \begin{equation*}
        \hat{S}_k = \big(\prod_{j=k}^{1} R_j \circ S \big) \times \{{\rm bin}(k)\},
    \end{equation*}
    for each $k \in I$.
\end{definition}

Clearly then the states of \cref{eq:history-state-ex-1,eq:history-state-ex-2} are forms of encoded hybrid subset states.
The former being a classically encoded hybrid subset state and the latter being slightly more general, due to the presence of the Hadamard gates.
It is not hard to see that a Hadamard gate will not preserve the size of the support set of a subset state.
The main reason for introducing the above definitions is due to \cref{eq:history-state-ex-1}.

\begin{lemma}[CEHSS equivalencey]\label{lma:CEHSS-equivalence}
    A CEHSS on $(S,I,\mathcal{V})$ is equivalent to a subset state over $\mathcal{S} \subseteq \B^{(n+\abs{{\rm bin}(k)})}$.
\end{lemma}

\begin{proof}
    Classical gates preserve the size of $S$.
    It is clear that each of the subsets $\hat{S}_k$ are orthogonal to each other due to product with the unique strings ${\rm bin}(k)$.
    Thus let $\mathcal{S} = \bigcup_{k \in I} \hat{S}_k$, then $\abs{\mathcal{S}} = k\abs{S}$.
    The resulting state is then a subset state on $\mathcal{S}$.
\end{proof}

As we move onto the more general form of encoded hybrid subset state notice that \cref{eq:history-state-ex-2} includes contributions from the Hadamard and $T$ gates.
We consider the exact representation of this state.
Moreover, we now assume that the amplitude of subset states can be exactly represented as an element of $\mathbb{A}^{(\sqrt{\cdot})}_1 \times \mathbb{Q}^+_{\log_2\abs{S}}$.

\elma

\subsection{Properties of General Succinct States}\label{sec:properties-general-succinct-states}
Now we consider operations on general succinct states.
We first consider the tensor product of two succinct states.
Note that these results can be extended to various other families of succinct states.

\alma

While the proof holds for $\mathbb{S} \in \{\mathbb{N},\mathbb{Q}^+,\mathbb{Q},\mathbb{C}\}$, we demonstrate it for $\mathbb{C}$ as the others follow trivially.

\blma

In summary, the analysis on succinct state properties has revealed that given a set of initial succinct states it is possible to yield subsequent succinct states by considering specific operations.
This is particularly useful when considering the tensor product of states and the action of classically reversible gates.
Our main goal has been to understand how access to initial query algorithms can be used to efficiently calculate the amplitudes of the resulting states.
In the sequel we show how this analysis becomes important for the \sc{Succinct State Local Hamiltonian} problem.
We have shown that it is natural to consider subset states as algebraic encoded succinct states.
Specifically, the query algorithm outputs a binary string using the most significant bit to track the action of a square root operation --- this is analogous to the idea of using bits to track the sign of a number.

\subsection{Multi-Alphabet Query Access}\label{sec:multi-alphabet-query-access}
Given the type of access model we have defined for succinct states, it is natural to consider how powerful additional query algorithms can be.
For example, our restriction is on computational basis state overlap, yet it is possible to consider overlap with a general product states.
Though, \cref{cor:hadamard-seq-subset state} suggests such access can require exponentially many additional computational steps.
To make this more concrete, we define \emph{multi-alphabet states}.\footnote{For our purposes, a multi-alphabet state is an alternative name for a general product state. The specification of the name is to highlight the required encoding of the state. A more general definition of multi-alphabet states can be given.}

Let $B$ be a single-qubit basis and $\varSigma_B$ be the alphabet of $B$,  i.e., $B = \{\ket{b^0},\ket{b^1}\}$ and $\varSigma_B = \{b^0,b^1\}$.
Consider the set $\Sigma = \bigtimes_{l =1}^{n} \varSigma_{B_l}$ and define $\sigma \in \Sigma$ as a multi-alphabet string, e.g., $\sigma = (b^{x_1}_1, b^{x_2}_2,\dots,b^{x_n}_n)$, where $x_l \in \B$.
We define the product state $\ket{\sigma} = \bigotimes_{l=1}^{n} \ket{b^{x_l}_l}$.
It follows trivially that for each $\ket{b^{x_l}_l}$, there exists a unitary operator $U_l$ such that $U_l \ket{x_l} = \ket{b^{x_l}_l}$.
Furthermore, $U_l$ can be expressed efficiently.
Let us restrict ourselves to the case where $U_l$ can be exactly expressed in a polynomial number of bits.
We define $N_{\sigma}$ as the number of non-zero amplitudes in the product state $\ket{\sigma}$ when expressed in the computational basis.

\begin{lemma}[Multi-Alphabet Query Access]\label{lma:multi-alphabet-query-access}
    Consider an $\mathbb{S}$-succinct state $(\ket{\psi}, \mathbb{S}_{p(n)}, \mathcal{Q}_\psi)$, where $p(n)$ is a polynomial in $n$.
    Let $\sigma \in \Sigma = \bigtimes_{l =1}^{n} \varSigma_{B_l}$ be a multi-alphabet string.
    The cost of computing the amplitude $\braket{\sigma}{\psi}$ requires $N_\sigma$ calls to the query algorithm $\mathcal{Q}_\psi$.
\end{lemma}

\begin{proof}
    It follows that 
    \begin{align*}
        \braket{\sigma}{\psi} &= \bigotimes_{l=1}^{n} \braket{b^{x_l}_l}{\psi} \\
        &= \bigotimes_{l=1}^{n} \sum_{x_l \in \B} \alpha_{x_l} \braket{x_l}{\psi} \\
        &= \sum_{\boldsymbol{x} \in \B^n} \alpha_{\boldsymbol{x}} \braket{\boldsymbol{x}}{\psi},
    \end{align*}
    where $\boldsymbol{x} = (x_1,\dots,x_n)$ and $\alpha_{\boldsymbol{x}} = \prod_{l=1}^{n} \alpha_{x_l}$.
    Let $N_\sigma$ be the number of $\alpha_{\boldsymbol{x}}$ that are non-zero.
    The query algorithm $\mathcal{Q}_\psi$ can be used to compute the individual amplitudes $\braket{\boldsymbol{x}}{\psi}$ and therefore the total amplitude $\braket{\sigma}{\psi}$ requires $N_\sigma$ calls.
    The cost of addition is assumed to be negligible in this analysis.
\end{proof}

For $n$ alphabets that are non-trivial superposition of the computational basis states, the number of non-zero amplitudes can be exponential in $n$.
For similar reasons, it is easy to see that computing the partial trace of $\rho_\psi = \ketbra{\psi}$ is only efficient if the number of qubits in the region being traced out is at most logarithmic in $n$.
A multi-alphabet access model is exponentially more powerful than one limited to computational basis states~\cite{cotler2021revisiting}.

\section{The Succinct State Local Hamiltonian Problem}\label{sec:LHP-succinct-ground-states}
The standard \sc{Local Hamiltonian} problem~\cite{kitaev2002classical} admits only a classical description of the Hamiltonian and thresholds, with no auxiliary information relating to the ground state. 
Several problem extensions have explored what happens when such information is added.

\citet{bravyi2015monte} introduced the \sc{Guided Local Stoquastic Hamiltonian} problem (also called \sc{Guidable Local Hamiltonian}~\cite{weggemans2024guidable}), where a guiding state is promised to have non-negligible pointwise overlap with the ground state, though its explicit description is not given as input.
It was shown this variant is \clw{MA}{complete} for $6$-local stoquastic Hamiltonians for inverse-polynomial precision.
Thus, for stoquastic Hamiltonians, this suggests modifying the problem from a guiding state with point-wise overlap to a ground state with a succinct representation does not yield any computational advantage in estimating the ground-state energy to inverse-polynomial precision.

Explicit inclusion of the guiding state as part of the input has been shown to give \clw{BQP}{completeness} in estimating the ground-state energy for local Hamiltonians~\cite{richter2007two}; \clw{BQP}{hardness} is maintained even for a range of so-called ``physically relevant'' local Hamiltonian families~\cite{cade2023improved}.
It is clear that making the ground-state information explicit can fundamentally alter the complexity.

We now introduce the main computational problem of this work.

\begin{definition}[\sc{$\mathbb{S}_{p(n)}$-Succinct State $k$-Local Hamiltonian}]\label{def:LHP-succinct-state}
Let $H = \sum_{i=1}^{m} H_i$ be a $k$-local Hamiltonian on $n$ qubits with $m = O(\poly{n})$.
Given parameters $0 \leq a < b \leq 1$ such that $b - a = 1/\poly{n}$, the promise problem is:
\begin{itemize}
    \item[]\hspace{-1.2em} (\textsc{yes}): There exists a $\mathbb{S}_{p(n)}$-succinct state $\ket{\xi}$ such that $\bra{\xi}H\ket{\xi} \le a$.
    \item[]\hspace{-1.2em} (\textsc{no}): For all states $\ket{\psi}$, $\bra{\psi}H\ket{\psi} \ge b$.
\end{itemize}
\end{definition}

Here the type of succinct state $\mathbb{S}$ is left unspecified.
Conceptually, this problem parallels the guided variants, but assumes that the ground state itself is succinctly describable rather than merely approximable by a guiding state for example.
This change in representation exposes an important conceptual question: how much classical structure in the ground state is sufficient to collapse quantum verification?

Theoretically, this formulation highlights a transition point where quantum verification (as in the standard \sc{Local Hamiltonian} problem) can be replaced by pure classical (probabilistic) verification.
While one might expect containment in \cl{QCMA}, we show that the problem in fact admits fully classical verification via standard mappings to stoquastic Hamiltonians.
This reveals that succinct representability can itself serve as a form of verification power, rather than merely an encoding convenience.

Although our setting does not originate from a single physical scenario, this is typical for models that explore the limits of classical assistance.
Even guided state heuristics lack a general method for identifying useful trial states, despite there being numerous families with efficient classical descriptions~\cite{froese1987general,white1992density,vidal2003efficient,verstraete2008matrix,carleo2017solving}.\footnote{See~\cite{jiang2025local} for further discussion.}
Though, exact amplitude access also appears in practically motivated work, such as a recent quantum state certification protocol~\cite{huang2025certifying}, where the existence of the protocol relies on obtaining such information.
In this respect, our access model is no more idealised than those already employed in application-oriented settings.
The study of succinct states therefore remains valuable: it identifies when algebraic descriptions enable efficient classical verification and forms a basis for future approximate-succinct or guided state frameworks.

In the sections that follow we discuss the containment in \cl{MA} and the \clw{MA}{hardness} of the problem.
Using the analysis of the previous section, proving the main ideas is straightforward.

\subsection{Class Containment}\label{sec:containment}
To prove the problem can be decided in \cl{MA} requires the following assumptions on the problem instance:
\begin{assumption}
    Given \cref{def:LHP-succinct-state}, we assume:
    \begin{enumerate}
        \item[(i)] We have query access to the Hamiltonian $H$.
        \item[(ii)] For any $x,y \in \B^n$, the output $\bra{x}H\ket{y} \in \mathbb{C}_{p(n)}$, for some polynomial $p(n)$.
        \item[(iii)] All of $a(n)$, $b(n)$ and $m$ can be represented in $p(n)$ bits.
        \item[(iv)] The succinct ground state $\ket{\xi}$ is an $\mathbb{S}_{p(n)}$-succinct state (i.e., its (scaled) amplitudes are algebraically encoded as an element of $\mathbb{S}_{p(n)}$).
    \end{enumerate}
\end{assumption}

\begin{remark}
    Query access to the Hamiltonian implies that for any $x,y \in \B^n$, there exists an efficient classical (query) algorithm $\mathcal{Q}^{(1)}_H$ such that $\mathcal{Q}^{(1)}_H(x,y) = \bra{x}H\ket{y}$.
    Additionally, for any given row index $x\in \B^n$, there exists an efficient classical (query) algorithm $\mathcal{Q}^{(2)}_H$ such that $\mathcal{Q}^{(2)}_H(x) = \{y : \bra{x}H\ket{y} \neq 0\}$.
    Under the assumptions above, it is clear that $\mathcal{Q}^{(1)}_H(x,y) \in \mathbb{C}_{p(n)}$ and $y \in \mathbb{N}_{n}$ for any $y \in \mathcal{Q}^{(2)}_H(x)$.
\end{remark}

We note that $p(n)$ is an upper bound on the number of bits required to encode the amplitudes of the ground state; it may be the case that there exists a polynomial $q(n) < p(n)$ such that the amplitudes are $\mathbb{S}_{q(n)}$-succinct, yet since $\mathbb{S}_{q(n)} \subset \mathbb{S}_{p(n)}$, we can consider the more general case.

The bulk of the proof of containment in \cl{MA} follows results of \citet{bravyi2023rapidly} with minor modifications to the generator used in the Markov chain.
We therefore only outline the main ideas here; see Refs.~\cite{bravyi2023rapidly,jiang2025local} for further details.

\subparagraph{Arthur's Algorithm.}
The main idea of the verification algorithm used to place the problem in \cl{MA} is to map the given instance $(n,H,a,b)$ to a \emph{stoquastic} Hamiltonian and apply Gillespie's algorithm~\cite{gillespie1977exact} to distinguish between the \sc{yes} and \sc{no} cases, using a message from Merlin containing $(\lambda^\star, \mathcal{Q}_\xi,x^\star)$.
The element $\lambda^\star$ can be taken to be the ground-state energy of the Hamiltonian $H$, $\mathcal{Q}_\xi$ is the query algorithm for the succinct ground state $\ket{\xi}$ and $x^\star$ is a bit string satisfying important technical conditions (see Ref.~\cite{jiang2025local} for further details).
Clearly, if $\lambda^\star > b$ then Arthur can immediately reject the instance.

The first step Arthur must perform is to map the Hamiltonian $H$ to a real Hamiltonian.
We outline this procedure below.
A subsequent mapping of real Hamiltonians to stoquastic Hamiltonians can be performed using the fixed-node quantum Monte Carlo method~\cite{tenHaaf1995proof, bravyi2023rapidly}; this method can be viewed as a way of ``curing'' the sign problem typically found in quantum Monte Carlo simulations.
It follows that if the fixed-state in this mapping is the ground state of the (real) Hamiltonian, then for the fixed-node Hamiltonian $F$, we have $F\ket{\xi} = H \ket{\xi}$.
If this is not the case, a straightforward proof shows that the ground-state energy of $F$ is bounded below by that of $H$~\cite{bravyi2023rapidly}.

In order to compute this transformation to $F$ efficiently, we implicitly require that the ground state is a succinct state.
Since the fixed-node Hamiltonian is stoquastic, there is an efficient algorithm that maps the problem to a classical (continuous-time) Markov chain.
Moreover, the generator of the Markov chain $G$ is related to the Hamiltonian $F$ via
\begin{equation*}
    \bra{y} G \ket{x} = - \frac{\braket{y}{\xi}}{\braket{x}{\xi}}\bra{y} F \ket{x}.
\end{equation*}
The difference here from previous work~\cite{bravyi2023rapidly} is that the provided state $\ket{\xi}$ can be adversarial.
The stationary distribution of the Markov chain is the probability distribution sampling the ground state $\ket{\xi}$, i.e., $\pi(x) = \abs{\braket{x}{\xi}}^2$.
The use of Gillespie's algorithm then allows Arthur to simulate the Markov chain for a time $t = O(\poly{n})$.

To distinguish between the \sc{yes} and \sc{no} cases, Arthur must check that $G$ defines a \emph{legal} generator; in the \sc{no} case, the generator can have ill-defined parts and therefore rejection occurs when the walk hits these sectors.
It is well-known that Markov chains in these circumstances exhibit mixing times that scale with the inverse of the spectral gap~\cite{bravyi2022simulate,bravyi2023rapidly}, though the requirement of $G$ being a legal generator is independent of a large spectral gap~\cite{jiang2025local}.
Furthermore, it is clear that the computational steps can be performed efficiently given the assumptions and inputs outlined above.

From Arthur's algorithm and previous work on Markov chains for real Hamiltonians~\cite{bravyi2023rapidly}, we can see that the problem is in \cl{MA} if the Hamiltonian is real.

\begin{theorem}[\cite{jiang2025local}]\label{thrm:Jiang-kLRH-Qp-MA}
    The \sc{$\mathbb{Q}_{p(n)}$-Succinct State $k$-Local Real Hamiltonian} problem is in \cl{MA}, for all $k\geq 2$.
\end{theorem}

\begin{remark}
    It is known that polynomially-gapped $k$-local stoquastic Hamiltonians (with polynomially-bounded norm), can be mapped to a Metropolis-Hastings Markov chain with a mixing time scaling as $\tilde{O}(\poly{n})$~\cite{bravyi2022simulate}.
    As we discuss later, the fixed-node stoquastic Hamiltonians need not be local, and can have unbounded norm if no structure is assumed on the size of the ground state's coefficients.
    Non-locality means a single configuration $x$ may have exponentially many neighbors $y$ with non-zero transition amplitude; a discrete-time Markov chain proposal step could then require enumerating or sampling among exponentially many possibilities, which is infeasible.
    An unbounded operator norm implies arbitrarily large transition rates.
    Continuous-time Markov chains can avoid both issues.
\end{remark}

For our purposes, we must ensure the general (complex) to real Hamiltonian mapping results in an appropriate succinct state and is computable in classical polynomial time.
If this is the case, then the subsequent steps of~\refcite{jiang2025local} can be followed to show that the problem is in \cl{MA}.
Below, we demonstrate how a general (complex) Hamiltonian can be transformed into a real Hamiltonian at the cost of increasing the locality by one and doubling the dimension.
We further show how the query algorithm for the succinct state $\ket{\xi}$ can be adapted to this transformation.
Additionally, we show a similar preservation of the properties when the real Hamiltonian is transformed into a stoquastic Hamiltonian.
Our analysis demonstrates this transformation can be carried out in polynomial time via an efficient classical algorithm.

\subsubsection{Complex to Real Hamiltonians}
It is discussed in Ref.~\cite[Remark 1]{jiang2025local} that the amplitudes of states considered for \cref{def:LHP-succinct-state} when $\mathbb{S} = \mathbb{C}$, must be of the form 
\begin{equation*}
    \frac{a}{b} + {\rm i} \frac{c}{d},
\end{equation*}
where $a,b,c,d \in \mathbb{N}_{p(n)}$ such that $\frac{a}{b}, \frac{c}{d} \in \mathbb{Q}_{p(n)}$.
Hence, we immediately have encodings of the signs of each component.
\cref{thrm:Jiang-kLRH-Qp-MA} does not capture general complex Hamiltonians.
In order to demonstrate how this result can be extended for a wider class of Hamiltonians and associated succinct states, we consider the following:

\begin{lemma}[{{\cite[Lemma 1]{bravyi2023rapidly}}}]\label{thrm:real-complex-ham-reduction}
    Given a $k$-local Hamiltonian on $n$-qubits $H \in \mathbb{C}^{2^n \times 2^n}$ with a $\mathbb{C}_{p(n)}$-succinct ground state, there exists a $(k+1)$-local Hamiltonian on $(n+1)$-qubits $\hat{H} \in \mathbb{R}^{2^{n+1} \times 2^{n+1}}$ with a $\mathbb{Q}_{p(n)}$-succinct ground state.
    Let $\sigma(H)$ be non-degenerate, then $\sigma(\hat{H})$ is the $2$-multiset of $\sigma(H)$.
\end{lemma}

\begin{proof}
    We decompose $H$ into its respective real and complex parts, $H = H_{R} + {\rm i}H_{I}$.
    The Hamiltonian $\hat{H}$ is then
    \begin{equation}\label{eq:real-complex-hamiltonian}
        \hat{H} = H_{R} \otimes I + H_{I} \otimes - {\rm i}Y,
    \end{equation}
    where $Y$ is the Pauli-$Y$ operator acting on the $(n+1)$-th qubit.
    Clearly this locality of $\hat{H}$ is $k+1$ and $\hat{H}$ is self-adjoint.
    The energy eigenvectors of $\hat{H}$ are related to those of $H$ via 
    \begin{equation*}
        |\hat{\lambda}_{\pm}\rangle = \ket{\lambda_{R}}\ket{0} \pm \ket{\lambda_{I}}\ket{1},
    \end{equation*}
    where $\ket{\lambda} = \ket{\lambda_{R}} + {\rm i}\ket{\lambda_{I}}$ is an eigenvector of $H$.
    Furthermore, we assume the ground state of $H$, $\ket{\lambda_0}$, is $\mathbb{C}_{p(n)}$-succinct.
    By \cref{lma:complex-succinct-state-to-real-succinct-state} we conclude that $|\hat{\lambda}_{\pm}\rangle$ are $\mathbb{Q}_{p(n)}$-succinct.
    Since $|\hat{\lambda}_{\pm}\rangle$ are orthogonal to each other and for and eigenvector $\ket{\lambda}$ of $H$, it follows that the spectrum of $\hat{H}$ is the $2$-multiset of the spectrum of $H$.
\end{proof}

In fact, further steps can be taken with respect to this Theorem.
Specifically, it is shown in \refcite{bravyi2023rapidly} that if the initial local Hamiltonian has a non-degenerate ground space and spectral gap $\Delta$, then the resulting real Hamiltonian can also be made to have a non-degenerate ground space and spectral gap at least $\min\{1,\Delta\}$.
To make the ground space non-degenerate, an additional real Hermitian operator $V$, relating to the amplitude of the original ground state, can be added to the Hamiltonian (see \refcite{bravyi2023rapidly} for details).
Extending the mapping this way results in the Markov chain obtained from the fixed-node Hamiltonian having a unique stationary distribution; this is required for~\cite[Lemma 2]{bravyi2023rapidly}.
The subsequent analysis that follows \cref{thrm:real-complex-ham-reduction} in \refcite{jiang2025local}, appears to handle the degenerate case.
An immediate corollary of \cref{thrm:Jiang-kLRH-Qp-MA} and \cref{thrm:real-complex-ham-reduction} is then a similar result for the containment in \cl{MA} of (complex) Hamiltonians with complex succinct ground states.

\begin{corollary}[\cite{jiang2025local}]
    The \sc{$\mathbb{C}_{p(n)}$-Succinct State $k$-Local (Complex) Hamiltonian} problem is in \cl{MA}, for all $k\geq 2$.
\end{corollary}

We additionally show how to query the real Hamiltonian given query access to the complex Hamiltonian.
Note that analogous arguments can be given if one considers the extended version of \cref{thrm:real-complex-ham-reduction} ---~\cite[Lemma 1]{bravyi2023rapidly}.

\begin{proposition}\label{prop:real-complex-hamiltonian-query}
    Given query access to the $k$-local Hamiltonian $H \in \mathbb{C}^{2^n \times 2^n}$, then we have query access to the $(k+1)$-local Hamiltonian $\hat{H} \in \mathbb{R}^{2^{n+1} \times 2^{n+1}}$.
\end{proposition}

\begin{proof}
    Note that $\bra{x}H\ket{y}$ outputs some complex number $z = a + {\rm i}b \in \mathbb{C}_{p(n)}$.
    Then consider $\hat{H}$ as in \cref{eq:real-complex-hamiltonian}, and two $(n+1)$-bit strings $x' = \conc{x}{u}$ and $y' = \conc{y}{v}$.
    Query access to $\hat{H}$ is then
    \begin{align*}
        \bra{x'} \hat{H} \ket{y'} &= \bra{x}\bra{u}H_{R}\otimes I\ket{y}\ket{v} - {\rm i}\bra{x}\bra{u}H_{I}\otimes Y\ket{y}\ket{v} \\
        &= \bra{x}H_{R}\ket{y}\delta_{u,v}^{I} - {\rm i}\bra{x}H_{I}\ket{y}\delta_{u,v}^{Y}\\
        &= a \,\delta_{u,v}^{I} - {\rm i}b \,\delta_{u,v}^{Y}.
    \end{align*}
    Here $\delta_{u,v}^{I}$ is $1$ if $u=v$ and $\delta_{u,v}^{Y}$ is $-{\rm i}$ if $u=0$ and $v=1$ and ${\rm i}$ if $u=1$ and $v=0$.
    Therefore, 
    \begin{equation*}
        \bra{x'} \hat{H} \ket{y'}  = \begin{cases} a, &\text{if}~~ u = v = 0, \\
     -b, &\text{if}~~ u = 0, v = 1, \\
     b, &\text{if}~~ u = 1, v = 0, \\
     a, &\text{if}~~ u = v = 1.
        \end{cases}
    \end{equation*}
    Clearly these values lie in $\mathbb{Q}_{p(n)}$ as require simply logic based on a query to $H$.

    We now check that given a row index $x'$, the query algorithm can output the columns of $\hat{H}$ with non-zero entries.
    This of course requires logic and a query to $H$ to determine the non-zero entries.
    This can be seen by considering the real part contributions:
    \begin{equation*}
        \hat{H}_{ij} = \begin{cases}
            i ~\text{even}, j~\text{even}, &~~ H_{\frac{i}{2} \frac{j}{2}}, \\
            i ~\text{even}, j~\text{odd}, &~~ 0\\
            i ~\text{odd}, j~\text{even}, &~~ 0 \\
            i ~\text{odd}, j~\text{odd}, &~~ H_{\frac{i-1}{2} \frac{j-1}{2}}.
        \end{cases}
    \end{equation*}
    Thus, for example, calling the query for row $x'$, we subsequently call the query for row $\frac{x'}{2}$ on $H$.
    This outputs a set $\{y : \langle{\frac{x'}{2}}|H\ket{y} \neq 0\}$, then the set ${\{y' : \bra{x'}\re(\hat{H})\ket{y'} \neq 0\}}$ is given by $\{y' : y' = 2y\}$.
    With a bit more thought, utilising the symmetry of the Hamiltonians, similar arguments can be constructed for the imaginary part of $\hat{H}$.
    The union of these sets gives the non-zero entries of $\hat{H}$. 
\end{proof}

Hamiltonians that are $k$-local are $\Theta(n^k)$-sparse, and therefore the rows of the resulting Hamiltonian $\hat{H} + V$ can be also computed efficiently.

\subsubsection{Real to Stoquastic Hamiltonians}
The mapping of real Hamiltonians to stoquastic Hamiltonians follows the fixed-node quantum Monte Carlo method~\cite{tenHaaf1995proof}.
We must ensure that the stoquastic Hamiltonian also has a succinct ground state and can be queried from efficiently.

\begin{definition}[Fixed-Node Hamiltonian]\label{def:FNH}
    Let $\ket{\psi} \in (\mathbb{R}^2)^{\otimes n}$ be a normalised state and $H$ be a $k$-local real Hamiltonian on $n$-qubits.
    Define the sets 
    \begin{align*}
        \mathcal{P} &\coloneqq \{(x,y) \mid x\neq y ~\text{and}~ \alpha(x)\mel{x}{H}{y} \alpha(y) > 0\}, \\
        \mathcal{N} &\coloneqq \{(x,y) \mid x\neq y ~\text{and}~ \alpha(x)\mel{x}{H}{y} \alpha(y) \leq 0\},
    \end{align*}
    where $\alpha(x) = \braket{x}{\psi}$.
    The fixed-node Hamiltonian $F = F(\psi,H)$ is defined as
    \begin{equation*}
        \mel{x}{F}{y} = \begin{cases}
            \mel{x}{H}{y} &\text{if}~~ (x,y) \in \mathcal{N} \\
            0 &\text{if}~~ (x,y) \in \mathcal{P}
        \end{cases},
    \end{equation*}
    and
    \begin{equation*}
        \mel{x}{F}{x} = \mel{x}{H}{x} + \sum_{(x,y) \in \mathcal{P}} \frac{\alpha(y)}{\alpha(x)} \mel{x}{H}{y}.
    \end{equation*}
\end{definition}

\begin{lemma}\label{lma:real-stoq-ham-reduction}
    Given a $k$-local real Hamiltonian $\hat{H}$ on $n$ qubits with a $\mathbb{Q}_{p(n)}$-succinct ground state $\ket{\hat{\xi}}$, there exists a stoquastic Hamiltonian $F$ on $n$ qubits with a $\mathbb{Q}_{p(n)}$-succinct ground state $\ket{\hat{\xi}}$ such that $\lambda_0(F) = \lambda_0(\hat{H})$.
\end{lemma}

\begin{proof}
    It suffices to prove that the ground state $\ket{\hat{\xi}}$ of a real Hamiltonian $\hat{H}$ is also a ground state of the fixed-node Hamiltonian $F = F(\hat{\xi},\hat{H})$.
    To see this, note that $\alpha(x)\lambda_0(\hat{H}) = \mel{x}{\hat{H}}{\hat{\xi}}$ where $\alpha(x) = \langle{x}\lvert\hat{\xi}\rangle$.
    It follows that,
    \begin{align*}
        \mel{x}{F}{\hat{\xi}} &= \sum_{y \in \B^n} \alpha(y)\,\mel{x}{F}{y} \\
            &= \alpha(x)\mel{x}{F}{x} + \sum_{(x,y) \in \mathcal{N}} \alpha(y)\mel{x}{F}{y} \\
            &\qquad+ \sum_{(x,y) \in \mathcal{P}} \alpha(y)\mel{x}{F}{y} \\
            &= \alpha(x)\big( \mel{x}{\hat{H}}{x} + \sum_{(x,y) \in \mathcal{P}} \frac{\alpha(y)}{\alpha(x)} \mel{x}{H}{y} \big)\\
            &\qquad+ \sum_{y \in \mathcal{N}} \alpha(y)\mel{x}{H}{y} \\
            &= \sum_{y \in \B^n} \alpha(y)\mel{x}{H}{y} \\
            &= \mel{x}{\hat{H}}{\hat{\xi}}.
    \end{align*}
    Therefore, $\mel{x}{F}{\hat{\xi}} = \lambda_0(\hat{H})\alpha(x)$ which implies $\lambda_0(F) = \lambda_0(\hat{H})$.
    Hence, the structure of the ground state is preserved.
    It follows that $F$ is stoquastic after choosing the sign gauge $\ket{x} \mapsto {\rm sgn}(\alpha(x))\ket{x}$.
\end{proof}

Notice that we do not require the stoquastic Hamiltonian to be local.\footnote{We note that even if $F$ was local, this would not imply the complexity collapse $\cl{QMA} = \cl{StoqMA}$ since additional ground state information is required to construct $F$.}
Since the ground state likely lacks strong local structure, terms in the fixed-node Hamiltonian $F$ such as $\alpha(y)/\alpha(x)$ can introduce non-locality.
Querying an element from the Hamiltonian $F$ follows by first checking if the pair $(x,y)$ is in $\mathcal{P}$ or $\mathcal{N}$.
If the latter holds then a query following \cref{prop:real-complex-hamiltonian-query} can be performed.

\begin{lemma}
    Given query access to the $k$-local real Hamiltonian $\hat{H} \in \mathbb{R}^{2^n \times 2^n}$ and query access to its ground state $\ket{\hat{\xi}}$, there exists query access to the fixed-node Hamiltonian $F = F(\hat{\xi},\hat{H}) 
    \in \mathbb{R}^{2^n \times 2^n}$.
\end{lemma}

\begin{proof}
    Note that $\mel{x}{\hat{H}}{y}$ outputs a real number $s \in \mathbb{Q}_{p(n)}$.
    Given two $n$-bit strings $x$ and $y$, query access to $F$ follows \cref{def:FNH}.
    That is, $\mathcal{Q}^{(1)}_F(x,y)$ first checks if $(x,y) \in \mathcal{P}$ or $\mathcal{N}$.
    This check is done by performing the calculation $\alpha(x)\alpha(y)\mel{x}{\hat{H}}{y}$ which requires two queries to $\hat{\xi}$ and one to $\hat{H}$.
    If $(x,y) \in \mathcal{N}$ then $\mathcal{Q}^{(1)}_F(x,y) = \mathcal{Q}^{(1)}_{\hat{H}}(x,y)$.
    If $(x,y) \in \mathcal{P}$ then $\mathcal{Q}^{(1)}_F(x,y) = 0$.
    The diagonal elements of $F$ can be computed by first querying $\mathcal{Q}^{(2)}_{\hat{H}}(x)$ to find the non-zero entries of $\hat{H}$.
    Then, for each $y$ we define a set $A = \{y : (x,y) \in \mathcal{P}\}$.
    This can be done efficiently since $\hat{H}$ is $\Theta(n^k)$-sparse.
    It follows that $\mathcal{Q}^{(1)}_F(x,x) = \mathcal{Q}^{(1)}_{\hat{H}}(x,x) + \sum_{y \in A} \frac{\alpha(y)}{\alpha(x)}\mathcal{Q}^{(1)}_{\hat{H}}(x,y)$.

    To query the rows of $F$ given a row index $x$, we first query $\mathcal{Q}^{(2)}_{\hat{H}}(x)$ to find the non-zero entries of $\hat{H}$.
    We then partition the set into three disjoint sets: $A = \{y : (x,y) \in \mathcal{P}\}$, $B = \{y : (x,y) \in \mathcal{N}\}$ and $C = \{y : (x,y) \notin \mathcal{P} \cup \mathcal{N}\}$.
    This can be done efficiently since $\hat{H}$ is $\Theta(n^k)$-sparse and thus $\mathcal{Q}^{(2)}_{\hat{H}}(x) = A \sqcup B \sqcup C$.
    We exclude the elements of $A$ from the output since they are zero.
    It follows that $\mathcal{Q}^{(2)}_F(x) = B \sqcup C$.
\end{proof}

\begin{corollary}
    Let $H$ be a $k$-local Hamiltonian on $n$ qubits with a $\mathbb{C}_{p(n)}$-succinct ground state $\ket{\xi}$.
    Then, take $\hat{H}$ as defined in \cref{thrm:real-complex-ham-reduction}, and subsequently $F = F(\hat{\xi},\hat{H})$ as defined in \cref{lma:real-stoq-ham-reduction}.
    Given query access to $H$ and $\ket{\xi}$, there exists query access to the fixed-node Hamiltonian $F$.
\end{corollary}

\subsection{Extension of MA Containment}\label{sec:extension-MA-containment}
We extend the \cl{MA} containment proof to succinct states with amplitudes in $\mathbb{C}_{p(n)}\dbbrckt{\sqrt{\cdot}}$, which we show is the natural encoding for normalised history states of Feynman-Kitaev Hamiltonians used in the \clw{MA}{hardness} construction.
The original proof~\cite{jiang2025local} assumed history-state amplitudes lie in $\mathbb{Q}_{p(n)}$, but proper normalisation introduces square root terms: normalisation factors of the form $1/\sqrt{K+1}$ rather than $1/(K+1)$.
Since normalised succinct states are central to the problem definition, the encoding must accommodate these terms.

Fortunately, amplitudes in $\mathbb{C}_{p(n)}\dbbrckt{\sqrt{\cdot}}$ remain classically tractable.
The encoding elements can be extracted and manipulated by efficient classical algorithms using layered binary representations: rationals as numerator-denominator pairs, signed rationals via sign bits, complex numbers as pairs of rationals, and algebraic numbers through their minimal polynomial representations.
The query model thus becomes slightly more expressive than one outputting rational approximations, but without compromising computational efficiency.
This extension is essential for Solution~\ref{sol:2}, where we establish \clw{MA}{hardness} for $3$-local Hamiltonians with properly normalised history states.

\acorol

It is clear that the proof of containment in \cl{MA} outlined in \refcite{jiang2025local} can be adapted to this more general case as we have proven the modification to the succinct encodings retain polynomial-time computability.
A natural question that follows this is when the classical algorithm can only output approximations to the amplitudes.
The \cl{MA} protocol outlined in \refcite{jiang2025local} is sensitive to the precision of the amplitudes.
However, intuitively speaking, for a large enough polynomial $q(n)$ the precision should be sufficient to surpass induced protocol errors and subsequently the inverse-polynomial gap in \cref{def:LHP-succinct-state}.

\begin{restatable}[]{conjecture}{conjectureMAContainment}
    The \cl{MA} protocol outlined in \textnormal{Ref.~\cite{jiang2025local}} is robust again the inclusion of succinct states expressing values to a precision $2^{-q(n)}$ for some sufficiently large polynomial $q(n)$.
\end{restatable}

The consequence of this conjecture (if proven true) is that states with a succinct representation such that the amplitudes are output to a precision of $2^{-q(n)}$ can be considered in the \cl{MA} protocol.
Whether or not these states are interesting or realistic is beyond our current scope.

\subsection{Class Hardness}\label{sec:hardness}
The \clw{MA}{hardness} reduction uses the Feynman-Kitaev clock construction, following the approach of Bravyi \emph{et al}.~\cite{bravyi2006complexity}, who originally used this technique for the \sc{Local Stoquastic Hamiltonian} problem.
A key additional requirement is proving that the ground state of the Feynman-Kitaev clock Hamiltonian is a succinct state.

Clock Hamiltonian constructions $H+V$ typically decompose as: 
\begin{equation*}
    H = H_{\text{in}} + H_{\text{prop}} + H_{\text{clock}}, \quad V = H_{\text{out}},
\end{equation*}
where $H_{\text{in}}$ and $H_{\text{out}}$ penalise incorrect input and output states, $H_{\text{prop}}$ penalises incorrect state propagation, and $H_{\text{clock}}$ penalises violations of the clock register.

\begin{remark}\label{rmk:history-state-ground}
    The \emph{history state} (see \cref{eq:general-history-state}) of the Feynman-Kitaev clock Hamiltonian is a frustration-free ground state of $H$, but is a ground state of $H+V$ only when the underlying circuit accepts with probability $1$, i.e., perfect completeness holds.
    Since \cl{MA} verification circuits can be constructed with perfect completeness, the history state is indeed the ground state of $H+V$ in our reduction.
\end{remark}

The general form of the history state follows as
\begin{equation}\label{eq:general-history-state}
    \ket{\eta(x,\chi)} \coloneqq \frac{1}{\sqrt{K+1}} \sum_{t=0}^{K} \ket{\varphi_t}\ket{t},
\end{equation}
where $\ket{t}$ is the unary encoding of the time step $t$ and $\ket{\varphi_t} = U_t \ket{\varphi_{t-1}}$, with $\ket{\varphi_0} = \ket{x,\chi,0^m,+^p}$ being the initial state of the computation.
Notice that the amplitudes are uniform and normalised.
Furthermore, by choosing a unary encoding for the clock register, $\ket{t}$, the amplitudes of the history state are defined as 
\begin{equation*}
    \frac{1}{\sqrt{K+1}} \frac{1}{\sqrt{2^p}},
\end{equation*}
since \cl{MA} circuits have gates over the field $\mathbb{F}_2$ and $p$ many $\ket{+}$ ancillae are used.

\aprop

In this proposition we have made the modification from some previous results concerning the succinct-ness of subset states, cf. \cref{rmk:subset-more-general}.
Specifically, we have explicitly included the tracking of the $\sqrt{\cdot}$ quantity.
Recall that we previously assumed that subset states were $\mathbb{N}_1$-succinct states for simplicity, cf. \cref{lma:classical-gates-subset state}.
In principle, we can propagate the idea forward that subset states are $\mathbb{N}_{1}$-succinct states to the history state.
However, this is undesirable and does not reflect the arguments and analysis of \refcite{jiang2025local}. 

\subparagraph{Revisiting the History State Normalisation.}\label{sec:discussion-proof-error}
Re-examining the proof of \clw{MA}{hardness}~\cite{jiang2025local} the history state is defined with the prefactor $1/(K+1)$, rather than $1/\sqrt{K+1}$. 
In the former case, the value is clearly a rational number in $\mathbb{Q}^+_{2 \log_2(K+1)}$ ($1/(2^{p/2} (K+1))$ is a quotient of two integers when $p$ is even).
However, in the latter case the value is not necessarily rational. 
The original argument states that the unscaled amplitudes $\braket{x}{\eta}$ satisfy the succinct property\footnote{The argument implies that there exists a classical query algorithm that can output the amplitudes of the history state when the global constant $c=1$, and therefore express $\braket{x}{\eta}$ exactly.} and can be computed by a polynomial-sized classical circuit.

Considering an unnormalised history state presents several difficulties.
First, the proof would need to establish that the (normalised) history state is inherently succinct.
Second, the original proof consistently normalises the tensor product of $\ket{+}$ states.\footnote{The number of $\ket{+}$ ancillae is made even to ensure the normalised value is rational, i.e., $p$ is even as we readily assumed.}
More fundamentally, the definition of a succinct state requires normalisation for the classical circuit to compute $\alpha(j)$ up to a common factor; thus, our candidate history state $\ket{\eta}$ must be normalised, meaning its components $\braket{x}{\eta}$ are not necessarily rational.
These requirements for a normalised history state complicate stating the precise outcome of the reduction if an unnormalised version were used.

We present four potential solutions to this problem.
\begin{enumerate}[label=(\Roman*), ref=(\Roman*)]
    \item \label{sol:1}Circuits with $(K+1) \in \{x : \exists\,n \in \mathbb{N}(x = n^2)\}$.
    \item \label{sol:2}Permit algebraic encoding of the form $\mathbb{A}^{(\sqrt{\cdot})}_1 \times \mathbb{Q}_{2\log_2(K+1)}^+$.
    \item \label{sol:3}Consider the pre-idled quantum verifier.
    \item \label{sol:4}Adjust the proof for scaled amplitudes.
\end{enumerate}

Solution \ref{sol:1} allows the normalisation factor in the history state to be a rational number.
This is not a general solution and is unsatisfactory.
What this solution does tell us is that the \clw{MA}{hardness} proof is valid for all circuits with $K+1$ a square number.\footnote{Or an integer power of $2$ is appropriate provided it does not result in an exponential number of gates.}
Solution \ref{sol:2} is more general and thus captures a wider class of circuits.\footnote{All circuits for that matter!}
An implication of this solution is that the scaling value $c$ may not be necessary, and the amplitudes can be expressed exactly~\cite[Appendix 5.2]{jiang2025local}.
The consequence of considering this particular solution is that we require a ``more powerful'' classical algorithm to determine the amplitudes.
What we mean to say is that while it is not unreasonable to consider that amplitudes \emph{can} be expressed in this form, the classical algorithm that can compute such amplitudes is not necessarily as straightforward as the ``original'' succinct states considered.
Solution \ref{sol:3} is related to the original error in the proof and thus also Solution \ref{sol:1}.
By pre-idling the circuit with a polynomial number of identity gates, we can force $N+K+1$ to be a square number.\footnote{Here $K$ is the total number of original gates and $N$ is the number of pre-idling Identity gates.}
A consequence of this solution is a change in the spectral gap of the Hamiltonian.
We discuss this further in Appendix~\ref{app:pre-idled-quantum-verifier}.
Solution \ref{sol:4} works for the original setting of the proof where the prefactor of the query algorithm can be the value $\sqrt{K+1}$, e.g., the query algorithm encode $\sqrt{K+1}\braket{x}{\eta}$.
This solution however, restricts and presents a situation where we are no longer interested in the minimal structure requirements of the amplitude encodings, as outlined in the previous sections and still may present exact encoding issues.

In an effort to retain generality and the theme of this work, we will consider Solution \ref{sol:2} and provide a proof of the \clw{MA}{hardness} under this assumption, i.e., \cref{prop:history-state-solution-assumption} holds.
One reason for taking this approach is that we be slightly more relaxed with the specific encoding present.
This then allows us to consider a wider range of possible states.

\section{Locality Reduction}\label{sec:locality-reduction}
The reducing of locality from six to four is a straightforward application of the method described in \refcite{kempe2003local}.
Specifically, we take the standard Feynman-Kitaev clock construction and make two small modifications.
The first is to only couple the unitaries $U_t$ with a single clock qubit and the second is to very heavily penalise incorrect clock propagations (a.k.a illegal clock states).
The latter is achieved by scaling the $H_{\text{clock}}$ term by a factor of $K^{12}$ where $K$ denotes the number of gates in the circuit. 

\athrm

We can study the effect on the locality of the Hamiltonian by considering the decomposition of the Toffoli gates.
Given any \clsb{MA}{q} circuit, we will assume that all gates are Toffoli gates, since they are universal for classical computation.
We now explore a class of quantum circuits that utilise the full gate set $\{\Gate{Cnot},\Gate{Had},T\}$ alongside $\ket{+}$-ancillae, but are heavily constrained by a strict structural rule.
Specifically, we impose that the sequence of operations within the circuit must mimic the exact behaviour of Toffoli gates, with each $15$-gate block corresponding precisely to the action of a single Toffoli gate.
This constraint ensures that the circuits cannot perform operations beyond those achievable by Toffoli gates alone, despite the more powerful quantum gates available in the set.
For instance, operations that might introduce additional quantum phenomena, such as inserting Hadamard gates between Toffoli gates, are explicitly prohibited.
As a result, the computational power of these restricted circuits is exactly equivalent to that of circuits composed solely of Toffoli gates.
The heavy structural constraint effectively nullifies any potential quantum advantage from using the $\{\Gate{Cnot},\Gate{Had},T\}$ gates, reducing the circuit to one that is fundamentally classical in nature, albeit with quantum ancillae, and thereby maintaining the equivalence to classical \cl{MA} circuits.\footnote{We note that it could be argued the $\ket{+}$-ancillae can be efficiently prepared under this new circuit family.}
We refer to this class of circuits as \emph{structured Toffoli-equivalent circuits} (STEC).
Then the new class of promise problems utilising such circuits is formally defined using a preliminary definition.

\begin{definition}[Structured Toffoli-Equivalent Verification Circuit (STEVC)]\label{def:STEVC}
    A structured Toffoli-equivalent verification circuit is a tuple $J_n = (n,w,m,p,U)$ where $n$ is the number of input qubits, $w$ is the number of proof qubits, $m$ is the number of ancillae initialised in the $\ket{0}$ state and $p$ is the number of ancillae initialised in the $\ket{+}$ state.
    The circuit $U$ is a quantum circuit, specifically a structured Toffoli-equivalent circuit, on $M\coloneqq n + w + m + p$ qubits, comprised of $K = O(\poly{n})$ gates.
    The acceptance probability of a structured Toffoli-equivalent verification circuit $J_n$, given some input string $x\in \varSigma^n$ and a proof state $\ket{\chi} \in \mathbb{K}^{\otimes w}$ is defined as:
    \begin{equation*}
        \Pr\left[J_n(x,\ket{\chi})\right]= \bra{\phi}U^\dagger \Pi_{\text{out}} U \ket{\phi},
    \end{equation*}
    where $\ket{\phi} = \ket{x,\chi,0^{m},+^{p}}$ and $\Pi_{\text{out}} = \ketbra{1}_1$ is a projector onto the output qubit.
\end{definition}

\begin{definition}[\clsb{StMA}{q}]
    \label{def:StMAq}
    A promise problem $L = (L_{\textsc{yes}}, L_{\textsc{no}})$ belongs to the class \clsb{StMA}{q} if there exists a polynomial-time generated stoquastic circuit family $\mathcal{J} = \{J_n : n \in \mathbb{N}\}$, where each STEVC $J_n$ acts on $n + w+m + p$ input qubits and produces one output qubit, such that:
    \begin{itemize}
        \item[] \textbf{Completeness}: For all $x\in L_{\textsc{yes}}$, $\exists \ket{\chi}\in \mathbb{K}^{\otimes w}$, such that, $ \Pr\left[J_{|x|}(x,\ket{\chi})=\mathtt{1} \right] \geq 2/3$
        \item[] \textbf{Soundness}: For all $x\in L_{\textsc{no}}$, $\forall\ket{\chi}\in \mathbb{K}^{\otimes w}$, then, $ \Pr\left[J_{|x|}(x,\ket{\chi})=\mathtt{1} \right] \leq 1/3$
    \end{itemize}
\end{definition}

By the heavy structural constraint imposed on the circuits, the following lemma is immediate:

\flma

\subsection{Reduction to 3-Local Hamiltonians}\label{sec:3-local-reduction}
Using the same clock construction as in the proof \cref{thrm:4-local-stoquastic-hamiltonian-with-Qp-succinct-ground-states} and leveraging the solution of \cref{lma:history-state-ex-2}, we have the following result:

\bthrm

Due to the decomposition of the Toffoli gates (see Appendix~\ref{app:toffoli-gate-decomposition} for detailed discussions on this decomposition), the Hamiltonian is now $3$-local but as a consequence it is no longer stoquastic or real.
Furthermore, the gate set $\mathcal{G} = \{\Gate{Cnot},\Gate{Had},T\}$ generates unitaries with elements in the field $\mathbb{Q}({\rm i},\sqrt{2})$ with amplitudes being of the form $a + {\rm i} b + \sqrt{2}(c + {\rm i} d)$ where $a,b,c,d$ are rational numbers.
This implies that we cannot explicit say anything about the complexity of the problem when we restrict to either stoquastic or real Hamiltonians.
One other consequence of the decomposition is the form of the succinct states.
Since the Hadamard and $T$ gates introduce irrational components to the amplitudes at given intervals in the history state, we cannot say anything about the complexity of the problem when the ground states are $\mathbb{C}_{p(n)}$-succinct.
It is likely these problems are still \clw{MA}{hard}.

An application of \cref{thrm:main-result-3l} to the easy witness case of Ref.~\cite{liu2021stoqma} and independently extensions to spatially sparse circuits are presented in Appendix~\ref{app:stoquastic-hamiltonians-easy-witness} and
Appendix~\ref{app:local-hamiltonians-spatially-sparse-graphs}, respectively.

\subsection{Reduction to 2-Local Hamiltonians}\label{sec:2-local-reduction}
We now describe how to obtain the class of circuits we call \emph{regular-interval structured Toffoli-equivalent circuits} (RI-STECs).
These circuits are equivalent to STECs but are built from the gate set $\mathcal{R} = \{\Gate{C}Z, \Gate{Had}, T, Z\}$, and each $\Gate{C}Z$ gate is surrounded by a fixed number of single-qubit gates on both sides.
The motivation for this particular structure comes from the clock Hamiltonian construction of Kempe, Kitaev, and Regev~\cite{kempe2006complexity}.
To encode the propagation of a $\Gate{C}Z$ gate using only $2$-local terms, the construction requires each $\Gate{C}Z$ gate to be padded on both sides by a fixed number of single-qubit gates; this padding ensures that when projecting specific penalty terms, in the Hamiltonian construction, to appropriate nullspaces, the resulting effective Hamiltonian recovers the standard $\Gate{C}Z$ propagation Hamiltonian as an emergent lower bound.

The conversion from STECs to RI-STECs proceeds in two steps.
First, we replace each \Gate{Cnot} gate in the STEC with a $\Gate{C}Z$ gate and two \Gate{Had} gates, passing to a circuit over $\mathcal{R}$.
Second, we enforce the regular-interval structure by inserting $Z$ gates adjacent to each $\Gate{C}Z$ gate, exploiting two properties: $Z^2 = I$, so pairs of $Z$ gates cancel and do not change the computed unitary; and $Z$ commutes with $\Gate{C}Z$, so $Z$ gates can be inserted on either side of any $\Gate{C}Z$ gate without affecting the computation.
Together these allow each $\Gate{C}Z$ gate to be padded as illustrated in \cref{fig:CZ-padding}, producing the required regular-interval structure.

\begin{figure}[!ht]
    \centering
    \begin{tikzpicture}
        \draw[thick] (5.5,0) -- (0,0);
        \draw[thick] (5.5,-1) -- (0,-1);
        \draw[thick, fill=white] (0.5,-1.25) rectangle ++(0.5,0.5) node[pos=0.5] {$Z$};
        \draw[thick, fill=white] (1.5,-0.25) rectangle ++(0.5,0.5) node[pos=0.5] {$Z$};
        \draw[thick, fill=black] (2.75,0) circle (0.06); \draw[thick, fill=black] (2.75,-1) circle (0.06); \draw[thick] (2.75,0) -- (2.75,-1);
        \draw[thick, fill=white] (3.5,-0.25) rectangle ++(0.5,0.5) node[pos=0.5] {$Z$};
        \draw[thick, fill=white] (4.5,-1.25) rectangle ++(0.5,0.5) node[pos=0.5] {$Z$};
        \draw[thick, dashed, myOrange] (0.25,0.5) -- (0.25,-1.5) node[below, rotate=20] {\scriptsize $t-3$};
        \draw[thick, dashed, myOrange] (1.25,0.5) -- (1.25,-1.5) node[below, rotate=20] {\scriptsize $t-2$};
        \draw[thick, dashed, myOrange] (2.25,0.5) -- (2.25,-1.5) node[below, rotate=20] {\scriptsize $t-1$};
        \draw[thick, dashed, myOrange] (3.25,0.5) -- (3.25,-1.5) node[below, rotate=20] {\scriptsize $t$};
        \draw[thick, dashed, myOrange] (4.25,0.5) -- (4.25,-1.5) node[below, rotate=20] {\scriptsize $t+1$};
        \draw[thick, dashed, myOrange] (5.25,0.5) -- (5.25,-1.5) node[below, rotate=20] {\scriptsize $t+2$};
    \end{tikzpicture}
    \caption{
        Padding a $\Gate{C}Z$ gate at time step $t$ with two $Z$ gates on each support qubit, occupying time steps $t-3$ and $t-2$ (before the $\Gate{C}Z$) and $t+1$ and $t+2$ (after the $\Gate{C}Z$).
        The $Z$ gates do not alter the computed unitary: they cancel in pairs via $Z^2 = I$ and commute past the $\Gate{C}Z$ gate.
        This padding ensures that the $\Gate{C}Z$ gate is embedded in a locally uniform single-qubit context, as required by the $2$-local Hamiltonian construction.
        Reproduced from \cite[Fig.~2]{kempe2006complexity}.
    }
    \label{fig:CZ-padding}
\end{figure}

Recall that STECs are obtained from an \cl{MA}\textsubscript{q} verification circuit of $K$ \Gate{Toffoli} gates by decomposing each \Gate{Toffoli} gate into a sequence of $15$ gates from $\mathcal{G} = \{\Gate{Cnot}, \Gate{Had}, T\}$, yielding a circuit of $15K$ gates in total.
We now convert to a circuit over $\mathcal{R}$ in two stages.

\begin{restatable}[Gate Set Conversion]{proposition}{propositionGateSetConversion}\label{prop:gate_set_conversion}
    For any STEC $J$ over $\mathcal{G}$, there exists an efficiently computable circuit $C'(J)$ over $\mathcal{R} = \{\Gate{C}Z, \Gate{Had}, T, Z\}$ acting on the same number of qubits and computing the same unitary as $J$.
\end{restatable}

A proof is given in Appendix~\ref{app:2local}.
The key step is the identity $\Gate{Cnot} = (I \otimes \Gate{Had})\, \Gate{C}Z\, (I \otimes \Gate{Had})$, which replaces each \Gate{Cnot} gate with a $\Gate{C}Z$ gate conjugated by \Gate{Had} gates on the target qubit, at the cost of a constant-factor increase in circuit length.

We then enforce the regular-interval condition.


\begin{restatable}[Regular-Interval Structure]{proposition}{propositionRegularIntervalStructure}
\label{prop:regular_interval_structure}
    For any circuit $C'(J)$ over $\mathcal{R}$ associated with a STEC of $K$ \Gate{Toffoli} gates, there exists an efficiently computable RI-STEC $C$ over $\mathcal{R}$, acting on one additional
    \emph{buffer} qubit initialised to $\ket{0}$, 
    such that the unitary implemented by $C$ satisfies $U_C = U_{C'(J)} \otimes I_{\emph{buffer}}$, and in which every $\Gate{C}Z$ gate is preceded and followed by exactly two $Z$ gates on each of its support qubits, with consecutive $\Gate{C}Z$ gates separated by exactly $\ell = 9$ single-qubit gates.
\end{restatable}

The buffer qubit plays a purely structural role; it absorbs excess padding gates in regions where the gap between consecutive $\Gate{C}Z$ gates would otherwise be insufficient to accommodate the required $\ell = 9$ single-qubit timesteps without acting on qubits involved in the computation.
Since $Z\ket{0} = \ket{0}$, the buffer qubit contributes no non-trivial amplitude to the history state and does not affect the succinct state class.
The parameter $\ell = 9$ is determined by the construction in Ref.~\cite{kempe2006complexity}, although there may exist other choices of $\ell$ that also work.

The corresponding promise class \clsb{RIStMA}{q} is defined by straightforward modifications to \cref{def:STEVC} and \cref{def:StMAq}, replacing the gate set $\mathcal{G}$ with $\mathcal{R}$ and requiring the regular-interval condition.
The following corollary is then immediate from \cref{prop:gate_set_conversion}, \cref{prop:regular_interval_structure}, and the equivalence $\clsb{StMA}{q} = \clsb{MA}{q}$.

\ricor

Let $C_{|x|}$ denote the verification circuit of an instance of a promise problem in \clsb{RIStMA}{q}.
We apply the construction of Ref.~\cite{kempe2006complexity} to map $C_{|x|}$ to a $2$-local Hamiltonian $H$ whose ground state is the history state of $C_{|x|}$ at zero energy.
At a high level, the Hamiltonian structured as
\begin{equation*}
    H = H_{\text{in}} + H_{\text{out}} + H_{\text{clock}} + H_{\text{prop}}^{(1)} + H_{\text{prop}}^{(2)},
\end{equation*}
where $H_{\text{in}}$ penalises states in which the input, ancilla, and coin qubits are not correctly initialised at time $t = 0$; 
$H_{\text{out}}$ penalises states in which the output qubit is not in the accepting state $\ket{1}$ at the final timestep $t = K$; 
$H_{\text{clock}}$ enforces unary clock structure by penalising illegal clock configurations; 
$H_{\text{prop}}^{(1)}$ encodes the propagation of single-qubit gates via the standard circuit-to-Hamiltonian construction, with each term being $2$-local (one computational qubit and one clock qubit); and $H_{\text{prop}}^{(2)}$ encodes the propagation of $\Gate{C}Z$ gates.

The term $H_{\text{prop}}^{(2)}$ deserves special attention.
It contains no explicit $\Gate{C}Z$ operator, rather it is a sum of $2$-local terms comprising clock-only interactions $H_{\text{time}}^{(2)}(t)$ and operators $H_{\text{qb}}^{(2)}(t)$, each coupling a single computational qubit to the clock at regular intervals.
The action of the $\Gate{C}Z$ gate is not encoded directly but instead emerges from the structure of $H_{\text{prop}}^{(2)}$ upon restriction to the nullspace of $H_{\text{prop}}^{(1)}$.
Specifically, within this nullspace, it is shown that the projected Hamiltonian term dominates a self-adjoint operator $\tilde{H}$ that takes precisely the form of the standard propagation Hamiltonian for the $\Gate{C}Z$ gates:
\begin{equation*}
    \tilde{H} = \ketbra{t} + \ketbra{t-1} - \Gate{C}Z\ketbra{t}{t-1} - \Gate{C}Z\ketbra{t-1}{t}.
\end{equation*}
It follows that $\tilde{H} \preceq H_{\text{prop}}^{(2)}\bigr|_{S}$ (see \cref{prop:tilde_H}), where $S$ is the nullspace of $H_{\text{prop}}^{(1)}$, and thus $H_{\text{prop}}^{(2)}$ \emph{implicitly} implements the $\Gate{C}Z$ propagation.

The spectral gap of $H$ is lower bounded via repeated applications of the \emph{nullspace projection lemma} (see \cref{lem:nullspace_projection} and \cref{cor:nullspace_projection}), which gives a lower bound on the smallest non-zero eigenvalue of a sum of positive semidefinite operators in terms of the structure of the individual terms.
It remains to verify that the ground state of $H$ is succinct.
The history state of $C_{|x|}$ is a superposition over computational basis states of the circuit at each time step; its amplitudes lie in $\mathbb{Q}({\rm i}, \sqrt{2})$ by the same argument as for STECs, since $\mathcal{R}$ generates the same field as $\mathcal{G}$.
Succinctness then follows from the equivalences $\Gate{C}Z \leftrightarrow \{\Gate{Cnot}, \Gate{Had}\}$ and $Z \leftrightarrow T$-type gates, together with straightforward adaptations of \cref{cor:T-T-dagger-gate-seq-subset state} and \cref{cor:CRG-T-T-dagger-Hadamard-gate-seq-subset state}, which show that the history state of $C_{|x|}$ is a $\mathbb{C}_{p(n)}\dbbrckt{\sqrt{\cdot}}$-succinct state.
Note that the buffer qubit, initialised to $\ket{0}$ and acted upon only by $Z$ gates, does not affect the succinct state class since $Z\ket{0} = \ket{0}$ up to a global phase that is absorbed into the amplitude representation.

\cthrm

Together with \cref{thrm:main-result-3l}, \cref{thrm:main-result-2l} establishes that the \clw{MA}{completeness} of the \sc{Succinct State Local Hamiltonian} problem is robust to locality: the problem remains \clw{MA}{complete} for any fixed $k \geq 2$.

\section{Conclusion}\label{sec:conclusion}
In this work we study a variant of the \sc{Local Hamiltonian} problem where there is additional promise on the form of the ground state.
Specifically, the \sc{Succinct State Local Hamiltonian} problem introduces the notion of succinct ground states, which can be efficiently described using a classical query algorithm.
The amplitudes of the ground state are expressed in an exact rational form, with real and imaginary parts $a + {\rm i}b$ where $a,b \in \mathbb{Q}$; both components can be represented in a polynomial number of bits.
This definition of succinct state naturally gives rise to multiple classes of such states.
In contrast to the standard problem, which is \clw{QMA}{complete}, it has been shown that the \sc{Succinct State Local Hamiltonian} problem is (promise) \clw{MA}{complete}~\cite{liu2021stoqma,jiang2025local}.
Our results have shown that this complexity classification remains, even for $2$-local Hamiltonians with succinct ground states.

\begin{result*}[(Informal) \cref{thrm:main-result-2l}]
    The \sc{Succinct State $2$-Local Hamiltonian} problem is \clw{MA}{complete}.
\end{result*}

To achieve this result, we explored simple examples of succinct states and the resulting effects of combining these states in different ways.
For instance, given a succinct state, is the state still succinct when acted on by a unitary operator?
By defining four natural classes of succinct states, each admitting exact binary representations, we were able to characterise a wider range of states.
This was possible via the use of algebraic encodings of rational values, a common idea in classical computing.
Using these ideas, we constructed arguments demonstrating that the history state, resulting from the Feynman-Kitaev clock construction of \cl{MA} circuits, was a succinct state.
Combining this result with previous work~\cite{jiang2025local} was sufficient to prove our main result.
We have been able to fully resolve the question of whether the complexity of the problem depends on the locality of the Hamiltonian.

The \sc{Succinct State Local Hamiltonian} problem represents an interesting modification of the standard \sc{Local Hamiltonian} problem.
Few results have studied the complexity of determining the energy of states of a given type \cite{kallaugher2024complexity,bremner2025parameterized}, and even fewer have restricted the form of the ground state.
It is clear that upon doing so, the class of local Hamiltonians for which the problem is defined on is a lot smaller than the general case.
However, this particular line of thinking is useful in the context of other Hamiltonian complexity problems.
For example, when local Hamiltonians are guided by a state that has promised overlap with the true ground state, the problem of determining the ground-state energy to inverse-polynomial precision is \clw{BQP}{complete} \cite{richter2007two}.
In fact, if the guiding state is given via an efficient quantum circuit and has a promised overlap of at least inverse-polynomial in the size of the system, it is well-understood that repeated applications of the Quantum Phase Estimation algorithm can be used to estimate the energy of the Hamiltonian to high precision.
The problem studied in this work may narrow the gap of applicability of such ideas since known classical heuristics often approximate ground states as having succinct descriptions \cite{tsuneda2014density,white1992density,arute2020hartree}.

\cref{fig:flow-diagram} provides a (pictorial) structured overview of the complexity landscape we explored.
The combination of ideas is difficult to summarise in a linear narrative due to the different types of succinct states considered and the different way complexity results were obtained.

\subparagraph{Discussion and Future Work.} 
We demonstrated there are various types of succinct states, and that characterising their effects on the \sc{Local Hamiltonian} problem can be non-trivial.
The definitions and notation presented here are intended to offer a clearer framework for understanding these complexities.
We have shown that being specific about the type of succinct state is crucial.
For instance, the history state resulting from the Feynman-Kitaev clock construction typically does not have uniform rational amplitudes.
This highlights the importance of precisely defining the class of succinct states used in the constructions.
Our results on the study of succinct states and the action of Hamiltonian mappings may have implications for problems such as state tomography and verification.
A verification algorithm working for stoquastic Hamiltonians can likely be extended to more general Hamiltonians, provided access to the state can be interpreted similar to the above analysis.
Specifically, we have shown informed queries to the Hamiltonian and the state are sufficient to determine resultant properties under certain mappings.

Due to the equivalency between the history states and subset states, it might be argued that the query algorithm for history states could simply verify membership of a computational basis state.
However, upon deeper analysis of the original proof~\cite{jiang2025local}, it becomes apparent that this is insufficient.
Specifically, the amplitudes $\braket{x}{\eta}$ are claimed to be exactly representable as rational values, but this is generally not the case.
Thus, the query algorithm must output the normalised amplitudes with exact precision.

Our results do not follow a completely linear narrative; this is in part due to the different types of succinct states we considered.
To offer a more structured overview of the complexity landscape, see \cref{fig:flow-diagram}.
Important open questions that remain, include:
\begin{enumerate}
    \item \textbf{Conjecture 1.}~\emph{The \cl{MA} protocol outlined in \textnormal{Ref.~\cite{jiang2025local}} is robust again the inclusion of succinct states expressing values to a precision $2^{-q(n)}$ for some sufficiently large polynomial $q(n)$.}
    \item \textbf{Conjecture 2.}~\emph{The \sc{Succinct State $2$-Local Stoquastic Hamiltonian} problem is \clw{MA}{complete}, even on $2$D lattices.}
    \item Investigating the consequences of restricting elements of the Hamiltonian to being exactly representable in a fixed number of bits.
    \item Developing a perturbative gadget framework that preserves the succinctness of the ground state.
    \item Determining the complexity of $2$-local Hamiltonians with $\mathbb{C}_{p(n)}$-succinct ground states.
    \item Investigating the complexity of \sc{Succinct State Frustration-Free Local Hamiltonian} problem.
    \item Effects of Hamiltonian element precision.
\end{enumerate}

The third point is an interesting question that could lead to a better understanding of how we construct verification circuits for Hamiltonian terms.
Demanding the elements of the Hamiltonian terms are exactly expressible in a fixed number of bits can have impact on the depth of the verification circuit.
For example a constant number of bits for specification can impact the accuracy of the verification circuit thus requiring a large depth which may be unwanted.
This is of course important to consider when dealing with practical Hamiltonians, in Quantum Chemistry for example.
In partial response to the last point, we explored the \clw{MA}{hardness} of the problem when the Hamiltonian is defined on a spatially sparse graph (see Appendix~\ref{app:local-hamiltonians-spatially-sparse-graphs}).
The other points represent promising directions for future research and could lead to a deeper understanding of the complexity of the \sc{Succinct State Local Hamiltonian} problem.

It follows from standard isomorphism arguments that the complex $T$ gate can be replaced with a real gate of the form $I \oplus \tfrac{1}{\sqrt{2}}\big(\begin{smallmatrix}1 & -1\\1 & 1\end{smallmatrix}\big)$, therefore, is it possible to extend the technique direct Hamiltonian reduction technique of Ref.~\cite{kempe2006complexity} in combination with a real version of RI-STECs to show the \cl{MA}-completeness of the \sc{Succinct State $2$-Local Real Hamiltonian} problem?

As a final remark, we comment on the~\cite[Conjecture 3]{jiang2025local}, specifically on the underlying state type --- \emph{strong guided states}~\cite[Definition 2]{jiang2025local}.
The idea of strong guided states is motivated by the original work of Bravyi~\cite{bravyi2015monte} and is defined as follows:

\begin{definition*}[Strong Guided States~\cite{jiang2025local}]
    Let $\ket{\psi}$ be an $n$-qubit normalised state.
    We say that $\ket{\psi}$ admits a strong guiding state if there exists an $n$-qubit normalised state $\ket{\eta}$, such that $\ket{\eta}$ is a succinct state and satisfies:
    \begin{equation*}
        \braket{\eta}{x} \braket{x}{\psi} \geq \frac{\abs{\braket{x}{\psi}}^2}{\poly{n}},
    \end{equation*}
    for all $x \in \{0,1\}^n$.
\end{definition*}

This describes an entry-wise correlation between the states and is an extremely strong condition.
Since it is not possible to define a total ordering over the complex numbers, this statement implies that there can be no relative phases between the states.
We instead propose an alternate definition that circumvents non-relative phase requirements.

\begin{definition}
    \label{def:relaxed-guiding-state}
    Let $\ket{\psi}$ be an $n$-qubit normalised state.
    We say that $\ket{\psi}$ admits an $\varepsilon$-relaxed (generalised) entry-wise guiding state if there exists an $n$-qubit normalised state $\ket{\vartheta}$, such that $\ket{\vartheta}$ is a succinct state and satisfies:
    \begin{enumerate}
        \item $\big| \abs{\vartheta_x} - \abs{\psi_x}\big| \leq \varepsilon \abs{\psi_x}, ~ \forall x \in \{0,1\}^n,$
        \item $\abs{{\rm arg}(\vartheta_x) - {\rm arg}(\psi_x)} \leq \varepsilon, ~ \forall x \in \{0,1\}^n.$
    \end{enumerate}
    Where $\vartheta_x \coloneqq \braket{x}{\vartheta}$ and $\psi_x \coloneqq \braket{x}{\psi}$.
\end{definition}

A relative error between the magnitude and an additive error between the argument of complex numbers is a more natural way to compare two values that are expected to be close.
We propose an alternative conjecture using the relaxed guiding state definition.

\begin{conjecture}\label{conj:relaxed-guiding-state}
    The \sc{Local Hamiltonian} problem with $\varepsilon$-relaxed entry-wise guiding states is \clw{MA}{complete}.
\end{conjecture}

A special case of this conjecture is true for stoquastic Hamiltonians~\cite{bravyi2015monte}.
Difficulties arise in proving this conjecture.
For example, the relationship between the guiding state and the ground states of the real Hamiltonian that is constructed from the initial Hamiltonian is difficult to establish.
Additionally, the use of the fixed-node quantum Monte Carlo method using the guiding state, rather than the ground state causes analytical problems; if the guiding state is not phase-aligned with the ground state, the fixed-node Hamiltonian constructed from the guiding state may have a ground-state energy too far from the true ground-state energy.
A route to resolving this conjecture or one similar is to determine appropriate conditions on the guiding state that render the fixed-node Hamiltonian a sufficiently-good approximation in the sense of the continuous-time Monte Carlo method (or any alternative procedure).
However, this route likely produces guiding states with extremely strong promises rendering the problem uninteresting and perhaps artificial.

\section*{Declarations}\label{sec:declarations}
\textbf{Acknowledgments.}
GW thanks Karl Lin for helpful discussions and comments on this work.
GW would also like to thank Sam Elman for feedback and a preliminary review, and Karl Rombauts and Ryan Mann for discussions on representing algebraic numbers in binary.
GW was supported by a scholarship from the Sydney Quantum Academy and also supported by the ARC Centre of Excellence for Quantum Computation and Communication Technology (CQC2T), project number CE170100012.

\textbf{Author Contributions.}
GW and KL conceived the project and developed the main ideas.
GW performed the analysis and wrote the manuscript.

\textbf{Competing Interest.}
The authors declare no competing interests.

\textbf{Data Availability.}
There are no data associated with this work.

%% file: appendix_0.tex
\begin{center}
\textbf{Complexity Landscape}
\begin{figure}[!ht]
    \centering
    \pgfdeclarelayer{background layer}
    \pgfdeclarelayer{foreground layer}
    \pgfsetlayers{background layer,main,foreground layer}
    \begin{tikzpicture}
        \pic[scale=0.9]{flow};
    \end{tikzpicture}
    \caption{A flow diagram of the complexity of the \sc{Succinct State Local Hamiltonian} problem. Arrows (loosely) represent modifications/reductions. We note that bold (solid) arrows indicate the flow of ideas akin to a reduction. The dashed arrows represent the combination of results needed to establish new complexity classifications. Smaller boxes with a grey background represent results from prior work, while orange boxes denote results from this work. The larger three boxes represent groupings of specific complexity results; namely \clw{MA}{hardness} and \cl{MA} containment. Also note that some arrows have been omitted to improve readability.}
    \label{fig:flow-diagram}
\end{figure}
\end{center}

%% file: appendix_a.tex
\section{Binary Number Class Structure}\label{app:number-form}
Recall that 
\begin{align*}
    \mathbb{A}^{(\#)}_{p_{\#}} &\coloneqq \{\alpha \in \B^{p_\#}\},\\
    \mathbb{N}_p &\coloneqq \{{\rm bin}(n) : n\in\mathbb{N}, ~ n \leq 2^p\},\\
    \mathbb{Q}^+_p &\coloneqq \{{\rm bin}(q): q\in\mathbb{Q}^+, ~ q=\frac{n}{m}, ~ n,m \in \mathbb{N}_p, ~ m\neq 0\} \\
    \mathbb{Q}_p &\coloneqq \mathbb{A}_{1}^{({\rm sgn})} \times\mathbb{A}_{1}^{({\rm sgn})} \times \mathbb{Q}_p^+,\\
    \mathbb{C}_p &\coloneqq \mathbb{Q}_p \times \mathbb{Q}_p,
\end{align*}
where 
\begin{align*}
    \forall\, n \in \mathbb{N}_p,&~0 \leq n\leq 2^p, \\
    \forall\, q \in \mathbb{Q}_p^+,&~2^{-p} \leq q \leq 2^p, \\
    \forall\, q \in \mathbb{Q}_p,&~2^{-p} \leq \abs{q} \leq 2^p, \\
    \forall\, z \in \mathbb{C}_p,&~2^{-p} \leq \abs{\re(z)},\abs{\im(z)} \leq 2^p,\\
    &\implies~2^{-p} \leq \abs{z} \leq 2^{p+\frac{1}{2}}.
\end{align*}

We provide examples of the binary encoding for these classes of numbers below.

\begin{align}
    \forall\, n \in \mathbb{N}_p, ~{\rm bin}(n)&=(\longleftarrow {\rm bin}(n) \longrightarrow), \notag\\[0.25cm]
    \forall\, \sfrac{n}{m} \in \mathbb{Q}_p^+, ~{\rm bin}(\sfrac{n}{m})&=\conc{(\longleftarrow {\rm bin}(n) \longrightarrow)}{(\longleftarrow {\rm bin}(m) \longrightarrow)}, \notag\\[0.25cm]
    \forall\, \sfrac{n}{m} \in \mathbb{Q}_p, ~{\rm bin}(\sfrac{n}{m})&=\conc{\conc{{\rm bin}({\rm sgn}_n)}{{\rm bin}({\rm sgn}_m)}}{\conc{(\longleftarrow {\rm bin}(n) \longrightarrow)}{(\longleftarrow {\rm bin}(m) \longrightarrow)}},\notag\\[0.25cm]
    \forall\, z \in \mathbb{C}_p, ~{\rm bin}(z)&=\conc{\conc{{\rm bin}({\rm sgn}_{n_{a}})}{{\rm bin}({\rm sgn}_{m_{a}})}}{\conc{(\longleftarrow {\rm bin}({n_{a}}) \longrightarrow)}{(\longleftarrow {\rm bin}({m_{a}}) \longrightarrow)}}\notag\\
    &\qquad\mathbin\Vert\conc{\conc{{\rm bin}({\rm sgn}_{n_{b}})}{{\rm bin}({\rm sgn}_{m_{b}})}}{\conc{(\longleftarrow {\rm bin}({n_{b}}) \longrightarrow)}{(\longleftarrow {\rm bin}({m_{b}}) \longrightarrow)}},\label{eq:z-bin-breakdown}
\end{align}
where $z = a + {\rm i} b$ such that $a = \sfrac{n_a}{m_a}$ and $b = \sfrac{n_b}{m_b}$.

A visual breakdown of a $\mathbb{C}_3$ number encoding is given by:
\begin{equation*}
    \frac{-6}{3} + {\rm i} \frac{2}{-7} \mapsto \rcb{myOrange!75}{1}{{\rm sgn}_{n_a}}\, \rcb{myOrange!95!black}{0}{{\rm sgn}_{m_a}}\, \rcb{myOrange!75}{110}{{n_{a}}}\, \rcb{myOrange!95!black}{011}{{m_{a}}}\, \rcb{myOrange!75}{0}{{\rm sgn}_{n_b}}\, \rcb{myOrange!95!black}{1}{{\rm sgn}_{m_b}}\, \rcb{myOrange!75}{010}{{n_{b}}}\, \rcb{myOrange!95!black}{111}{{m_{b}}}
\end{equation*}

%% file: appendix_b.tex
\section{Proof of Main Text Results}\label{app:proofs}

\cprop*
\begin{proof}
    We break the proof into cases.
    For brevity, we drop the variable $n$ from the polynomial $p(n)$.
    We also assume the denominator is non-zero in each case.

    \paragraph{$\mathbb{N}_p$-succinct states:} The ratio of two numbers $n,m \in \mathbb{N}_p$ can be exactly represented in $2p$ bits since the ratio is an element of $\mathbb{Q}_p^+$.
    This requires one call of $\mathcal{Q}_\psi(x) = n$ and one call of $\mathcal{Q}_\psi(y) = m$; the output is 
    \begin{equation*}
    \mathcal{Q}'_\psi(x,y) = \conc{{\rm bin}(n)}{{\rm bin}(m)}.
    \end{equation*}

    \paragraph{$\mathbb{Q}_p^+$-succinct states:} The ratio of two numbers $q,r \in \mathbb{Q}_p^+$ can be calculated after two multiplications are performed.
    Let $q = n_q / m_q$ and $r = n_r / m_r$, then $q/r = (n_q \cdot m_r) / (m_q \cdot n_r)$.
    Note that $n \cdot m \in \mathbb{N}_{2p}$ $\forall n,m \in \mathbb{N}_p$.
    Thus, the output ratio can be exactly expressed in $4p$ bits and is an element of $\mathbb{Q}_{2p}^+$.
    This requires one call of $\mathcal{Q}_\psi(x) = q$ and one call of $\mathcal{Q}_\psi(y) = r$ followed by the appropriate multiplication; the output is 
    \begin{equation*}
    \mathcal{Q}'_\psi(x,y) = \conc{{{\rm bin}(n_q)}\cdot{{\rm bin}(m_r)}}{{{\rm bin}(m_q)}\cdot{{\rm bin}(n_r)}}.
    \end{equation*}

    \paragraph{$\mathbb{Q}_p$-succinct states:} The ratio of two numbers $q,r \in \mathbb{Q}_p$ can be calculated after two multiplications are performed.
    Let $q = n_q / m_q$ and $r = n_r / m_r$, then $q/r = (n_q \cdot m_r) / (m_q \cdot n_r)$.
    Note that $n \cdot m \in \mathbb{N}_{2p}$ $\forall n,m \in \mathbb{N}_p$.
    Thus, the output ratio can be exactly expressed in $4p+2$ bits and is an element of $\mathbb{Q}_{2p}$.
    This requires one call of $\mathcal{Q}_\psi(x) = q$ and one call of $\mathcal{Q}_\psi(y) = r$ followed by the appropriate multiplication and logic on the sign bits; the output is 
    \begin{equation*}
    \mathcal{Q}'_\psi(x,y) = \conc{{{\rm bin}({\rm sgn}_{n_q})}\oplus{{\rm bin}({\rm sgn}_{m_r})}}{\conc{{{\rm bin}({\rm sgn}_{m_q})}\oplus{{\rm bin}({\rm sgn}_{n_r})}}{\conc{{{\rm bin}(n_q)}\cdot{{\rm bin}(m_r)}}{{{\rm bin}(m_q)}\cdot{{\rm bin}(n_r)}}}}.
    \end{equation*}

    \paragraph{$\mathbb{C}_p$-succinct states:} The ratio of two numbers $z,w \in \mathbb{C}_p$ can be calculated after four multiplications are performed.
    Let $z = a + {\rm i} b$ and $w = c + {\rm i} d$, then $z/w = (a \cdot c + b \cdot d) / (c^2 + d^2) + {\rm i} (b \cdot c - a \cdot d) / (c^2 + d^2)$.
    Note that $a\cdot b \in \mathbb{Q}_{2p}$ and $a \pm b \in \mathbb{Q}_{2p+1}$ $\forall a,b \in \mathbb{Q}_p$.
    Thus, the output ratio can be exactly expressed in $32p+12$ bits and is an element of $\mathbb{C}_{8p+2}$.
    This requires two calls of $\mathcal{Q}_\psi(x) = z$ and two calls of $\mathcal{Q}_\psi(y) = w$ followed by the appropriate multiplications, divisions, and logic on the sign bits; the output is 
    \begin{align*}
        \mathcal{Q}'_\psi(x,y) &=\conc{ {\rm bin}({\rm sgn}_{n_a})\oplus{{\rm bin}({\rm sgn}_{m_c})} }{ {\rm bin}({\rm sgn}_{m_a})}\oplus{{\rm bin}({\rm sgn}_{n_c}) } \uu \conc{ {{\rm bin}(n_a)}\cdot{{\rm bin}(m_c)} }{ {{\rm bin}(m_a)}\cdot{{\rm bin}(n_c)} }\\
        &\qquad\uu {\conc{{{\rm bin}({\rm sgn}_{n_b})}\oplus{{\rm bin}({\rm sgn}_{m_d})}}{\conc{{{\rm bin}({\rm sgn}_{m_b})}\oplus{{\rm bin}({\rm sgn}_{n_d})}}{\conc{{{\rm bin}(n_b)}\cdot{{\rm bin}(m_d)}}{{{\rm bin}(m_b)}\cdot{{\rm bin}(n_d)}}}}}.
    \end{align*}

    \paragraph{$\mathbb{C}_p^{(\omega)}$-succinct states:} The ratio of two numbers $z,w \in \mathbb{C}_p^{(\omega)}$ can be calculated after four multiplications are performed.
    Let $z = \omega^{s_z}(a + {\rm i} b)$ and $w = \omega^{s_w}(c + {\rm i} d)$, then $z/w = \omega^{s_z - s_w}((a \cdot c + b \cdot d) / (c^2 + d^2) + {\rm i} (b \cdot c - a \cdot d) / (c^2 + d^2))$.
    Note that $s_z - s_w \equiv (s_z - s_w) \mod 8$.
    Thus, using the logic above, the output ratio can be exactly expressed in $32p+15$ bits and is an element of $\mathbb{C}_{8p+2}^{(\omega)}$.
    This requires two calls of $\mathcal{Q}_\psi(x) = z$ and two calls of $\mathcal{Q}_\psi(y) = w$ followed by the appropriate multiplications, divisions, logic on the sign bits and logic on the algebraic encoding of the powers of $\omega$; the output is 
    \begin{equation*}
        \mathcal{Q}'_\psi(x,y) = {\rm bin}((s_z - s_w)\mod 8) \mathbin\Vert \cdots .
    \end{equation*}
\end{proof}

\lmasubsetsuccinct*

\begin{proof}
    The proof is trivial since the amplitude of each computational basis state, in the support, is $1/\sqrt{|S|}$.
    Let $c_S = \sqrt{|S|} \leq 2^{n/2}$, then there exists an efficient classical algorithm $\mathcal{Q}_S$ that, given an $n$-bit string $x$, outputs the exact binary representation of
    \begin{equation*}
        \mathcal{Q}_S(x) = c_S \cdot \ind_{S}(x) = 1.
    \end{equation*}
    Since we only need to output a single bit, and hence is in $\mathbb{N}_1$.
    The classical algorithm can call from the uniform distribution of the support set $S$.
\end{proof}

\tensorproductsubsetstates*

\begin{proof}
    Defining the two subset states, we have 
    \begin{align*}
        \ket{S} &= \frac{1}{\sqrt{|S|}}\sum_{s\in S} \ket{s}, \\
        \ket{T} &= \frac{1}{\sqrt{|T|}}\sum_{t\in T} \ket{t}.
    \end{align*}
    The tensor product of these states is
    \begin{equation*}
        \ket{S}\ket{T} = \frac{1}{\sqrt{|S||T|}}\sum_{s \in S, t \in T} \ket{s}\ket{t} = \frac{1}{\sqrt{|S||T|}}\sum_{r \in S \times T} \ket{r}.
    \end{equation*}
    The amplitudes of the resulting state are of the form $\gamma(r) = \gamma(\conc{s}{t}) = \alpha(s)\beta(t)$, where $\alpha(s) = 1/\sqrt{|S|}$ and $\beta(t) = 1/\sqrt{|T|}$.
    Let $c_{ST} = c_S\cdot c_T = \sqrt{|S||T|} \leq 2^{(n+m)/2}$, then there exists an efficient classical algorithm $\mathcal{Q}_{ST}$ that, given an $(n+m)$-bit string $x$, outputs the exact binary representation of
    \begin{equation*}
        \mathcal{Q}_{ST}(x) = c_{ST} \cdot \ind_{S\times T}(x) = 1.
    \end{equation*}
    Moreover, $\mathcal{Q}_{ST}(x = \conc{y}{z})$ requires one call to $\mathcal{Q}_S(y)$ and one call to $\mathcal{Q}_T(z)$ followed by the appropriate multiplication.
    The output is $1$ if and only if $y \in S$ and $z \in T$.
    Note that, given two $(n+m)$-bit strings $x$ and $x'$, the amplitude ratio $\gamma(x)/\gamma(x')$ can be efficiently calculated since 
    \begin{equation*}
        \mathcal{Q}'_{ST}(x,x') = \frac{\gamma(x)}{\gamma(x')} = \frac{\alpha(y)}{\beta(z)}\frac{\alpha(y')}{\beta(z')} = \mathcal{Q}'_S(y,y')\cdot\mathcal{Q}'_T(z,z'),
    \end{equation*}
    where $x = \conc{y}{z}$ and $x' = \conc{y'}{z'}$.
\end{proof}

\clma*

\begin{proof}
    For a superposition state in the computational basis, classically reversible gates will map computational basis states to computational basis states.
    Furthermore, since they are unitary, the normalisation of the state is preserved.
    For some $j \in [K]$, let $x$ represent an $n$-bit string, then
    \begin{align*}
        \braket{x}{A_k} &= \bra{x}R_k\ket{S} = \braket{x'}{S},\\
        \braket{x}{B_k} &= \bra{x}R_k \cdots R_1\ket{S} = \braket{x''}{S}
    \end{align*}
    where $x'$ and $x''$ are the images of $x$ under the respective gate actions.
    Since each $R_j$ is classical, both $x'$ and $x''$ can be efficiently calculated classically using $\bsr$ and $x$.
    The query algorithm for the states $\ket{A_k}$ and $\ket{B_k}$ is then 
    \begin{align*}
        \mathcal{Q}_{A_k}(x) = \mathcal{Q}_S(x'), \\
        \mathcal{Q}_{B_k}(x) = \mathcal{Q}_S(x'').
    \end{align*}
    Since $\ket{S}$ is an $\mathbb{N}_1$-succinct state, both of $\ket{A_k}$ and $\ket{B_k}$ are $\mathbb{N}_1$-succinct states.
    Note that, given two $n$-bit strings $x$ and $y$, the amplitude ratios $\braket{x}{A_k}/\braket{y}{A_k}$ and $\braket{x}{B_k}/\braket{y}{B_k}$ can be efficiently calculated since
    \begin{align*}
        \mathcal{Q}'_{A_k}(x,y) &= \frac{\braket{x}{A_k}}{\braket{y}{A_k}} = \frac{\braket{x'}{S}}{\braket{y'}{S}} = \frac{\mathcal{Q}_S(x')}{\mathcal{Q}_S(y')}, \\
        \mathcal{Q}'_{B_k}(x,y) &= \frac{\braket{x}{B_k}}{\braket{y}{B_k}} = \frac{\braket{x''}{S}}{\braket{y''}{S}} = \frac{\mathcal{Q}_S(x'')}{\mathcal{Q}_S(y'')},
    \end{align*}
    where $x',y'$ and $x'',y''$ are the images of $x$ and $y$ under the respective gate actions.
\end{proof}

\hadamardgatesubsetstate*

\begin{proof}
    The action of the Hadamard gate on a computational basis state is 
    \begin{equation*}
        \Gate{Had}_q \ket{x} = \frac{1}{\sqrt{2}}\big(\ket{y} + (-1)^{x[q]}\ket{\bar{y}} \big),
    \end{equation*}
    and note that the Hadamard gate is unitary.
    Given some $q \in [n]$, let $x$ represent an $n$-bit string, then
    \begin{equation*}
        \braket{x}{C_q} = \bra{x}\Gate{Had}_q\ket{S} = \frac{1}{\sqrt{2}}\big(\braket{y}{S} + (-1)^{x[q]}\braket{\bar{y}}{S} \big),
    \end{equation*}
    where $y[j] = \bar{y}[j] = x[j]$ for any $j \neq q$ and then, $y[q] = 0$, $\bar{y}[q] = 1$.
    Note that $y$ and $\bar{y}$ can be easily computed given $q$ and $x$.
    The query algorithm for the state $\ket{C_q}$ must then output 
    \begin{equation*}
        \mathcal{Q}_{C_q}(x) = c_{C_q}\cdot \braket{x}{C_q} = c_{C_q}\cdot \frac{1}{\sqrt{2}}\big(\braket{y}{S} + (-1)^{x[q]}\braket{\bar{y}}{S} \big).
    \end{equation*}
    This can be achieved by using two calls to the query algorithm $\mathcal{Q}_S$ with the appropriate multiplications and additions.
    Specifically, 
    \begin{equation*}
        \mathcal{Q}_{S}(y) + (-1)^{x[q]}\mathcal{Q}_S(\bar{y}) = c_{S}\cdot \braket{y}{S} + (-1)^{x[q]}c_{S}\cdot \braket{\bar{y}}{S} = c_{C_q}\cdot\big(\braket{y}{S} + (-1)^{x[q]}\braket{\bar{y}}{S} \big).
    \end{equation*}
    Notice that we do not require an algebraic encoding of the sign for the value $(-1)^{x[q]}$.
    The resulting output is in $\mathbb{A}^{(1/\sqrt{2})}_1 \times \mathbb{N}_2$.
\end{proof}

\Tgatesubsetstate*

\begin{proof}
    The action of the $T$ gate on a computational basis state is 
    \begin{equation*}
        T_q \ket{x} = \omega^{x[q]}\ket{x},
    \end{equation*}
    for some $q \in [n]$.
    Given some $q \in [n]$, let $x$ represent an $n$-bit string, then
    \begin{equation*}
        \braket{x}{E_q} = \bra{x}T_q\ket{S} = \omega^{-x[q]}\braket{x}{S}.
    \end{equation*}
    The query algorithm for the state $\ket{E_q}$ must then output 
    \begin{equation*}
        \mathcal{Q}_{E_q}(x) = c_{E_q}\cdot \braket{x}{E_q} = c_{E_q}\cdot \omega^{-x[q]}\braket{x}{S}.
    \end{equation*}
    This can be achieved by using a single call to the query algorithm $\mathcal{Q}_S$ and the appropriate multiplication.
    Specifically, 
    \begin{equation*}
        \omega^{-x[q]}\mathcal{Q}_S(x) = c_{S}\cdot \omega^{-x[q]}\braket{x}{S} = c_{E_q}\cdot \omega^{-x[q]}\braket{x}{S}.
    \end{equation*}
    Notice that we now require an algebraic encoding of the power of $\omega$.
    The resulting output is in $\mathbb{N}_{p}\dbbrckt{{\omega}_3}$.
\end{proof}

\Tgateseqsubsetstate*

\begin{proof}
    Let $x$ represent an $n$-bit string, then
    \begin{equation*}
        \braket{x}{F_{\bs{q}}} = \bra{x}\prod_{q \in \bs{q}}T_q\ket{S} = \omega^{-\sum_{q \in \bs{q}} x[q]}\braket{x}{S}.
    \end{equation*}
    The exponent $\sum_{q \in \bs{q}} x[q]$ is at most a summation over $n$ elements, i.e., the Hamming weight of $x$, hence can be computed efficiently.
    Note that the specific calculation is $h \coloneqq -\sum_{q \in \bs{q}} x[q] \mod 8$.
    The query algorithm for the state $\ket{F_{\bs{q}}}$ must then output
    \begin{equation*}
        \mathcal{Q}_{F_{\bs{q}}}(x) = c_{F_{\bs{q}}}\cdot \braket{x}{F_{\bs{q}}} = c_{F_{\bs{q}}}\cdot \omega^{h}\braket{x}{S}.
    \end{equation*}
    This can be achieved by using a single call to the query algorithm $\mathcal{Q}_S$ and the appropriate multiplication.
    Specifically,
    \begin{equation*}
        \omega^{h}\mathcal{Q}_S(x) = c_{S}\cdot \omega^{h}\braket{x}{S} = c_{F_{\bs{q}}}\cdot \omega^{h}\braket{x}{S}.
    \end{equation*}
    Hence, the resulting output is in $\mathbb{A}^{(\omega)}_{3} \times \mathbb{N}_1 \xmapsto{{\rm\cref{rmk:T-gate-succinct-equivalency}}} \mathbb{C}_{1}\dbbrckt{\frac{1}{\sqrt{2}}_1}$.
\end{proof}

\dlma*

\begin{proof}
    Using \cref{cor:hadamard-seq-subset state,cor:CRG-hadamard-subset state,cor:T-T-dagger-gate-seq-subset state} and their results we make the following employing \cref{rmk:T-gate-succinct-equivalency}.
    Therefore, we can say 
    \begin{align*}
        \textnormal{\cref{cor:hadamard-seq-subset state}} ~\implies~ &\mathbb{N}_p \dbbrckt{\frac{1}{\sqrt{2}}_p}, ~~\text{where}~~ p = O(1), \\
        \textnormal{\cref{cor:CRG-hadamard-subset state}} ~\implies~ &\mathbb{N}_p \dbbrckt{\frac{1}{\sqrt{2}}_p}, ~~\text{where}~~ p = O(1), \\
        \textnormal{\cref{cor:T-T-dagger-gate-seq-subset state}} ~\implies~ &\mathbb{C}_{1}\dbbrckt{\frac{1}{\sqrt{2}}_1}.
    \end{align*}
    Hence, the result is an $\mathbb{A}^{(1/\sqrt{2})}_p \times \mathbb{C}_p$-succinct state, where $p = O(1)$.
    Note that the constant $p$ is determined by the number of Hadamard gates in the circuit.
\end{proof}

\historystateexone*

\begin{proof}
    Recall that $\ket{B_k} = R_k \cdots R_1\ket{S}$ for some $k \in [K]$, where $R_k$ is a classically reversible gate and $\ket{S}$ is a subset state on $S \subseteq \B^n$.
    Note that $\ket{k}$ is a computational basis state where $k$ can be expressed as the binary encoding of the decimal $k$, or even as the unary encoding.
    Given an $(n + \abs{{\rm bin}(K)})$-bit string $x = \conc{y}{z}$ where $y \in \B^n$ and $z \in \B^{\abs{{\rm bin}(K)}}$, then 
    \begin{equation*}
        \braket{x}{\eta} = \braket{y}{B_y} \cdot \frac{1}{\sqrt{\abs{K}}}\ind_{I}(z),
    \end{equation*}
    where $I=[{\rm bin}(K)]$ and $\ind_{I}(z)$ is $1$ if and only if $z$ lies in the set $I$.
    We can interpret the $\ket{k}$ component as a subset state such that each term contributes an (equal) amplitude of $1/\sqrt{K}$.
    Using the form of $\ket{B_k}$ we also have that 
    \begin{equation*}
        \braket{x}{\eta} = \frac{1}{\sqrt{\abs{S}}}\ind_{S}(y') \cdot \frac{1}{\sqrt{\abs{K}}}\ind_{I}(z),
    \end{equation*}
    where $y'$ is the image of $y$ under the action of the reversible gates $R_k\cdots R_1$.
    The query algorithm for the state $\ket{\eta}$ must then output
    \begin{equation*}
        \mathcal{Q}_{\eta}(x) = c_{\eta}\cdot \braket{x}{\eta} = c_{\eta}\cdot \frac{1}{\sqrt{\abs{S}}}\ind_{S}(y') \cdot \frac{1}{\sqrt{\abs{K}}}\ind_{I}(z) = \mathcal{Q}_S(y')\cdot\mathcal{Q}_I(z).
    \end{equation*}
    Using \cref{lma:tensor-product-subset state} we therefore conclude that the output is $1$ is and only if $y' \in S$ and $z \in I$.

    Note that, given two $(n+|I|)$-bit strings $x$ and $w$, the amplitude ratio $\braket{x}{\eta}/\braket{w}{\eta}$ can be efficiently calculated since 
    \begin{equation*}
        \mathcal{Q}'_{\eta}(x,w) = \frac{\braket{x}{\eta}}{\braket{w}{\eta}} = \frac{\ind_{S}(y')\ind_{I}(z)}{\ind_{S}(u')\ind_{I}(z)} = \mathcal{Q}_\eta(y')\cdot\mathcal{Q}_\eta(z).
    \end{equation*}
    Recall that the second bit string must not result in a zero amplitude hence will always be $1$.
\end{proof}

\elma*

\begin{proof}
    Recall that $\ket{H_k} = U_k \cdots U_1\ket{S}$ for some $k \in [K]$, where $U_k$ is gate from the set $\{U_k\}_{k\in [K]}$ formed by $O(\poly{n})$ classically reversible gates, $O(n)$ $T$ gates, $O(n)$ $T^\dagger$ gates and $O(1)$ Hadamard gates.
    There is an $O(\poly{n})$ size bit string $\bsr$ that represents the information of the gates.
    Given an $(n + \abs{{\rm bin}(K)})$-bit string $x = \conc{y}{z}$ where $y \in \B^n$ and $z \in \B^{\abs{{\rm bin}(K)}}$, then
    \begin{equation*}
        \braket{x}{\eta} = \braket{y}{H_k} \cdot \frac{1}{\sqrt{\abs{K}}}\ind_{I}(z).
    \end{equation*}
    The first term then follows as 
    \begin{equation*}
        \braket{y}{H_k} = \bra{y}{U_k \cdots U_1\ket{S}} = \braket{\psi_y}{S}.
    \end{equation*}
    The superposition state $\ket{\psi_y}$ is formed via the action of the unitary gates on the computational basis state $\ket{y}$.
    The case where for a given $l$ the sequence of unitaries $U_k \cdots U_1$ is entirely classical, we resort to the result of \cref{lma:classical-gates-subset state}.
    From \cref{rmk:subset-more-general} we see that this gives a $\mathbb{Q}^+_{p(n)}\dbbrckt{\sqrt{\cdot}_1}$-succinct state.
    In the more general scenario where in which $\ket{\psi_y}$ is truly a superposition state, we must employ the information stored in $\bsr$ to track the action of the gates.
    Moreover, the output of $U_1\cdots U_k\ket{y}$ can be efficiently calculated.
    Let $a,b,c$ represent the total number of $T$, $T^\dagger$ and Hadamard gates respectively.
    Since $c = O(1)$, the largest number of amplitudes needed to be combined is $2^c = O(1)$.
    The effect of each $T$ and $T^\dagger$ gate is tracked by $\bsr$.
    For some $l$ we obtain the general form,
    \begin{equation*}
        \bra{\psi_y} = \frac{1}{\sqrt{2^{c'}}} \sum_{i \in [2^{c'}]} \omega^{\textsl{g}(y_i)}(-1)^{\textsl{f}(y_i)} \bra{y_i}.
    \end{equation*}
    We have denoted $y_i$ as the image of some bit string (related to $y$) under the action of a sequence of gates.
    The functions $\textsl{g}(y_i)$ and $\textsl{f}(y_i)$ represent the phase power and sign power, respectively, for the image bit string $y_i$.
    Note also that $c' \leq c$.
    Then, 
    \begin{equation*}
        \braket{x}{\eta} = \left(\sum_{i} m(y_i)\ind_{S}(y_i)\right) \cdot \frac{1}{\sqrt{\abs{K}}}\ind_{I}(z).
    \end{equation*}
    It is clear that each component of the superposition state $\ket{\psi_y}$ is a succinct-state.
    Furthermore, by appropriate multiplications and additions of components, the output of $\braket{y}{H_k}$ can be efficiently calculated.
    The query algorithm for the state $\ket{\eta}$ must then output
    \begin{equation*}
        \mathcal{Q}_{\eta}(x) = c_\eta \cdot \braket{x}{\eta}.
    \end{equation*}
    This is achieved by appropriate standard arithmetic operations just outlined.
    An exact representation of the amplitude thus lies in 
    \begin{equation*}
        \big(\mathbb{A}_1^{(\sqrt{\cdot})} \times \mathbb{Q}^+_{c' + 3} \big)\times \big(\mathbb{A}_1^{({\rm sgn})} \times \mathbb{C}_2\big) \times \big(\mathbb{A}_1^{(\sqrt{\cdot})} \times \mathbb{Q}^+_{\log_2\abs{S}}\big) \times \big(\mathbb{A}_1^{(\sqrt{\cdot})} \times \mathbb{Q}^+_{\log_2\abs{K}}\big),
    \end{equation*}
    where the bracketed terms are one quantity in isolation.
    This format is a little messy; with some rearrangement and classical computation, we can massage the set to 
    \begin{equation*}
        \mathbb{A}_1^{(\sqrt{\cdot})} \times \mathbb{C}_r.
    \end{equation*}
    We interpret the action of the square root on the individual integers making up the complex number, i.e., $\sqrt{a} + {\rm i}\sqrt{b}$ and not $\sqrt{a + {\rm i} b}$.
    Notice that simple arithmetic operations can square the values in $\mathbb{C}_2 \mapsto \mathbb{C}_3$.
    Then $\mathbb{Q}^+_{c' + 3} \times \mathbb{C}_3 \times \mathbb{Q}^+_{\log_2(\abs{S})} \times \mathbb{Q}^+_{\log_2(\abs{K})}\mapsto \mathbb{C}_r$, where $r = \poly{n}$.
\end{proof}

\alma*

\begin{proof}
    Defining the states in the computational basis, we have
    \begin{align*}
        \ket{\psi} &= \sum_{i \in \B^n} \alpha(i)\ket{i}, \\
        \ket{\phi} &= \sum_{j \in \B^m} \beta(j)\ket{j}.
    \end{align*}
    The tensor product of these states is
    \begin{equation*}
        \ket{\psi}\ket{\phi} = \sum_{i \in \B^n, j \in \B^m} \alpha(i)\beta(j)\ket{i}\ket{j} = \sum_{k \in \B^{n+m}} \gamma(k)\ket{k},
    \end{equation*}
    where $\gamma(k) = \gamma(\conc{i}{j}) = \alpha(i)\beta(j)$.
    Note that the order of the tensor product is important since for the $(n+m)$-bit strings $k$, the first $n$ bits correspond to the first state and the last $m$ bits correspond to the second state.
    The query algorithm for the state $\ket{\psi}\ket{\phi}$ must then output
    \begin{equation*}
        \mathcal{Q}_{\psi\phi}(k) = c_{\psi\phi}\cdot \gamma(k) = c_{\psi\phi}\cdot \alpha(i)\beta(j) = \big(c_{\psi}\cdot \alpha(i) \big) \big(c_{\phi}\cdot \beta(j) \big) = \mathcal{Q}_{\psi}(i)\mathcal{Q}_{\phi}(j),
    \end{equation*}
    i.e., the query algorithm for the tensor product state is the multiplication of the query algorithms for the individual states respective of the input bit strings.
    Since both $\mathcal{Q}_\psi$ and $\mathcal{Q}_\phi$ are classically efficient algorithms, then the query algorithm for the tensor product state is also classically efficient.
    Furthermore, in the cases where $\mathcal{Q}_\psi$ and $\mathcal{Q}_\phi$ should return zero will be reciprocated in $\mathcal{Q}_{\psi\phi}$.

    Recall that $0 < c_\psi \leq 2^{p(n)}$ and $0 < c_\phi \leq 2^{q(m)}$, then $0 < c_{\psi\phi} = c_\psi c_\phi \leq 2^{p(n)+q(m)}$; this satisfies the definition of a $\mathbb{C}_{r(n,m)}$-succinct state for $\ket{\psi}\ket{\phi}$ since $p(n)+q(m) \leq r(n,m)$ for some polynomial $r$.
    Moreover, let $s = \max\{n,m\}$ then $r(s) = \max\{p(s),q(s)\}$, then the tensor product state is a $\mathbb{C}_{2r(s) + 1}$-succinct state.

    Note also that 
    \begin{equation*}
        \mathcal{Q}'_{\psi\phi}(k,l) = \frac{\gamma(k)}{\gamma(l)} = \frac{\gamma(i_k || j_k)}{\gamma(i_k || j_k)} = \frac{\alpha(i_k)\beta(j_k)}{\alpha(i_k)\beta(j_k)} = \mathcal{Q}'_\psi(i_k,i_k)\cdot\mathcal{Q}'_\phi(j_k,j_k),
    \end{equation*}
    where $i_k,i_k \in \B^n$ and $j_k,j_k \in \B^m$.
    This is the same as the product of the query algorithms for the individual states and hence requires $32(2r + 1) + 12$ bits to represent exactly.
\end{proof}

\blma*

\begin{proof}
    Isolating to the state $\ket{\varphi_1}$ we see that this is a superposition of two real-valued states, each in a tensor product with a subset state.
    If each of $\ket{\phi_R}$ and $\ket{\phi_I}$ are $\mathbb{Q}_{p(n)}$-succinct states then $\ket{\phi_R}\ket{0}$ and $\ket{\phi_I}\ket{1}$ are also $\mathbb{Q}_{p(n)}$-succinct states (cf. \cref{lma:tensor-product-C-succinct}).

    We first check the normalisation of $\ket{\varphi_1}$:
    \begin{align*}
        \braket{\varphi_1} &= \bigg(\sum_{j \in \B^n} R(j) \bra{\conc{j}{0}} + I(j) \bra{\conc{j}{1}} \bigg) \bigg(\sum_{k \in \B^n} R(k) \ket{\conc{k}{0}} + I(k) \ket{\conc{k}{1}}\bigg)\\
            &= \bigg(\sum_{j \in \B^n} R(j)^2 + I(j)^2 \bigg) \\
            &= \bigg(\sum_{j \in \B^n} |R(j) + iI(j)|^2 \bigg) \\
            &= 1.
    \end{align*}
    It therefore suffices to show that the query algorithm for $\ket{\varphi_1}$ is classically efficient.
    To this end we introduce notation --- let the output string of a classical query algorithm $\mathcal{Q}_A(i)$ be some bit string $\bin{q}_{A,i}$.
    Note that an input bit string to $\mathcal{Q}_{\varphi_1}$ is of the form $l = \conc{j}{b}$ where $j \in \B^n$ and $b \in \B$.
    The query algorithm for $\ket{\varphi_1}$ is then
    \begin{equation*}
        \mathcal{Q}_{\varphi_1}(l = \conc{j}{b}) = \begin{cases}
            \bin{q}_{\phi,j}[1:p(n)] \eqqcolon \mathcal{Q}_{\phi}(j)\Bigr|_{b=0}, & \text{if}~b = 0, \\
            \bin{q}_{\phi,j}[p(n)+1 : 2p(n)] \eqqcolon \mathcal{Q}_{\phi}(j)\Bigr|_{b=1}, & \text{if}~b = 1.
        \end{cases}
    \end{equation*}
    Specifically, for any input $l$ we query $\mathcal{Q}_{\phi}$ using the first $n$ bits of $l$ and then output one of the halves of the query output, conditioned on the last bit of $l$.
    Recall that $\mathcal{Q}_{\phi}$ outputs a bit string of length $2p(n)$ since amplitude in $\ket{\phi}$ are complex values.
    The latter half of these bit strings represent the imaginary part of the amplitudes.
    This is clearly a classically efficient algorithm.
    Furthermore, the constant $c_{\varphi_1} = c_{\phi}$.
    A similar argument can be made for $\ket{\varphi_2}$ making note that when $b=1$ the output should carry a minus sign.
    Note the amplitude ratio $\mathcal{Q}'_{\varphi_1}(k,l)$ can be efficiently calculated using the conditional output of $\mathcal{Q}_{\phi}$ and appropriate arithmetic operations.
\end{proof}

\athrm*

\begin{proof}
    Let $F_{|x|}$ be Arthur's \clsb{MA}{q} verification circuit equipped with an $O(\poly{n})$-bit string $\bsr$ representing the information of the gate sequence.
    Let the input to the circuit be an $N = n + w + m + p$ qubit register comprised of four parts: the input state $\ket{x}$ of $n$ qubits, the proof state $\ket{\chi}$ of $w$ qubits, the \emph{ancilla} register of $m$ qubits initialised to $\ket{0}$ and the \emph{coin} register of $p$ qubits initialised to $\ket{+}$.
    Let $F_{|x|}$ comprise a sequence of $K$ Toffoli gates denoted as $R_K, \dots, R_1$.

    Define a Hamiltonian $H = H_{\text{in}} + H_{{\rm out}} + H_{\text{prop}} + H_{\text{clock}}$ acting on a register of $K$ \emph{clock} qubits labelled as $c_1, \dots, c_K$ and the $N$ qubit input register.
    Let the output measured qubit be denoted $q$; for this instance, Arthur can measure using only the $Z$-basis.
    Each Hamiltonian term is defined to be a penalising Hamiltonian and must be stoquastic.
    
    \begin{align*}
        H_{\text{in}} &= \bigg(\sum_{j=1}^n I - \ketbra{x_j} + \sum_{j=1}^{m} \ketbra{1}_{\emph{anc},j} + \sum_{i=1}^{p} \ketbra{-}_{\emph{coin},i}\bigg)\otimes\ketbra{0}_{c_1}, \\
        H_{{\rm out}} &= \ketbra{0}_q \otimes \ketbra{1}_{c_T}, \\
        H_{\text{clock}} &= K^{12}\sum_{1\leq i < j \leq K} \ketbra{01}_{c_i,c_j},\\
        H_{\text{prop}} &= \frac{1}{2}\sum_{t=1}^{K} H_{\text{prop}}(t).
    \end{align*}

    We define the propagation Hamiltonian terms the following way:
    \begin{align*}
        H_{\text{prop}}(1) &= \ketbra{10}_{c_1,c_2} + \ketbra{0}_{c_1} - R_1\otimes(\ketbra{1}{0}_{c_1} + \ketbra{0}{1}_{c_1}), \\
        H_{\text{prop}}(t) &= \ketbra{10}_{c_t,c_{t+1}} + \ketbra{10}_{c_{t-1},c_t} - R_t\otimes(\ketbra{1}{0}_{c_t} + \ketbra{0}{1}_{c_t}), \quad 1<t<K \\
        H_{\text{prop}}(K) &= \ketbra{1}_{c_K} + \ketbra{10}_{c_{K-1},c_K} - R_K\otimes(\ketbra{1}{0}_{c_K} + \ketbra{0}{1}_{c_K}).
    \end{align*}

    Note that $H_{\text{in}}$, $H_{{\rm out}}$ and $H_{\text{clock}}$ are all $2$-local Hamiltonians.
    The terms $H_{\text{prop}}(t)$ are $4$-local $\forall\, t\in[K]$.
    It is trivial to show each Hamiltonian term is stoquastic.
    Notice that $\ketbra{-} = \frac{1}{2}(I - X)$, $\ketbra{1} = \frac{1}{2}(I + Z)$ and $H_{{\rm out}}$, $H_{\text{clock}}$ are diagonal; hence $H_{\text{in}}$, $H_{{\rm out}}$ and $H_{\text{clock}}$ are all $2$-local \emph{stoquastic} Hamiltonians.
    The terms $R_t\otimes(\dots)$ in $H_{\text{prop}}(t)$ will have off-diagonal elements that are strictly positive.
    Therefore, each $H_{\text{prop}}(t)$ term is stoquastic. 
    Moreover, each element of the Hamiltonian is efficiently computable since the gates $R_t$ are classically reversible and the information of the gate sequence is stored in $\bsr$.
    That is, the elements lie in the set $\mathbb{Q}^+_{p(n)}$ for some polynomial $p$.

    The history state $\ket{\eta}$ is then defined as
    \begin{equation*}
        \ket{\eta(x,\chi)} = \frac{1}{\sqrt{K+1}}\sum_{t=0}^{K} \ket{\varphi_t}\ket{t}, \quad \ket{\varphi_t} = R_t \cdots R_1\ket{x,\chi,0^m,+^p}.
    \end{equation*}
    Recall from \cref{rmk:history-state-ground} that the history state $\ket{\eta(x,\chi)}$ is the unique ground state of $H$ with eigenvalue zero if and only if the proof state $\ket{\chi}$ is accepted with probability $1$ by the verification circuit $F_{|x|}$.
    Therefore by \cref{lma:history-state-ex-1}, \cref{rmk:subset-more-general} and \cref{lma:CEHSS-equivalence} we have that $\ket{\eta(x,\chi)}$ is a $\mathbb{Q}^+_{p(n)}\dbbrckt{\sqrt{\cdot}_1}$-succinct state.
    
    To conclude, we leverage the original arguments from \refcite{kempe2003local} to show that in the \textsc{yes} case, there exists a proof state such that the Hamiltonian $H$ has ground state energy $0$ (since the verification circuit accepts with probability $1$).
    In the \textsc{no} case, all eigenvalues are at least $c/K^3$ for some constant $c$.
\end{proof}

%% file: appendix_c.tex
\section{Local Stoquastic Hamiltonians with Easy Witness Ground States}\label{app:stoquastic-hamiltonians-easy-witness}

\begin{definition}[Stoquastic Verification Circuit]
    A stoquastic verification circuit is a tuple $S_n = (n,w,m,p,U)$ where $n$ is the number of input qubits, $w$ is the number of proof qubits, $m$ is the number of ancillae initialised in the $\ket{0}$ state and $p$ is the number of ancillae initialised in the $\ket{+}$ state.
    The circuit $U$ is a quantum circuit on $M\coloneqq n + w + m + p$ qubits, comprised of $K = O(\poly{n})$ gates from the set $\{X, \Gate{Cnot}, \Gate{Toffoli}\}$.
    The acceptance probability of a stoquastic verification circuit $S_n$, given some input string $x\in \varSigma^n$ and a proof state $\ket{\xi} \in \mathbb{C}^{2^w}$ is defined as:
    \begin{equation*}
        \Pr\left[S_n(x,\ket{\xi})\right]= \bra{\phi}U^\dagger \Pi_{\text{out}} U \ket{\phi},
    \end{equation*}
    where $\ket{\phi} = \ket{x,\xi,0^{m},+^{p}}$ and $\Pi_{\text{out}} = \ketbra{+}_1$ is a projector onto the output qubit.
\end{definition}

Note that $w,m,p = O(\poly{n})$.

\begin{definition}[\cl{StoqMA}($\alpha$,$\beta$)]\label{def:StoqMA}
    A promise problem $L = (L_{\textsc{yes}}, L_{\textsc{no}})$ belongs to the class \cl{StoqMA}($\alpha$,$\beta$) if there exists a polynomial-time generated stoquastic circuit family $\mathcal{S} = \{S_n : n \in \mathbb{N}\}$, where each stoquastic circuit $S_n$ acts on $n + w+m + p$ input qubits and produces one output qubit, such that:
    \begin{itemize}
        \item[] \textbf{Completeness}: For all $x\in L_{\textsc{yes}}$, $\exists \ket{\xi}\in(\mathbb{C}^2)^{\otimes w}$, such that, $ \Pr\left[S_{|x|}(x,\ket{\xi})=\mathtt{1} \right] \geq \alpha(|x|)$
        \item[] \textbf{Soundness}: For all $x\in L_{\textsc{no}}$, $\forall\ket{\xi}\in(\mathbb{C}^2)^{\otimes w}$, then, $ \Pr\left[S_{|x|}(x,\ket{\xi})=\mathtt{1} \right] \leq \beta(|x|)$
    \end{itemize}
    The term $\alpha$ refers to the completeness parameter and $\beta$ the soundness parameter, where $1/2 \leq \beta(|x|) < \alpha(|x|) \leq 1$ and satisfying $\alpha-\beta\geq\frac{1}{\poly{|x|}}$.
\end{definition}

Note that \cl{StoqMA} does not admit amplification (to the best of our knowledge.\footnote{There is one exception to this rule requiring a polynomial number of copies of the proof state for soundness amplification~\cite{liu2021stoqma}.})
A work from Liu~\cite{liu2021stoqma} discusses a modification of \cl{StoqMA} called \cl{eStoqMA}; the ``\cl{e}'' represents the addition of an ``easy witness''.
The motivation for considering an easy witness stems from the associated lemma in \refcite{impagliazzo2002search}.
We only provide a brief overview of the class \cl{eStoqMA} here (cf.~\cite[Definition 3.1]{liu2021stoqma}).
Essentially, the class \cl{eStoqMA} is the same as \cl{StoqMA} but with the addition that in the \textsc{yes} case there exists an $n$-qubit non-negative witness state
\begin{equation*}
    \ket{\xi} \coloneqq \sum_{j \in \B^n} \sqrt{D_{\xi}(j)}\ket{j},
\end{equation*}
where there is an efficient classical algorithm $\mathcal{Q}_{\xi}$ that outputs the ratio $D_{\xi}(\conc{0}{k})/D_{\xi}(\conc{1}{k})$ for a $(n-1)$-bit string $k$.

The dual-access model described in \refcite{liu2021stoqma} is adapted from \refcite{canonne2014testing}.
We briefly describe the model here.
Let $D$ represent a fixed distribution over $\{0,\dots,2^n -1\}$.
Having \emph{sample access} to $D$ implies there exists a query algorithm (oracle) $\mathcal{S}_D$ that returns an element $j \in \B^n$ with probability $D(j)$, independent of prior calls.
\emph{Query access} to $D$ implies the existence of a query algorithm (oracle) $\mathcal{Q}_D$ that, given an input $j \in \B^{n-1}$, returns the quotient $D(\conc{0}{j}) / D(\conc{1}{j})$.

It follows that a subset state is a natural easy witness.
If we assume that $\ket{\xi}$ is normalised then the value of $D_{\xi}(j) = 1/\abs{{\rm supp}(\ket{\xi})}$ for all $j \in {\rm supp}(\ket{\xi})$, i.e., $\ket{\xi}$ is a subset state.
Clearly, the query algorithm will only output $1$ if for a given $k\in\B^{n-1}$, both $\conc{0}{k}$ and $\conc{1}{k}$ are in the support of $\ket{\xi}$.

As it was already shown that \cl{eStoqMA} is equivalent to \cl{MA}, it suffices to conclude that the \clw{MA}{hardness} proof of \refcite{bravyi2006complexity} holds in this setting.
Specifically, we must ensure the history state is an easy witness.
Naturally this follows from \cref{lma:history-state-ex-1} and \cref{lma:CEHSS-equivalence}. 

\begin{theorem}[\cite{liu2021stoqma}]
    The \sc{$6$-Local Stoquastic Hamiltonian with an Easy Witness Ground State} problem is \clw{MA}{complete}.
\end{theorem}

\begin{corollary}
    The \sc{$4$-Local Stoquastic Hamiltonian with an Easy Witness Ground State} problem is \clw{MA}{complete}.
\end{corollary}

%% file: appendix_d.tex
\section{Toffoli Gate Decomposition}\label{app:toffoli-gate-decomposition}
In this appendix we discuss the exploitation of the structure of STEC (\cl{StMA}\textsubscript{q} circuits.\footnote{``\textbf{St}ructured \textbf{{MA}\textsubscript{q}}''})
Consider a classically reversible circuit comprised of $K = O(\poly{n})$ \Gate{Toffoli} gates.
The exact decomposition of the \Gate{Toffoli} gate is well-known and results in a sequence of \Gate{Cnot}, \Gate{Had}, and $T$ gates.
Since $\bsr$ encodes the information of the circuit, we can generate a new bit string $\bsr'$ that encodes the decomposition of the \Gate{Toffoli} gate.
This new bit string will be of size $O(\poly{n})$ and will be used to track the action of the $T$ gates and Hadamard gates.
Moreover, let $\bsr'$ follow the decomposition exactly in the sense that the gates $1$ to $15$ correspond to the first \Gate{Toffoli} gate, $16$ to $30$ correspond to the second \Gate{Toffoli} gate, and so on.
The decomposition of the \Gate{Toffoli} gate is as follows:
\begin{align*}
    \Gate{Toffoli}[a,b;c] &= T[a] \, \Gate{Cnot}[a;b] \, T^\dagger[b] \, \Gate{Cnot}[a;b] \, T[b] \notag \\
    &\qquad H[c] \, \Gate{Cnot}[b;c] \, T^\dagger[c] \, \Gate{Cnot}[a;c] \, T[c] \\ 
    &\qquad \Gate{Cnot}[b;c] \, T^\dagger[c] \, \Gate{Cnot}[a;c] \, T[c] \, H[c].\notag
\end{align*}
To study the effect of this decomposition on the amplitudes of a given state we consider each gate in turn.
For ease of analysis we let $a=1$, $b=2$ and $c=3$.
For a given input $n$-bit string $x$, the action on some specific state $\ket{\varphi}$ follows as:
\begin{align*}
    \bra{x}T[1]\ket{\varphi} &= \omega^{x[1]}\braket{x}{\varphi}, \\[0.3cm]
    \bra{x}T[1]\Gate{Cnot}[1;2]\ket{\varphi} &= \omega^{x[1]}\braket{x'}{\varphi}, \\[0.25cm]
        &\hspace{-6.5cm}\text{where}~x' = \Gate{Cnot}[1;2]x\\[0.3cm]
    \bra{x}T[1]\Gate{Cnot}[1;2]T^\dagger[2]\ket{\varphi} &= \omega^{x[1]-x'[2]}\braket{x'}{\varphi}, \\[0.3cm]
    \bra{x}T[1]\Gate{Cnot}[1;2]T^\dagger[2]\Gate{Cnot}[1;2]\ket{\varphi} &= \omega^{x[1]-x'[2]}\braket{x''}{\varphi}, \\[0.25cm]
        &\hspace{-6.5cm}\text{where}~x'' = \Gate{Cnot}[1;2]x'\\[0.3cm]
    \bra{x}T[1]\Gate{Cnot}[1;2]T^\dagger[2]\Gate{Cnot}[1;2]T[2]\ket{\varphi} &= \omega^{x[1]-x'[2]+x''[2]}\braket{x''}{\varphi}, \\[0.3cm]
    \bra{x}T[1]\cdots H[3]\ket{\varphi} &= \frac{1}{\sqrt{2}}\omega^{x[1]-x'[2]+x''[2]}(\braket{y}{\varphi} + (-1)^{x''[3]}\braket{\bar{y}}{\varphi}), \\[0.25cm]
        &\hspace{-6.5cm}\text{where}~\bar{y} = \dots || x''[3] = 1 || \dots \\[0.3cm]
    \bra{x}T[1]\cdots \Gate{Cnot}[2;3]\ket{\varphi} &= \frac{1}{\sqrt{2}}\omega^{x[1]-x'[2]+x''[2]}(\braket{y'}{\varphi} + (-1)^{x''[3]}\braket{\bar{y}'}{\varphi}), \\[0.25cm]
        &\hspace{-6.5cm}\text{where}~\bar{y}' = \Gate{Cnot}[2;3] \bar{y} \\[0.3cm]
    \bra{x}T[1]\cdots T^\dagger[3]\ket{\varphi} &= \frac{1}{\sqrt{2}}\omega^{x[1]-x'[2]+x''[2]}(\omega^{-y'_3}\braket{y'}{\varphi} + (-1)^{x''[3]}\omega^{-\bar{y}'_3}\braket{\bar{y}'}{\varphi}), \\[0.3cm]
    \bra{x}T[1]\cdots \Gate{Cnot}[1;3]\ket{\varphi} &= \frac{1}{\sqrt{2}}\omega^{x[1]-x'[2]+x''[2]}(\omega^{-y'_3}\braket{y''}{\varphi} + (-1)^{x''[3]}\omega^{-\bar{y}'_3}\braket{\bar{y}''}{\varphi}), \\[0.3cm]
\end{align*}
Continued\dots
\newpage
\begin{align*}
    \bra{x}T[1]\cdots T[3]\ket{\varphi} &= \frac{1}{\sqrt{2}}\omega^{x[1]-x'[2]+x''[2]}\big(\omega^{-y'[3]+y''[3]}\braket{y''}{\varphi}
        + (-1)^{x''[3]}\omega^{-\bar{y}'[3] + \bar{y}''[3]}\braket{\bar{y}''}{\varphi}\big), \\[0.3cm]
    \bra{x}T[1]\dots \Gate{Cnot}[2;3]\ket{\varphi} &= \frac{1}{\sqrt{2}}\omega^{x[1]-x'[2]+x''[2]}\big(\omega^{-y'[3]+y''[3]}\braket{y'''}{\varphi} 
        +(-1)^{x''[3]}\omega^{-\bar{y}'[3] + \bar{y}''[3]}\braket{\bar{y}'''}{\varphi}\big), \\[0.3cm]
    \bra{x}T[1]\cdots T^\dagger[3]\ket{\varphi} &= \frac{1}{\sqrt{2}}\omega^{x[1]-x'[2]+x''[2]}\big(\omega^{-y'[3]+y''[3] - y'''[3]}\braket{y'''}{\varphi}
        + (-1)^{x''[3]}\omega^{-\bar{y}'[3] + \bar{y}''[3] - \bar{y}'''[3]}\braket{\bar{y}'''}{\varphi}\big), \\[0.3cm]
    \bra{x}T[1]\cdots \Gate{Cnot}[1;3]\ket{\varphi} &= \frac{1}{\sqrt{2}}\omega^{x[1]-x'[2]+x''[2]}\big(\omega^{-y'[3]+y''[3] - y'''[3]}\braket{y''''}{\varphi} \\[0.25cm]
        &\quad\qquad+ (-1)^{x''[3]}\omega^{-\bar{y}'[3] + \bar{y}''[3] - \bar{y}'''[3]}\braket{\bar{y}''''}{\varphi}\big), \\[0.3cm]
    \bra{x}T[1]\cdots T[3]\ket{\varphi} &= \frac{1}{\sqrt{2}}\omega^{x[1]-x'[2]+x''[2]}\big(\omega^{-y'[3]+y''[3] - y'''[3] + y''''[3]}\braket{y''''}{\varphi} \\[0.25cm]
        &\quad\qquad+ (-1)^{x''[3]}\omega^{-\bar{y}'[3] + \bar{y}''[3] - \bar{y}'''[3] + \bar{y}''''[3]}\braket{\bar{y}''''}{\varphi}\big), \\[0.3cm]
    \bra{x}T[1]\cdots H[3]\ket{\varphi} &= \frac{1}{2}\omega^{x[1]-x'[2]+x''[2]}\big(\omega^{-y'[3]+y''[3] - y'''[3] + y''''[3]}(\braket{z}{\varphi} + (-1)^{y''''[3]}\braket{\bar{z}}{\varphi}) \\[0.25cm]
        &\quad\qquad+ (-1)^{x''[3]}\omega^{-\bar{y}'[3] + \bar{y}''[3] - \bar{y}'''[3] + \bar{y}''''[3]}(\braket{w}{\varphi} + (-1)^{\bar{y}''''[3]}\braket{\bar{w}}{\varphi})\big),    
\end{align*}
Note that:
\begin{align*}
    x' &= \Gate{Cnot}[1;2]x, & y &= \dots \mathbin\Vert x''_3 = 0 \mathbin\Vert \dots, & y' &= \Gate{Cnot}[2;3]y,  & y'' &= \Gate{Cnot}[1;3]y',\\
    x'' &= \Gate{Cnot}[1;2]x',& \bar{y} &= \dots \mathbin\Vert x''_3 = 1 \mathbin\Vert \dots, & \bar{y}' &= \Gate{Cnot}[2;3]\bar{y}, & \bar{y}''& = \Gate{Cnot}[1;3]\bar{y}',\\[0.5cm]
    y''' &= \Gate{Cnot}[2;3]y'', & y'''' &= \Gate{Cnot}[1;3]y''', & z &= \dots \mathbin\Vert y''''_3 = 0 \mathbin\Vert \dots, & w &= \dots \mathbin\Vert \bar{y}''''_3 = 0 \mathbin\Vert \dots, \\
    \bar{y}''' &= \Gate{Cnot}[2;3]\bar{y}'', & \bar{y}'''' &= \Gate{Cnot}[1;3]\bar{y}''', & \bar{z} &= \dots \mathbin\Vert y''''_3 = 1 \mathbin\Vert \dots, & \bar{w} &= \dots \mathbin\Vert \bar{y}''''_3 = 1 \mathbin\Vert \dots.  
\end{align*}
Some binary values here have simpler representations, for example $x'' \equiv x$ and $y'''' \equiv y$, but this does not help reduce to bulk to the final state much.
The final state should reduce to only considering $\bra{x}\Gate{Toffoli}[1,2;3]\ket{\varphi}$.
Indeed, it can be verified that this holds.
For clarity, the final state is given by:
\begin{align*}
    \bra{x}T[1]\cdots H[3]\ket{\varphi} &= \frac{1}{2}\omega^{x[1]-x'[2]+x''[2]-y'[3]+y''[3] - y'''[3] + y''''[3]}\braket{z}{\varphi} \\
    &\qquad+\frac{1}{2}\omega^{x[1]-x'[2]+x''[2]-y'[3]+y''[3] - y'''[3] + y''''[3]}(-1)^{y''''[3]}\braket{\bar{z}}{\varphi} \\
    &\qquad+\frac{1}{2}\omega^{x[1]-x'[2]+x''[2] -\bar{y}'[3] + \bar{y}''[3] - \bar{y}'''[3] + \bar{y}''''[3]}(-1)^{x''[3]}\braket{w}{\varphi} \\ 
    &\qquad+ \frac{1}{2}\omega^{x[1]-x'[2]+x''[2] -\bar{y}'[3] + \bar{y}''[3] - \bar{y}'''[3] + \bar{y}''''[3]}(-1)^{x''[3] + \bar{y}''''[3]}\braket{\bar{w}}{\varphi} \\
    &= \bra{x}\Gate{Toffoli}[1,2;3]\ket{\varphi}.
\end{align*}

Note that only two combinations of $x[1]$, $x[2]$, and $x[3]$ are affected by the action of the \Gate{Toffoli} gate.
The purpose of this analysis is to show that when working with Structured Toffoli-Equivalent Circuits (STEC), we can use the string $\bsr'$ to efficiently track the cumulative action of each gate in the circuit.
This is crucial because, although there may be a polynomial number of Hadamard gates, the tracking allows us to avoid calculating an exponential number of amplitudes.
The key reason for this efficiency lies in the ability to organize the computation by ``blocks'', as described above.
By identifying the relevant block for any given step (time index), the corresponding amplitudes can be computed in polynomial time.

To illustrate this more concretely, consider calculating the amplitude $\bra{x}U_j \cdots U_1 \ket{\varphi}$, where each $U_j$ is one of the gates discussed earlier.
The first step is to compute $j \mod 15$, which helps partition the sequence of gates into manageable blocks of size $15$.
Suppose $j = 15k + l$, where $0 < l < 15$.
In this case, $\bsr'$ can be used to trace the action of the gate sequence $U_j \cdots U_1 = U_k U_{k-1} \cdots \Gate{Toffoli}_k \Gate{Toffoli}_{k-1} \cdots \Gate{Toffoli}_1$.
The reason this modular reduction is important is that it helps us separate the action of the first $l$ gates (which are part of the Toffoli decomposition) from the subsequent $k$ Toffoli gates.
This separation allows us to compute the effect of the entire gate sequence on the bit string $x$ efficiently by first applying the $l$ gates, followed by the $k$ Toffoli gates.
Since each of these steps can be performed in polynomial time, this approach avoids exponential complexity while still tracking the complete action of the circuit.

As a final comment on this decomposition and STEC circuits, we note that the \clw{MA}{complete}ness proof resulting in a complex Hamiltonian does not impede on any prior analysis.
This is because in the context of the standard \sc{Local Hamiltonian} problem we would not be able to show containment of such Hamiltonians in \cl{MA} since without the assumption that the ground state is succinct, we have no known protocol.
This result is similar to basic arguments showing that the \sc{Local Hamiltonian} problem is at least \cl{NP}-hard because classical Hamiltonians are a subset of quantum Hamiltonians.

%% file: appendix_e.tex
\section{Pre-idled Quantum Verifier Scenario}\label{app:pre-idled-quantum-verifier}
Pre-idling the verification circuit consists of padding the start of the original gate sequence with a series of identity gates.
Assuming that we only add a polynomial number of such, denoted $N$, the total gate count becomes $N+K$ (where $K$ is the number of gates in the original circuit).
The purpose of this padding is to address Solution \ref{sol:3} discussed in \cref{sec:discussion-proof-error}.
The consequence of this padding is a change in the spectral gap of the resulting Hamiltonian.
It follows that the spectral gap is bounded as $\Omega(1/(\# {\rm gates})^3)$.
Since it is always so that $K = O(\poly{n})$, the spectral gap is still inverse-polynomial even with the pre-idling.
Our requirement for Solution \ref{sol:3} (and also Solution \ref{sol:1}) is that $N+K+1$ is a square number.
Let $x_i$ and $x_{i+1}$ be two consecutive square numbers such that $x_i < K+1 < x_{i+1}$.
Then trivially $N = x_{i+1} - K - 1$ (or $N = x_{i+1+c} - K - 1$ for some $c \in \mathbb{N}$) is a valid choice.
Clearly $N = O(\poly{n})$ and hence we confirm the spectral gap is still inverse-polynomial.

A polynomial increase to the number of clock qubits is then required in the subsequent Feynman-Kitaev circuit-to-Hamiltonian construction.
Yet, on the positive side, the uniform amplitude of the history state $1/\sqrt{N+K+1}$ is now a rational number.
Furthermore, due to \cref{lma:CEHSS-equivalence}, the amplitude of the history state expressed as a subset state is also rational since the number of $+$-ancillae qubits can be assumed to be even.
The consequence of this is that the original proof arguments of \refcite{jiang2025local} regarding the \clw{MA}{hardness} now hold.

This same trick of pre-idling or allowing for an even number of gates does not translate well to the STEC.
The reason for this is that the resulting history state is not a subset state.
We also note the difference between the present problem and the \sc{Guided Local Hamiltonian} problem considered in \refcite{cade2023improved}.
In the latter, the verifier has an additional input in the form of a guiding state that has promised overlap with the true ground state of the Hamiltonian.
The specific form of the guiding state is a `semi-classical encoded state'~\cite[Definition 3]{cade2023improved}; these are similar to ideas presented here but have a condition that the subset must have polynomial size.
Moreover, the verifier has the power to sample from the guiding state efficiently.\footnote{Given a description of said state.}
For the present problem, the verifier is given the classical circuit that has query access to the ground state's amplitude information.
This is in clear contrast as one problem uses \emph{guiding state} information and the other uses \emph{ground state} information.

%% file: appendix_f.tex
\section{Local Hamiltonians on Spatially Sparse Graphs}\label{app:local-hamiltonians-spatially-sparse-graphs}
Proving the \clw{MA}{complete}ness of the \sc{Succinct State Local Hamiltonian} problem on spatially sparse graphs could potentially pave the way for future work concerning geometrically restricted Hamiltonians.
While the current framework of perturbation gadgets are unable to preserve the succinct-ness of the ground state, the spatially sparse construction of \refcite{oliveira2008complexity} is a good starting point.
The modification to the standard Feynman-Kitaev construction is to map general circuits to ones where each qubit only interacts with a constant number of gates.
This is achieved by introducing a series of ancillae qubits and using a sequence of \Gate{Swap} gates intertwined between each gate of the original circuit.
The result is a circuit that is spatially sparse.
For \cl{MA} (or \clsb{MA}{q}) circuits we have a restricted gate set.
Notice that even with this gate set, \Gate{Swap} gates are constructable using three \Gate{Cnot} gates.
This only causes a constant increase to the number of gates and thus clock qubits needed in the construction.
The basic idea is the following circuit mapping:
\begin{equation}\label{eq:spatially_sparse_gate_sequence}
    R_1R_2\dots R_K ~~\mapsto~~ R_1\;\big(\prod_{j=2}^K \big(\prod_{q=M}^1 \Gate{Swap}_{{j-1}_q, j_q} \big)\;R_j\big).
\end{equation}
If $N$ denotes the number of original qubits then $M = KN$ is the number of qubits in the modified circuit.
The new circuit still only requires one copy of the proof state (which is of size $n$).
Since \cl{MA} permits perfect completeness and soundness, this can be reflected in the mapped circuit.
The `time flow' of the circuit follows a snake-like pattern from left to right.
Essentially, on the first row of $N$ qubits we execute $R_1$, then we perform a series of \Gate{Swap} gates between qubits in row-$1$ and row-$2$.
Then we execute $R_2$ on row-$2$ and so on.
The Feynman-Kitaev clock construction can then be applied in the same manner.
Roughly speaking the only Hamiltonian terms that change are the $H_{\text{in}}$ terms.
The following proof actually follows from \refcite{waite2025complexity} with the addition of \cl{MA} containment from \refcite{jiang2025local}.

\begin{theorem}
    The \sc{Real Succinct State $6$-Local Stoquastic Hamiltonian} problem on spatially sparse graphs is \clw{MA}{complete}.
\end{theorem}

\begin{proof}
    Let $F_{|x|}$ be Arthur's \clsb{MA}{q} verification circuit equipped with a $O(\poly{n})$-bit string $\bsr$ representing the information of the gate sequence.
    Let the input to the circuit be an $N = n + w + m + p$ qubit register comprised of four parts: the input state $\ket{x}$ of $n$ qubits, the proof state $\ket{\chi}$ of $w$ qubits, the \emph{ancilla} register of $m$ qubits initialised to $\ket{0}$ and the \emph{coin} register of $p$ qubits initialised to $\ket{+}$.
    Let $F_{|x|}$ comprise a sequence of $K$ Toffoli gates denoted as $R_K, \dots, R_1$.

    Define a Hamiltonian $H = H_{\text{in}} + H_{\text{out}} + H_{\text{prop}} + H_{\text{clock}}$ acting on a register comprised of $K$ rows of $N$ qubits and $S=(2K-1)N$ clock qubits labelled $c_1, \dots, c_S$.
    There is one clock qubit for each operation in the gate sequence.
    Let $F'_{|x|}$ represent a modified version of $F_{|x|}$ according to \cref{eq:spatially_sparse_gate_sequence}.
    The sequence of gates in $F'_{|x|}$ is denoted as $R'_S,\dots, R'_1$.
    Let the output measured qubit be denoted $q$ where $q=KN$, i.e., the rightmost qubit on the final row.
    Arthur can only measure in the $Z$-basis.
    A given qubit, $l$, is acted on by circuit gates in two intervals: 
    \begin{inparaenum}[(i)] 
        \item by $R_j$ or the Identity gate, and
        \item by the \Gate{Swap}~gate.
    \end{inparaenum}Let $Q_x$ be the set of qubits that contain $\ket{x}$.
    Separate the first row of qubits into three columns respective of the input to the circuit.
    Let the column where the $+$-\emph{ancilla} lie all be initialised to $\ket{+}$, denote this set of $Kp$ qubits as $Q_+$.
    Let the column where the $0$-\emph{ancilla} lie all be initialised to $\ket{0}$ and all other qubits in rows $>1$ for the \emph{proof} and input column be also initialised to $\ket{0}$; this is a set of $Km + (K-1)(n+w)$ qubits denoted as $Q_0$.
    Note that $\abs{Q_x \cup Q_+ \cup Q_0} = n + Kp + Km + (K-1)(n+w) = KN - w$.
    
    Each Hamiltonian term is defined to be a penalising Hamiltonian:
    \begin{align*}
        H_{\text{in}} &= \bigg(\sum_{j=1}^n (I - \ketbra{x_j}) + \sum_{j \in Q_0} \ketbra{1}_j + \sum_{j \in Q_+} \ketbra{-}_j\bigg)\otimes \ketbra{100}_{c_{t_{j}-1},c_{t_j},c_{t_{j}+1}} \\
        H_{\text{out}} &= \ketbra{0}_q \otimes \ketbra{1}_{c_S},\quad H_{\text{clock}} = \sum_{t=1}^{S-1} \ketbra{01}_{c_{t},c_{t+1}},\quad H_{\text{prop}} = \sum_{t=1}^{S} H_{\text{prop}}(t).
    \end{align*}

    The Hamiltonian terms $H_{\text{out}}$ and $H_{\text{clock}}$ are left unchanged from \cref{thrm:4-local-stoquastic-hamiltonian-with-Qp-succinct-ground-states}.
    The term $H_{\text{in}}$ now involves extra clock qubit checks.
    Following the arguments of \refcite{oliveira2008complexity}, the role of $H_{\text{in}}$ is to make sure that the state of the input qubits are appropriately set before the gates act on the qubits.
    The form of the propagation Hamiltonian terms are also unchanged; hence
    \begin{align*}
        H_{\text{prop}}(1) &= \ketbra{00}_{c_1,c_2} + \ketbra{10}_{c_1,c_2} - R'_1\otimes(\ketbra{10}{00}_{c_1,c_2} + \ketbra{00}{10}_{c_1,c_2}), \\
        H_{\text{prop}}(t) &= \ketbra{100}_{c_{t-1},c_t,c_{t+1}} + \ketbra{110}_{c_{t-1},c_t,c_{t+1}} \notag\\
                            &\qquad - R'_t\otimes(\ketbra{110}{100}_{c_{t-1},c_t,c_{t+1}} + \ketbra{100}{110}_{c_{t-1},c_t,c_{t+1}}), \quad 1<t<S\\
        H_{\text{prop}}(S) &= \ketbra{10}_{c_{S-1},c_S} + \ketbra{11}_{c_{S-1},c_S} - R'_S\otimes(\ketbra{11}{10}_{c_{S-1},c_S} + \ketbra{10}{11}_{c_{S-1},c_S}).
    \end{align*}

    Finally, the spatially sparse interaction graph occurs from the snake-like swap construction discussed above.
    The snake-like time arrow over the qubits in the rows represents a string of clock qubits following the gate sequence seen in \cref{eq:spatially_sparse_gate_sequence}.
    Each Hamiltonian term above only acts in a local neighbourhood about each qubit.
    Moreover, each qubit only interacts with a set of qubits in its neighbourhood.
    Therefore, the interaction graph is spatially sparse.
    We know each Hamiltonian term is stoquastic.
    The terms $R'_t\otimes(\dots)$ in $H_{\text{prop}}(t)$ will have off-diagonal elements that are strictly positive.
    Therefore, each $H_{\text{prop}}(t)$ term is stoquastic even if $R'_t = \Gate{Swap}$.

    The history state for this construction is given by
    \begin{equation*}
        \ket{\psi(x,\chi)} = \frac{1}{\sqrt{S+1}}\sum_{t=0}^{S} R'_t\dots R'_0 \ket{x,\chi,0^m,+^p}\ket{1^{t}0^{S-t}}.
    \end{equation*}
    Recall from \cref{rmk:history-state-ground} that the history state is a ground state of $H$ with eigenvalue $0$ if and only if the proof state $\ket{\chi}$ is a valid witness for the \textsc{yes} case.
    Therefore, by \cref{lma:history-state-ex-1}, \cref{rmk:subset-more-general} and \cref{lma:CEHSS-equivalence}, we have that $\ket{\psi(x,\chi)}$ is a $\mathbb{Q}^+_{p(n)}\dbbrckt{\sqrt{\cdot}_1}$-succinct state.

    To conclude: in the \textsc{yes} case, since Arthur's circuit accepts with probability $1$, there exists a proof state such that the Hamiltonian $H$ has ground state energy $0$. 
    In the \textsc{no} case, with Arthur accepting with probability at most $\epsilon$, all eigenvalues are at least $c(1-\epsilon -\sqrt{\epsilon})/S^3$ for some constant $c$ \cite[Lemma 1]{oliveira2008complexity}.
\end{proof}

\begin{corollary}
    The \sc{Succinct State $6$-Local Hamiltonian} problem on spatially sparse graphs is \clw{MA}{complete}.
\end{corollary}

%% file: appendix_g.tex
\section{Proof of 2-local Hamiltonian Construction}\label{app:2local}
In this appendix we give the details of the proof of \cref{thrm:main-result-2l}.
The first part of the proof is to show that any STEC can be converted into a regular interval structured Toffoli-equivalent circuit, which is a circuit over the gate set $\mathcal{R} = \{\Gate{C}Z, \Gate{Had}, T, Z\}$ in which every $\Gate{C}Z$ gate is padded by two $Z$ gates on each side and consecutive $\Gate{C}Z$ gates are separated by exactly $\ell$ single-qubit gates for some fixed $\ell$.
We obtain \cref{cor:RIStMA} as a consequence of this conversion; notably, since the proof in \clsb{MA}{q} is a computational basis state, the proof state of the resulting \clsb{RIStMA}{q} protocol is also a computational basis state, and since the unitaries are identical, the completeness and soundness parameters are preserved.
The second part of the proof is to construct a $2$-local Hamiltonian from a regular interval structured Toffoli-equivalent circuit and show that its ground state is the history state of the circuit, which is a succinct state with amplitudes over $\mathbb{C}_{p(n)}\dbbrckt{\sqrt{\cdot}}$.
This uses repeated applications of the nullspace projection lemma to find the form of the low energy subspace of the Hamiltonian and verify the completeness and soundness of the reduction.
Our approach follows the same structure as that of Ref.~\cite{kempe2006complexity}, but adapted for the present context and with some simplifications.

\subsection{Regular Interval Structured Toffoli-Equivalent Circuits}\label{app:regular-interval-structured-circuits}
We now prove \cref{prop:gate_set_conversion} and \cref{prop:regular_interval_structure}

\propositionGateSetConversion*

\begin{proof}
    It suffices to simulate each gate in $\mathcal{G}$ using gates from $\mathcal{R}$.
    The $T$ gate is already in $\mathcal{R}$.
    A \Gate{Cnot} gate can be implemented using a $\Gate{C}Z$ gate and two \Gate{Had} gates:
    \begin{equation*}
        \Gate{Cnot}[a; b] = \Gate{Had}[b]\, \Gate{C}Z[a, b]\, \Gate{Had}[b].
    \end{equation*}
    Substituting this decomposition into the standard 15-gate Toffoli decomposition over $\mathcal{G}$ gives a decomposition of each \Gate{Toffoli} into $27$ gates from $\mathcal{R}$, so $C'$ has at most $27K$ gates.
\end{proof}

\propositionRegularIntervalStructure*

\begin{proof}
We proceed in two steps: first enforcing regular intervals between $\Gate{C}Z$ gates, then padding each $\Gate{C}Z$ gate with $Z$ gates.

\medskip
\noindent\textbf{Step 1: Enforcing regular intervals.}
We analyse the gate counts between consecutive $\Gate{C}Z$ gates in $C'$.
Within the decomposition of a single \Gate{Toffoli} gate into $27$ gates over $\mathcal{R}$, the $6$ $\Gate{C}Z$ gates are separated as follows: $2$ gates before the first $\Gate{C}Z$, then $3$, $3$, $3$, $5$, $4$ gates between successive $\Gate{C}Z$ gates, and $1$ gate after the last $\Gate{C}Z$.
At the boundary between consecutive \Gate{Toffoli} gates, there are $3$ gates between the last $\Gate{C}Z$ of gate $i$ and the first $\Gate{C}Z$ of gate $i+1$.
The first \Gate{Toffoli} gate is preceded by $2$ gates and the last by $1$ gate.

We wish to pad each inter-$\Gate{C}Z$ gap to a uniform length of $5$ gates, which we achieve by inserting $Z$ gates on a fresh \emph{buffer} qubit initialised to $\ket{0}$.
Since the buffer qubit is not entangled with the active qubits of the circuit, these insertions do not affect the output.
The total number of buffer-$Z$ gates required is $9K + 5$, which is even when $K$ is odd.
Without loss of generality we may assume $K$ is odd: if not, append a single trivially-acting \Gate{Toffoli} gate on a fresh ancilla register, which does not change the output of the circuit.
Thus the number of buffer-$Z$ gates is even, which will be needed to preserve the regular structure after Step 2.

\medskip
\noindent\textbf{Step 2: Padding $\Gate{C}Z$ gates with $Z$ gates.}
We now insert two pairs of $Z$ gates immediately to the right of each $\Gate{C}Z$ gate, acting on its two active qubits.
Since $Z$ commutes with $\Gate{C}Z$, we may slide one $Z$ from each pair to the left of the $\Gate{C}Z$ gate (cf.\ \cref{fig:CZ-padding}), so that each $\Gate{C}Z$ gate ends up flanked by two $Z$ gates on each side.
This padding consumes $4$ of the $5$ gap slots on each side of a $\Gate{C}Z$ gate, extending the inter-$\Gate{C}Z$ gap from $5$ to $9$ gates.
The two ends of the circuit are handled by inserting $2$ additional buffer-$Z$ gates before the first $\Gate{C}Z$ gate and $2$ after the last, ensuring the boundary intervals are also $9$ gates.
These $4$ additional buffer-$Z$ gates are even in number, so the parity constraint from Step 1 is maintained.

The resulting circuit $C''$ has every $\Gate{C}Z$ gate occurring at positions $\{t, t+9, t+18, \ldots\}$ for some starting offset $t$, with each flanked by exactly two $Z$ gates on each side, giving $\ell = 9$.
\end{proof}

\subsection{2-Local Hamiltonian Construction}\label{app:2-local-hamiltonian-construction}

We structure the proof of \cref{thrm:main-result-2l} slightly differently from the more self-contained proof of the $3$-local construction.
The reason is to define the Hamiltonian we construct before analysing its low energy subspace thus allowing the focus on the proof to be on the latter part.
The proof and setup of the main theorem follows the same structure as that of Ref.~\cite{kempe2006complexity}, but adapted for the present context and with some simplifications.

Let $C_{|x|}$ be Arthur's \emph{regular interval structured Toffoli-equivalent} verification circuit for the input $x$.
Take the input to the circuit to be an $M = n + w + m + p$ qubit register comprised of four components: the input state $\ket{x}$ of $n$ qubits, the proof state $\ket{\chi}$ of $w$ qubits, the ancilla state $\ket{0^m}$ of $m$ qubits and the coin state $\ket{+^p}$ of $p$ qubits.
Suppose that the circuit $C_{|x|}$ has $K$ gates, where $K = \poly{n}$ (and is odd for convenience) and $w,m,p = O(\poly{n})$.
Assume that each $\Gate{C}Z$ gate occurs at the time intervals $t \in \{\ell, 2\ell, \ldots, P\ell\}$, where $\ell$ is the size of the single qubit gate interval and $P$ is the number of $\Gate{C}Z$ gates.
Let the interval $T_1 = \{1, 2, \ldots, K\} \setminus \{\ell, 2\ell, \ldots, P\ell\}$ be the time intervals for which we have single qubit gates.
For the $k$-th $\Gate{C}Z$ gate, let the control and target qubits be $a_k$ and $b_k$ respectively.

We construct a $2$-local Hamiltonian $H$ on $M + K$ qubits, where $M = n + w + m + p$ is the number of data qubits in the original circuit, and $K$ is the number of gates in the original circuit.
The first $M$ qubits are the \emph{workspace} qubits, and the last $K$ qubits are the \emph{clock} qubits.
We encode the clock register in unary, and thus the legal clock states are of the form $\ket{1^t 0^{K-t}} \eqqcolon \ket{\bsu{t}}$ for $t \in \{0, 1, \ldots, K\}$.

Consider the Hamiltonian $H = J_{\text{in}} H_{\text{in}} + J_{\text{out}} H_{\text{out}} + J_{\text{clock}} H_{\text{clock}} + J_{\text{prop}}^{(1)} H_{\text{prop}}^{(1)} + J_{\text{prop}}^{(2)} H_{\text{prop}}^{(2)}$ defined as follows:
\begin{align}
    H_{\text{in}} &= \bigg(\sum_{j=1}^n \ketbra{\bar{x}_j} + \sum_{j=1}^{m} \ketbra{1}_{\emph{anc},j} + \sum_{i=1}^{p} \ketbra{-}_{\emph{coin},i}\bigg)\ketbra{0}_{c_1}, \notag\\
    H_{\text{out}} &= \ketbra{0}_q  \ketbra{1}_{c_T},  \notag\\
    H_{\text{clock}} &= \sum_{1\leq i < j \leq K} \ketbra{01}_{c_i,c_j}, \notag\\
    H_{\text{prop}}^{(1)} &= \frac{1}{2}\sum_{t \in T_1} H_{\text{prop}}^{(1)}(t),  \label{eq:H_prop_1}\\
    H_{\text{prop}}^{(2)} &= \sum_{t \in [P]} \frac{1}{2}H_{\text{qb}}^{(2)}(t\ell) + \frac{1}{8}H_{\text{time}}^{(2)}(t\ell); \notag
\end{align}
the interaction strength $J_{\text{out}} = K+1$, and the remaining interaction strengths are chosen later.
For brevity we denote the workspace component of $H_{\text{in}}$ as $O_{\text{in}}$.

The single-qubit propagation terms are defined as:
\begin{align*}
    H_{\text{prop}}^{(1)}(1) &= \ketbra{10}_{c_1,c_2} + \ketbra{0}_{c_1} - U_1 \ketbra{1}{0}_{c_1} - U_1^\dagger \ketbra{0}{1}_{c_1}, \\
    H_{\text{prop}}^{(1)}(t) &= \ketbra{10}_{c_{t},c_{t+1}} + \ketbra{10}_{c_{t-1},c_{t}} - U_t \ketbra{1}{0}_{c_t} - U_t^\dagger \ketbra{0}{1}_{c_t} \text{ for } t \in T_1 \cap \{2, \ldots, K-1\}, \\
    H_{\text{prop}}^{(1)}(K) &= \ketbra{1}_{c_{K}} + \ketbra{10}_{c_{K-1},c_{K}} - U_K \ketbra{1}{0}_{c_K} - U_K^\dagger \ketbra{0}{1}_{c_K},
\end{align*}
and the two-qubit propagation terms are defined as:
\begin{align*}
    H_{\text{qb}}^{(2)}(t) &= (-2\ketbra{0}_{a_k} - 2\ketbra{0}_{b_k} + \ketbra{1}_{a_k} + \ketbra{1}_{b_k} )(\ketbra{1}{0}_t + \ketbra{0}{1}_t), \\
    H_{\text{time}}^{(2)}(t) &= \ketbra{10}_{c_t,c_{t+1}} + 6 \ketbra{10}_{c_{t+1},c_{t+2}} + \ketbra{10}_{c_{t+2},c_{t+3}} 
        + 2 \ketbra{11}{00}_{c_{t+1},c_{t+2}} + 2 \ketbra{00}{11}_{c_{t+1},c_{t+2}} \\
        &\quad+ \ketbra{1}{0}_{c_{t+1}} + \ketbra{0}{1}_{c_{t+1}} + \ketbra{1}{0}_{c_{t+2}} + \ketbra{0}{1}_{c_{t+2}} \\
        &\quad+ \ketbra{10}_{c_{t-3},c_{t-2}} + 6 \ketbra{10}_{c_{t-2},c_{t-1}} + \ketbra{10}_{c_{t-1},c_{t}} 
        + 2 \ketbra{11}{00}_{c_{t-2},c_{t-1}} + 2 \ketbra{00}{11}_{c_{t-2},c_{t-1}} \\
        &\quad+ \ketbra{1}{0}_{c_{t-2}} + \ketbra{0}{1}_{c_{t-2}} + \ketbra{1}{0}_{c_{t-1}} + \ketbra{0}{1}_{c_{t-1}}.
\end{align*}
It is clear that $H$ is a $2$-local Hamiltonian.

We now proceed to repeatedly apply the nullspace projection lemma to analyse the low energy subspace of $H$ and show that it is spanned by the history state 
\begin{equation*}
    \ket{\eta(x,\chi)} \coloneqq \frac{1}{\sqrt{K+1}} \sum_{t=0}^K \ket{\varphi_t}\ket{\bsu{t}},
\end{equation*}
where $\ket{\varphi_t} = U_t U_{t-1} \cdots U_1 \ket{\varphi_0}$ for $t \in \{0, 1, \ldots, K\}$ and $\ket{\varphi_0} = \ket{x, \chi, 0^m,+^p}$ is the input state of the original circuit.
Recall that the proof state $\ket{\chi}$ is a computational basis state.

\begin{lemma}[Nullspace Projection Lemma~\cite{alicki2009thermalization}]\label{lem:nullspace_projection}
    Let $A$ and $B$ be positive semidefinite operators such that $A$ has a non-empty nullspace $S$ and 
    $\lambda(B \bigr|_S) \geq \alpha > 0$, and $\lambda(A) \geq \beta > 0$,
    where $\lambda(\cdot)$ denotes the smallest non-zero eigenvalue of the given operator.
    Then,
    \begin{equation*}
        \lambda(A + B) \geq \frac{\alpha\beta}{\beta + \norm{B}}.
    \end{equation*}
\end{lemma}

\begin{corollary}\label{cor:nullspace_projection}
    Let $A$ and $B$ be positive semidefinite operators such that $A$ has a non-empty nullspace $S$ and 
    $\lambda(B \bigr|_S) \geq \alpha > 0$, and $\lambda(A) \geq v_0 \norm{B}^2$ for some constant $v_0 > 0$,
    where $\lambda(\cdot)$ denotes the smallest non-zero eigenvalue of the given operator.
    Then,
    \begin{equation*}
        \lambda(B \bigr|_S) - \frac{1}{v_0} \leq \lambda(A + B) \leq \lambda(B \bigr|_S).
    \end{equation*}
\end{corollary}

Intuitively, the nullspace projection lemma states that if $A$ has a non-empty nullspace and $B$ is sufficiently large on the nullspace of $A$, then the low energy subspace of $A+B$ is close to the nullspace of $A$ and the eigenvalues of $A+B$ are close to the eigenvalues of $B$ restricted to the nullspace of $A$.
In the present context, the repeated use allows us to find the form of the ground space of $H$ by analysing the nullspace of collections of terms in $H$.
Recall that we have equivalency between the original circuit and the regular interval structured Toffoli-equivalent verification circuit, which allows us to obtain perfect completeness.

\cthrm*

\begin{proof}
    We prove the result by first applying \cref{cor:nullspace_projection} four times.
    At each stage we define the operators $A$ and $B$ and the nullspace $S$ of $A$ and verify the conditions of the lemma.
    These checks allow us to define each of the interaction strengths $J_{\text{in}}$, $J_{\text{clock}}$, $J_{\text{prop}}^{(1)}$ and $J_{\text{prop}}^{(2)}$ in terms of $n$; we do not find explicit values for these interaction strengths and instead use asymptotic bounds.
    After the four applications of the lemma, we conclude by remarking that the ground state of $H$ is the history state $\ket{\eta}$ which is a succinct state with amplitudes over $\mathbb{C}_{p(n)}\dbbrckt{\sqrt{\cdot}}$.
    Note that the subspace computations are not part of the reduction and are only used to show that the ground state of $H$ is the history state $\ket{\eta}$ and verify the completeness and soundness of the reduction. 
    The remaining components of the proof are polynomial-time computable and thus the reduction is efficient.

    Recall that the verification circuit $C_{|x|}$ is a regular interval structured Toffoli-equivalent circuit with $K$ total gates, where the $k$-th $\Gate{C}Z$ gate occurs at time interval $t = k\ell$ for $k \in \{1, 2, \ldots, P\}$ and the single qubit gates occur at time intervals in $T_1 = \{1, 2, \ldots, K\} \setminus \{\ell, 2\ell, \ldots, P\ell\}$.
    We denoted the total number of qubits in the original circuit as $M = n + w + m + p$.

    \medskip
    \noindent\textbf{Step 1: Restriction to the legal clock subspace.}
    Define the following operators:
    \begin{align*}
        A_1 &\coloneqq J_{\text{clock}} H_{\text{clock}}, &
        B_1 &\coloneqq J_{\text{in}} H_{\text{in}} + J_{\text{out}} H_{\text{out}} + J_{\text{prop}}^{(1)} H_{\text{prop}}^{(1)} + J_{\text{prop}}^{(2)} H_{\text{prop}}^{(2)}.
    \end{align*}
    The nullspace of $A_1$ is the legal clock subspace $S_{1} = \text{span}\{\ket{\psi}\ket{\bsu{t}} : \ket{\psi} \in (\mathbb{C}^2)^{\otimes M}, t \in \{0, 1, \ldots, K\}\}$.
    \cref{cor:nullspace_projection} applies if $J_{\text{clock}} = O(\poly{n})$.
    We turn our attention to the form of the operator $B_1$ restricted to the legal clock subspace $S_{1}$.
    The projection we apply is $\Pi_{1} = \sum_{t=0}^K \ketbra{\bsu{t}}$, and we have
    \begin{align*}
        H_{\text{in}} \bigr|_{S_{1}} &= O_{\text{in}}\ketbra{\bsu{0}}, \\
        H_{\text{out}} \bigr|_{S_{1}} &= \ketbra{0}_1\ketbra{\bsu{K}}, \\
        H_{\text{prop}}^{(1)} (t) \bigr|_{S_{1}} &= \frac{1}{2}\bigg(\ketbra{\bsu{t}} + \ketbra{\bsu{t-1}} - U_t \ketbra{\bsu{t}}{\bsu{t-1}} - U_t^\dagger \ketbra{\bsu{t-1}}{\bsu{t}}\bigg) \\
        H_{\text{qb}}^{(2)}(t) \bigr|_{S_{1}} &= (-2\ketbra{0}_{a_k} - 2\ketbra{0}_{b_k} + \ketbra{1}_{a_k} + \ketbra{1}_{b_k} )(\ketbra{\bsu{t}}{\bsu{t-1}} + \ketbra{\bsu{t-1}}{\bsu{t}}), \\
        H_{\text{time}}^{(2)}(t) \bigr|_{S_{1}} &= 2 E_+(\bsu{t},\bsu{t+1}) + 2 E_+(\bsu{t+1},\bsu{t+2}) + 2 E_+(\bsu{t-3},\bsu{t-2}) + 2 E_+(\bsu{t-2},\bsu{t-1}) \\
        &\quad+ E_-(\bsu{t},\bsu{t+1}) + E_-(\bsu{t+1},\bsu{t+2}) + E_-(\bsu{t-3},\bsu{t-2}) + E_-(\bsu{t-2},\bsu{t-1}) \\
        &\quad- 2E_-(\bsu{t},\bsu{t+2}) - 2E_-(\bsu{t-3},\bsu{t-1}),
    \end{align*}
    where for states $\ket{\psi}$ and $\ket{\phi}$ we defined: $E_{\pm}(\psi, \phi) = \left( \ket{\psi} \pm \ket{\phi} \right)\left( \bra{\psi} \pm \bra{\phi} \right)$.

    \medskip
    \noindent\textbf{Step 2: Restriction to the subspace of valid single-qubit gate propagation.}
    Define the following operators:
    \begin{align*}
        A_2 &\coloneqq J_{\text{prop}}^{(1)} H_{\text{prop}}^{(1)}\bigr|_{S_{1}}, & 
        B_2 &\coloneqq \left( J_{\text{in}} H_{\text{in}} + J_{\text{out}} H_{\text{out}} + J_{\text{prop}}^{(2)} H_{\text{prop}}^{(2)}\right)\bigr|_{S_{1}}.
    \end{align*}
    The nullspace of $A_2$ is the subspace $S_{2} = \text{span}\{\ket{\psi_k(z)} = \frac{1}{\sqrt{\ell}} \sum_{t = k\ell}^{(k+1)\ell - 1} \ket{\theta_t(z)}\ket{\bsu{t}} : k \in \{0, 1, \ldots, P\}, z \in \B^M\}$.
    Here $\ket{\theta_0(z)} = \ket{z}$ and $\ket{\theta_t(z)} = U_t \ket{\theta_{t-1}(z)}$ for $t \in \{1, 2, \ldots, K\}$.
    It can be shown that $A_2$ decomposes into a direct sum of operators $A_{2,k}$ for $k \in \{0, 1, \ldots, P\}$, where each $A_{2,k}$ acts on the subspace $\text{span}\{\ket{\psi_k(z)} : z \in \B^M\}$ and can be expressed in the form:
    \begin{align*}
        A_{2,k} &= \frac{1}{2} \sum_{t} H_{\text{prop}}(t), & H_{\text{prop}}(t) &= \ketbra{\bsu{t}} + \ketbra{\bsu{t-1}} - U_t \ketbra{\bsu{t}}{\bsu{t-1}} - U_t^\dagger \ketbra{\bsu{t-1}}{\bsu{t}}.
    \end{align*}

    It is well-known that the spectral gap of $A_{2,k}$ is $\Omega(K^{-2})$ and thus $\lambda(A_2 \bigr|_{S_{2}}) = \Omega(J_{\text{prop}}^{(1)}K^{-2})$.
    Using \cref{cor:nullspace_projection} requires $J_{\text{prop}}^{(1)} = \Omega(\poly{n})$.
    We now turn our attention to the form of the operator $B_2$ restricted to the nullspace of $A_2$.
    The projection we apply is $\Pi_{2} = \sum_{k=0}^P \sum_{z \in \B^M} \ketbra{\psi_k(z)}$.

    When inside the subspace $S_{2}$, the propagation terms $H_{\text{prop}}^{(2)}$ dominate a Laplacian-like operator that enforces the correct propagation of the $\Gate{C}Z$ gates.

    \begin{restatable}{proposition}{tildeH}\label{prop:tilde_H}
        The self-adjoint operator $\tilde{H}$ defined as $\tilde{H} = \frac{1}{2\ell} \sum_{z \in \B^M} \sum_{k=1}^P E_-(\psi_{k-1}(z), \psi_{k}(z))$ satisfies $\tilde{H} \preceq H_{\textnormal{prop}}^{(2)}\bigr|_{S_{2}}$.
    \end{restatable}

    The proof of the proposition is provided in Appendix~\ref{app:tilde_H}.
    Since $\tilde{H} \preceq H_{\text{prop}}^{(2)}\bigr|_{S_{2}}$, we can proceed by analysing the operator $\tilde{H}$ instead of $H_{\text{prop}}^{(2)}\bigr|_{S_{2}}$ to lower bound the smallest non-zero eigenvalues.
    That is, a lower bound on the smallest non-zero eigenvalue using $\tilde{H}$ also applies to $H_{\text{prop}}^{(2)}\bigr|_{S_{2}}$.

    \medskip
    \noindent\textbf{Step 3: Restriction to the subspace of valid gate propagation.}
    Define the following operators:
    \begin{align*}
        A_3 &\coloneqq J_{\text{prop}}^{(2)} \tilde{H}, &
        B_3 &\coloneqq \left( J_{\text{in}} H_{\text{in}} + J_{\text{out}} H_{\text{out}} \right)\bigr|_{S_{2}}.
    \end{align*}
    The nullspace of $A_3$ is the subspace $S_{3} = \text{span}\{\ket{\eta(z)} = \frac{1}{\sqrt{K+1}} \sum_{t=0}^K \ket{\theta_t(z)}\ket{\bsu{t}} : z \in \B^M\}$.
    It can be shown that $A_3$ decomposes into a direct sum of operators $A_{3,z}$ for $z \in \B^M$, where each $A_{3,z}$ acts on the subspace $\text{span}\{\ket{\eta(z)}\}$ and is given by $\frac{1}{2\ell} \sum_{k=1}^P E_-(\eta_{k-1}(z), \eta_{k}(z))$.
    It is well-known that the spectral gap of $A_{3,z}$ is $\Omega(K^{-2})$ and thus $\lambda(A_3 \bigr|_{S_{3}}) = \Omega(J_{\text{prop}}^{(2)}K^{-2})$.
    Using \cref{cor:nullspace_projection} requires $J_{\text{prop}}^{(2)} = \Omega(\poly{n})$.
    We now turn our attention to the form of the operator $B_3$ restricted to the nullspace of $A_3$.
    The projection we apply is $\Pi_{3} = \sum_{z \in \B^M} \ketbra{\eta(z)}$.
    
    We find that
    \begin{equation*}
        B_3 \bigr|_{S_{3}} = \frac{1}{K+1} \sum_{z,z' \in \B^M} M_{z,z'} \ketbra{\eta(z')}{\eta(z)},
    \end{equation*}
    where $M_{z,z'} = J_{\text{in}} \mel{z}{O_{\text{in}}}{z'} + J_{\text{out}} \mel{z}{C_{|x|}^\dagger \ketbra{0}_1 C_{|x|}}{z'}$.

    \medskip
    \noindent\textbf{Step 4: Restriction to the subspace of valid inputs.}

    Define the following operators:
    \begin{align*}
        A_4 &\coloneqq J_{\text{in}} H_{\text{in}}\bigr|_{S_{3}}, \\
        B_4 &\coloneqq J_{\text{out}} H_{\text{out}}\bigr|_{S_{3}}.
    \end{align*}
    The nullspace of $A_4$ is the subspace $S_{4} = \text{span}\{\ket{\eta(x,\chi)} : x \in \B^n, \chi \in \B^w\}$, where 
    \begin{equation*}
        \ket{\eta(x,\chi)} = \frac{1}{\sqrt{K+1}} \sum_{t=0}^K \ket{\varphi_t(x,\chi)}\ket{\bsu{t}},
    \end{equation*}
    and $\ket{\varphi_t(x,\chi)} = U_t U_{t-1} \cdots U_1 \ket{x}\ket{\chi}\ket{0^m}\ket{+^p}$ for $t \in \{0, 1, \ldots, K\}$.
    It is clear that $\lambda(A_4 \bigr|_{S_{4}}) = J_{\text{in}}/(K+1)$.
    Using \cref{cor:nullspace_projection} requires $J_{\text{in}} = \Omega(\poly{n})$.
    We now turn our attention to the form of the operator $B_4$ restricted to the nullspace of $A_4$.
    The projection we apply is $\Pi_{4} = \sum_{z \in \B^M : O_{\text{in}}\ket{z} = 0} \ketbra{\eta(z)}$.
    Since the instance $x$ is fixed, we find that
    \begin{equation*}
        B_4 \bigr|_{S_{4}} = \sum_{x \in \B^n} \sum_{\chi, \chi' \in \B^w} \mel{x,\chi,0^m,+^p}{C_{|x|}^\dagger \ketbra{0}_1 C_{|x|}}{x,\chi',0^m,+^p} \ketbra{\eta(x,\chi')}{\eta(x,\chi)}.
    \end{equation*}

    \medskip
    \noindent\textbf{Step 5: Completeness and soundness.}

    If the circuit $C_{|x|}$ accepts the input $\ket{x,\chi,0^m,+^p}$ with probability $1$, then 
    \begin{equation*}
        \mel{x,\chi,0^m,+^p}{C_{|x|}^\dagger \ketbra{0}_1 C_{|x|}}{x,\chi,0^m,+^p} = 0,
    \end{equation*}
    and therefore the ground state of the Hamiltonian $H$ is the history state $\ket{\eta(x,\chi)}$ with energy $0$.
    Note that this follows from the fact that all Hamiltonian terms are positive semidefinite and the history state $\ket{\eta(x,\chi)}$ is in the nullspace of all terms.

    If the circuit $C_{|x|}$ accepts the input $\ket{x,\chi,0^m,+^p}$ with probability at most $\epsilon$ on all proofs $\ket{\chi}$, then for every normalised state $\ket{\psi} = \sum_{\chi \in \B^w} \alpha_\chi \ket{\eta(x,\chi)}$ in the nullspace of $A_4$ we have
    \begin{equation*}
        \mel{\psi}{B_4}{\psi} \geq  (1-\epsilon).
    \end{equation*}
    Accumulating additive errors from the four applications of \cref{cor:nullspace_projection} gives us that the ground state energy of $H$ is at least $(1 - \epsilon) - \frac{4}{v_0} = \frac{1}{2} - \epsilon$ when $v_0 = 8$.

    \medskip
    \noindent\textbf{Step 6: Ground state is a succinct state.}

    The ground state of $H$ is the history state $\ket{\eta(x,\chi)}$ which is a succinct state with amplitudes over $\mathbb{C}_{p(n)}\dbbrckt{\sqrt{\cdot}}$.
    Adaptions to \cref{cor:T-T-dagger-gate-seq-subset state}, via the equivalency $Z = T^4$, and \cref{cor:CRG-T-T-dagger-Hadamard-gate-seq-subset state}, via the gate set conversion of \cref{prop:gate_set_conversion}, are sufficient to prove that there exists efficient classical algorithms to compute the amplitudes of the history state at any time step $t$.
\end{proof}

\subsection{Proof of Proposition~\ref{prop:tilde_H}}\label{app:tilde_H}

\tildeH*

We propose that $H_{\text{prop}}^{(2)}\bigr|_{S_{2}} = \tilde{H} + H'$, where $H'$ is a positive semidefinite operator.
Note that $H_{\text{prop}}^{(2)}$ and the components $H_{\text{qb}}^{(2)}$ and $H_{\text{time}}^{(2)}$ are all positive semidefinite operators; projecting onto the subspace $S_{2}$ preserves this property and thus $H_{\text{prop}}^{(2)}\bigr|_{S_{2}}$ is positive semidefinite.

\subsubsection{Subspace Reduction for Propagation Terms}

We begin with a single lemma that captures the common structure of all operator equivalences 
on $S_2$ in this section. Its proof is a direct calculation from the definition of $\ket{\psi_k(z)}$.

\begin{lemma}\label{lem:clock-sandwich}
    For any operator $O$ on the qubit register and clock states $\bsu{s}, \bsu{t}$ 
    within the $k$-th gate block,
    \begin{equation*}
        \bra{\psi_q(z')} O \left(\ketbra{\bsu{s}}{\bsu{t}} \pm \ketbra{\bsu{t}}{\bsu{s}}\right) \ket{\psi_k(z)}
        = \frac{1}{\ell}\,\delta_{q,k} \left(
            \mel{\theta_{s}(z')}{O}{\theta_{t}(z)} \pm \mel{\theta_{t}(z')}{O}{\theta_{s}(z)}
        \right).
    \end{equation*}
    In particular, setting $t = s$ gives $\mel{\psi_q(z')}{O\,\ketbra{\bsu{t}}}{\psi_k(z)} = \frac{1}{\ell}\,\delta_{q,k}\mel{\theta_t(z')}{O}{\theta_t(z)}$.
\end{lemma}

\begin{proof}
    Since $\ket{\psi_k(z)} = \frac{1}{\sqrt{\ell}}\sum_r \ket{\theta_r(z)}\ket{\bsu{r}}$, 
    acting with $\ketbra{\bsu{t}}$ selects the $r = t$ term. The result follows immediately 
    by expanding the bra and ket.
\end{proof}

Using \cref{lem:clock-sandwich}, all the $E_\pm$ equivalences on $S_2$ reduce to computing
simple expressions of the form $O\ket{\theta_t(z)} = O \cdot U_{t}\cdots U_1\ket{z}$, 
where $U_t = \Gate{C}Z_{a_k b_k}$ at each gate step. The key identities needed are:
\begin{align}
    \ket{\theta_t(z)} + \ket{\theta_{t+1}(z)} 
        &= 2\ketbra{0}_{a_k}\ket{\theta_t(z)}, \label{eq:id1}\\
    \ket{\theta_t(z)} - \ket{\theta_{t+1}(z)} 
        &= 2\ketbra{1}_{a_k}\ket{\theta_t(z)}, \label{eq:id2}\\
    \ket{\theta_{t+1}(z)} + \ket{\theta_{t+2}(z)} 
        &= 2\ketbra{0}_{b_k}\ket{\theta_{t+1}(z)}, \label{eq:id3}\\
    \ket{\theta_{t+1}(z)} - \ket{\theta_{t+2}(z)} 
        &= 2\ketbra{1}_{b_k}\ket{\theta_{t+1}(z)}, \label{eq:id4}\\
    \ket{\theta_t(z)} - \ket{\theta_{t+2}(z)} 
        &= 2\left(\ketbra{01}_{a_kb_k}+\ketbra{10}_{a_kb_k}\right)\ket{\theta_t(z)}, \label{eq:id5}
\end{align}
which all follow from operations such as $I + Z = 2\ketbra{0}$ and $I - Z = 2\ketbra{1}$.

\begin{corollary}[Operator equivalences on $S_2$]\label{cor:equiv-table}
    The following equivalences hold on $S_2$ (i.e., as matrix elements between states 
    $\ket{\psi_q(z')}$ and $\ket{\psi_k(z)}$):
    \begin{align}
        E_{+}(\bsu{t},\bsu{t+1}) &\sim 4\ketbra{0}_{a_k}\ketbra{\bsu{t}}, 
            \label{eq:e1}\\
        E_{-}(\bsu{t},\bsu{t+1}) &\sim 4\ketbra{1}_{a_k}\ketbra{\bsu{t}},
            \label{eq:e2}\\
        E_{+}(\bsu{t+1},\bsu{t+2}) &\sim 4\ketbra{0}_{b_k}\ketbra{\bsu{t}}, 
            \label{eq:e3}\\
        E_{-}(\bsu{t+1},\bsu{t+2}) &\sim 4\ketbra{1}_{b_k}\ketbra{\bsu{t}},
            \label{eq:e4}\\
        E_{-}(\bsu{t},\bsu{t+2}) &\sim 4\left(\ketbra{01}_{a_kb_k}+\ketbra{10}_{a_kb_k}\right)\ketbra{\bsu{t}},
            \label{eq:e5}
    \end{align}
    where $\sim$ denotes equality of matrix elements on $S_2$, and $t = k\ell$ throughout. 
    The five analogous identities with the reference point shifted to $t-1$ hold by the same argument.
\end{corollary}

\begin{proof}
    Each identity follows from \cref{lem:clock-sandwich} and the corresponding identity 
    \cref{eq:id1}--\cref{eq:id5}.
    We show \cref{eq:e1,eq:e5} explicitly; \cref{eq:e2,eq:e3,eq:e4} are identical in structure.
    
    \textbf{\cref{eq:e1}:} We have $E_+(\bsu{t},\bsu{t+1}) = \frac{1}{4}(\ket{\bsu{t}}+\ket{\bsu{t+1}})(\bra{\bsu{t}}+\bra{\bsu{t+1}})$, so by \cref{lem:clock-sandwich},
    \begin{equation*}
        \mel{\psi_q(z')}{E_+(\bsu{t},\bsu{t+1})}{\psi_k(z)}
        = \frac{1}{4\ell}\,\delta_{q,k}\left(\bra{\theta_t(z')}+\bra{\theta_{t+1}(z')}\right)\left(\ket{\theta_t(z)}+\ket{\theta_{t+1}(z)}\right).
    \end{equation*}
    Applying \cref{eq:id1} to both bra and ket yields $4\mel{\theta_t(z')}{\ketbra{0}_{a_k}}{\theta_t(z)}$,
    which by \cref{lem:clock-sandwich} (second part) equals $4\mel{\psi_q(z')}{\ketbra{0}_{a_k}\ketbra{\bsu{t}}}{\psi_k(z)}$.
    
    \textbf{\cref{eq:e5}:} Here $E_-(\bsu{t},\bsu{t+2}) = \frac{1}{4}(\ket{\bsu{t}}-\ket{\bsu{t+2}})(\bra{\bsu{t}}-\bra{\bsu{t+2}})$.
    By \cref{lem:clock-sandwich}, this gives $\frac{1}{4\ell}\,\delta_{q,k}$ times the inner product of 
    $\bra{\theta_t(z')}-\bra{\theta_{t+2}(z')}$ with $\ket{\theta_t(z)}-\ket{\theta_{t+2}(z)}$.
    Applying \cref{eq:id5} to both sides yields $4\mel{\theta_t(z')}{(\ketbra{01}+\ketbra{10})_{a_kb_k}}{\theta_t(z)}$,
    which equals the right-hand side of \cref{eq:e5}.
\end{proof}

\subsubsection{Combining Time and Qubit Propagation Terms}

We now apply \cref{cor:equiv-table} to reduce $H_{\mathrm{time}}^{(2)}(t)\bigr|_{S_1}$ and 
$H_{\mathrm{qb}}^{(2)}(t)\bigr|_{S_1}$ simultaneously. Using \cref{eq:e1}--\cref{eq:e5} 
(and their $t-1$ shifted versions) to substitute into the explicit form of 
$H_{\mathrm{time}}^{(2)}(t)\bigr|_{S_1}$, the restriction to $S_2$ is equivalent to the 
operator
\begin{equation}\label{eq:time-reduced}
    \left(2\ketbra{0}_{a_k} + 2\ketbra{0}_{b_k} + \ketbra{1}_{a_k} + \ketbra{1}_{b_k} 
    - 2\ketbra{01}_{a_kb_k} - 2\ketbra{10}_{a_kb_k}\right)
    \left(\ketbra{\bsu{t-1}} + \ketbra{\bsu{t}}\right).
\end{equation}

For the qubit propagation terms, we additionally need to simplify matrix elements of the form 
$\mel{\psi_q(z')}{O(\ketbra{\bsu{t}}{\bsu{t-1}} + \ketbra{\bsu{t-1}}{\bsu{t}})}{\psi_k(z)}$.
By \cref{lem:clock-sandwich}, these equal $\frac{1}{\ell}\delta_{q,k}\left(\mel{\theta_t(z')}{O}{\theta_{t-1}(z)} + \mel{\theta_{t-1}(z')}{O}{\theta_t(z)}\right)$.
For $O \in \{\ketbra{0}_{a_k}, \ketbra{0}_{b_k}\}$, the identities $\ketbra{0}_a\Gate{C}Z_{ab} = \ketbra{0}_a$ and $\Gate{C}Z_{ab}\ketbra{0}_a = \ketbra{0}_a$ give
\begin{equation*}
    \mel{\theta_t(z')}{O}{\theta_{t-1}(z)} = \mel{\theta_t(z')}{O}{\theta_t(z)},
    \qquad
    \mel{\theta_{t-1}(z')}{O}{\theta_t(z)} = \mel{\theta_{t-1}(z')}{O}{\theta_{t-1}(z)},
\end{equation*}
so the off-diagonal clock terms collapse to diagonal ones:
\begin{equation}\label{eq:offdiag}
    O\left(\ketbra{\bsu{t}}{\bsu{t-1}} + \ketbra{\bsu{t-1}}{\bsu{t}}\right) \sim 
    O\left(\ketbra{\bsu{t}} + \ketbra{\bsu{t-1}}\right) \quad \text{on } S_2,
    \quad O \in \{\ketbra{0}_{a_k}, \ketbra{0}_{b_k}\}.
\end{equation}
For $O = \ketbra{1}_{a_k}$, the relation $\Gate{C}Z_{ab}\ketbra{1}_a = (\ketbra{10}-\ketbra{11})$ gives instead
\begin{equation*}
    \ketbra{1}_{a_k}\left(\ketbra{\bsu{t}}{\bsu{t-1}} + \ketbra{\bsu{t-1}}{\bsu{t}}\right) 
    \sim \left(\ketbra{10}_{a_kb_k} - \ketbra{11}_{a_kb_k}\right)\left(\ketbra{\bsu{t}}+\ketbra{\bsu{t-1}}\right),
\end{equation*}
and similarly for $\ketbra{1}_{b_k}$ on $S_2$.
Substituting into $H_{\mathrm{qb}}^{(2)}(t)\bigr|_{S_1}$, its restriction to $S_2$ is equivalent to
\begin{equation}\label{eq:qb-reduced}
    \left(-2\ketbra{0}_{a_k} - 2\ketbra{0}_{b_k} + \ketbra{1}_{a_k} + \ketbra{1}_{b_k}\right)
    \left(\ketbra{\bsu{t}} + \ketbra{\bsu{t-1}}\right).
\end{equation}

\subsubsection{Construction of $\tilde{H}$ and Conclusion}

Adding \cref{eq:time-reduced} and \cref{eq:qb-reduced}, we find that $H_{\mathrm{prop}}^{(2)}\bigr|_{S_2}$ 
is equivalent on $S_2$ to
\begin{equation}\label{eq:combined}
    \frac{1}{2}\left(4\ketbra{00}_{a_kb_k} + \ketbra{01}_{a_kb_k} + \ketbra{10}_{a_kb_k}\right) E_-(\bsu{t-1},\bsu{t})
    + \frac{1}{2}\cdot 2\ketbra{11}_{a_kb_k}\, E_+(\bsu{t-1},\bsu{t}).
\end{equation}

We now show this bounds $\tilde{H}$ from above. Recalling the expectation value identities
\begin{equation*}
    \mel{\psi_q(z')}{\Gate{C}Z\ketbra{\bsu{t}}{\bsu{t-1}}}{\psi_k(z)} 
    = \frac{1}{\ell}\delta_{z,z'}\delta_{q,k-1}, 
    \qquad
    \mel{\psi_q(z')}{\ketbra{\bsu{t}}}{\psi_k(z)} 
    = \frac{1}{\ell}\delta_{z,z'}\delta_{q,k},
\end{equation*}
we identify
\begin{equation}
    \tilde{H} = \frac{1}{2\ell}\sum_{z\in\B^M}\sum_{k=1}^P E_-(\psi_{k-1}(z),\psi_k(z)) 
    \;\equiv\; \frac{1}{2}\left(\ketbra{\bsu{t}}+\ketbra{\bsu{t-1}} - \Gate{C}Z\ketbra{\bsu{t}}{\bsu{t-1}} - \Gate{C}Z\ketbra{\bsu{t-1}}{\bsu{t}}\right)\Bigr|_{S_2},
\end{equation}
which on $S_2$ equals the restriction of
\begin{equation*}
    \frac{1}{2}\left(\ketbra{00}_{a_kb_k}+\ketbra{01}_{a_kb_k}+\ketbra{10}_{a_kb_k}\right) E_-(\bsu{t-1},\bsu{t})
    + \frac{1}{2}\ketbra{11}_{a_kb_k}\, E_+(\bsu{t-1},\bsu{t}).
\end{equation*}
Comparing with \cref{eq:combined}, we have $H_{\mathrm{prop}}^{(2)}\bigr|_{S_2} = \tilde{H} + H'$ where
\begin{equation*}
    H' = \frac{3}{2}\ketbra{00}_{a_kb_k}\,E_-(\bsu{t-1},\bsu{t}) 
    + \frac{1}{2}\ketbra{11}_{a_kb_k}\,E_+(\bsu{t-1},\bsu{t}),
\end{equation*}
which is manifestly positive semidefinite (as a sum of products of projectors with $E_\pm$, 
which are themselves positive semidefinite). Therefore $\tilde{H} \preceq H_{\mathrm{prop}}^{(2)}\bigr|_{S_2}$,
completing the proof of the proposition. \qed

%% file: main.bbl
\begin{thebibliography}{57}%
\makeatletter
\providecommand \@ifxundefined [1]{%
 \@ifx{#1\undefined}
}%
\providecommand \@ifnum [1]{%
 \ifnum #1\expandafter \@firstoftwo
 \else \expandafter \@secondoftwo
 \fi
}%
\providecommand \@ifx [1]{%
 \ifx #1\expandafter \@firstoftwo
 \else \expandafter \@secondoftwo
 \fi
}%
\providecommand \natexlab [1]{#1}%
\providecommand \enquote  [1]{``#1''}%
\providecommand \bibnamefont  [1]{#1}%
\providecommand \bibfnamefont [1]{#1}%
\providecommand \citenamefont [1]{#1}%
\providecommand \href@noop [0]{\@secondoftwo}%
\providecommand \href [0]{\begingroup \@sanitize@url \@href}%
\providecommand \@href[1]{\@@startlink{#1}\@@href}%
\providecommand \@@href[1]{\endgroup#1\@@endlink}%
\providecommand \@sanitize@url [0]{\catcode `\\12\catcode `\$12\catcode
  `\&12\catcode `\#12\catcode `\^12\catcode `\_12\catcode `\%12\relax}%
\providecommand \@@startlink[1]{}%
\providecommand \@@endlink[0]{}%
\providecommand \url  [0]{\begingroup\@sanitize@url \@url }%
\providecommand \@url [1]{\endgroup\@href {#1}{\urlprefix }}%
\providecommand \urlprefix  [0]{URL }%
\providecommand \Eprint [0]{\href }%
\providecommand \doibase [0]{https://doi.org/}%
\providecommand \selectlanguage [0]{\@gobble}%
\providecommand \bibinfo  [0]{\@secondoftwo}%
\providecommand \bibfield  [0]{\@secondoftwo}%
\providecommand \translation [1]{[#1]}%
\providecommand \BibitemOpen [0]{}%
\providecommand \bibitemStop [0]{}%
\providecommand \bibitemNoStop [0]{.\EOS\space}%
\providecommand \EOS [0]{\spacefactor3000\relax}%
\providecommand \BibitemShut  [1]{\csname bibitem#1\endcsname}%
\let\auto@bib@innerbib\@empty
\bibitem [{\citenamefont {Kitaev}\ \emph {et~al.}(2002)\citenamefont {Kitaev},
  \citenamefont {Shen},\ and\ \citenamefont {Vyalyi}}]{kitaev2002classical}%
  \BibitemOpen
  \bibfield  {author} {\bibinfo {author} {\bibfnamefont {A.~Y.}\ \bibnamefont
  {Kitaev}}, \bibinfo {author} {\bibfnamefont {A.~H.}\ \bibnamefont {Shen}},\
  and\ \bibinfo {author} {\bibfnamefont {M.~N.}\ \bibnamefont {Vyalyi}},\
  }\href@noop {} {\emph {\bibinfo {title} {{Classical and Quantum
  Computation}}}}\ (\bibinfo  {publisher} {American Mathematical Society},\
  \bibinfo {address} {USA},\ \bibinfo {year} {2002})\BibitemShut {NoStop}%
\bibitem [{\citenamefont {Kallaugher}\ \emph {et~al.}(2024)\citenamefont
  {Kallaugher}, \citenamefont {Parekh}, \citenamefont {Thompson}, \citenamefont
  {Wang},\ and\ \citenamefont {Yirka}}]{kallaugher2024complexity}%
  \BibitemOpen
  \bibfield  {author} {\bibinfo {author} {\bibfnamefont {J.}~\bibnamefont
  {Kallaugher}}, \bibinfo {author} {\bibfnamefont {O.}~\bibnamefont {Parekh}},
  \bibinfo {author} {\bibfnamefont {K.}~\bibnamefont {Thompson}}, \bibinfo
  {author} {\bibfnamefont {Y.}~\bibnamefont {Wang}},\ and\ \bibinfo {author}
  {\bibfnamefont {J.}~\bibnamefont {Yirka}},\ }\href@noop {} {\bibinfo {title}
  {Complexity classification of product state problems for local hamiltonians}}
  (\bibinfo {year} {2024}),\ \Eprint {https://arxiv.org/abs/2401.06725}
  {arXiv:2401.06725} \BibitemShut {NoStop}%
\bibitem [{\citenamefont {Elman}\ \emph {et~al.}(2021)\citenamefont {Elman},
  \citenamefont {Chapman},\ and\ \citenamefont {Flammia}}]{elman2021free}%
  \BibitemOpen
  \bibfield  {author} {\bibinfo {author} {\bibfnamefont {S.~J.}\ \bibnamefont
  {Elman}}, \bibinfo {author} {\bibfnamefont {A.}~\bibnamefont {Chapman}},\
  and\ \bibinfo {author} {\bibfnamefont {S.~T.}\ \bibnamefont {Flammia}},\
  }\href {https://doi.org/10.1007/s00220-021-04220-w} {\bibfield  {journal}
  {\bibinfo  {journal} {Communications in Mathematical Physics}\ }\textbf
  {\bibinfo {volume} {388}},\ \bibinfo {pages} {969} (\bibinfo {year}
  {2021})},\ \Eprint {https://arxiv.org/abs/2012.07857} {arXiv:2012.07857}
  \BibitemShut {NoStop}%
\bibitem [{\citenamefont {Chapman}\ \emph {et~al.}(2023)\citenamefont
  {Chapman}, \citenamefont {Elman},\ and\ \citenamefont
  {Mann}}]{chapman2023unified}%
  \BibitemOpen
  \bibfield  {author} {\bibinfo {author} {\bibfnamefont {A.}~\bibnamefont
  {Chapman}}, \bibinfo {author} {\bibfnamefont {S.~J.}\ \bibnamefont {Elman}},\
  and\ \bibinfo {author} {\bibfnamefont {R.~L.}\ \bibnamefont {Mann}},\
  }\href@noop {} {\bibinfo {title} {{A Unified Graph-Theoretic Framework for
  Free-Fermion Solvability}}} (\bibinfo {year} {2023}),\ \Eprint
  {https://arxiv.org/abs/2305.15625} {arXiv:2305.15625} \BibitemShut {NoStop}%
\bibitem [{\citenamefont {Richter}(2007)}]{richter2007two}%
  \BibitemOpen
  \bibfield  {author} {\bibinfo {author} {\bibfnamefont {P.~C.}\ \bibnamefont
  {Richter}},\ }\href@noop {} {\bibinfo {title} {Two remarks on the local
  {H}amiltonian problem}} (\bibinfo {year} {2007}),\ \Eprint
  {https://arxiv.org/abs/0712.4274} {arXiv:0712.4274} \BibitemShut {NoStop}%
\bibitem [{\citenamefont {Bravyi}(2015)}]{bravyi2015monte}%
  \BibitemOpen
  \bibfield  {author} {\bibinfo {author} {\bibfnamefont {S.}~\bibnamefont
  {Bravyi}},\ }\href@noop {} {\bibfield  {journal} {\bibinfo  {journal}
  {Quantum Info. Comput.}\ }\textbf {\bibinfo {volume} {15}},\ \bibinfo {pages}
  {1122} (\bibinfo {year} {2015})},\ \Eprint {https://arxiv.org/abs/1402.2295}
  {arXiv:1402.2295} \BibitemShut {NoStop}%
\bibitem [{\citenamefont {Stroeks}\ \emph {et~al.}(2022)\citenamefont
  {Stroeks}, \citenamefont {Helsen},\ and\ \citenamefont
  {Terhal}}]{stroeks2022spectral}%
  \BibitemOpen
  \bibfield  {author} {\bibinfo {author} {\bibfnamefont {M.~E.}\ \bibnamefont
  {Stroeks}}, \bibinfo {author} {\bibfnamefont {J.}~\bibnamefont {Helsen}},\
  and\ \bibinfo {author} {\bibfnamefont {B.~M.}\ \bibnamefont {Terhal}},\
  }\href {https://doi.org/10.1088/1367-2630/ac919c} {\bibfield  {journal}
  {\bibinfo  {journal} {New Journal of Physics}\ }\textbf {\bibinfo {volume}
  {24}},\ \bibinfo {pages} {103024} (\bibinfo {year} {2022})},\ \Eprint
  {https://arxiv.org/abs/2204.01113} {arXiv:2204.01113} \BibitemShut {NoStop}%
\bibitem [{\citenamefont {Weggemans}\ \emph {et~al.}(2024)\citenamefont
  {Weggemans}, \citenamefont {Folkertsma},\ and\ \citenamefont
  {Cade}}]{weggemans2024guidable}%
  \BibitemOpen
  \bibfield  {author} {\bibinfo {author} {\bibfnamefont {J.}~\bibnamefont
  {Weggemans}}, \bibinfo {author} {\bibfnamefont {M.}~\bibnamefont
  {Folkertsma}},\ and\ \bibinfo {author} {\bibfnamefont {C.}~\bibnamefont
  {Cade}},\ }in\ \href {https://doi.org/10.4230/LIPIcs.TQC.2024.10} {\emph
  {\bibinfo {booktitle} {19th Conference on the Theory of Quantum Computation,
  Communication and Cryptography (TQC 2024)}}},\ \bibinfo {series} {Leibniz
  International Proceedings in Informatics (LIPIcs)}, Vol.\ \bibinfo {volume}
  {310}\ (\bibinfo  {publisher} {Schloss Dagstuhl -- Leibniz-Zentrum f{\"u}r
  Informatik},\ \bibinfo {year} {2024})\ pp.\ \bibinfo {pages} {10:1--10:24},\
  \Eprint {https://arxiv.org/abs/2302.11578} {arXiv:2302.11578} \BibitemShut
  {NoStop}%
\bibitem [{\citenamefont {Fortnow}\ and\ \citenamefont
  {Rogers}(1999)}]{fortnow1999complexity}%
  \BibitemOpen
  \bibfield  {author} {\bibinfo {author} {\bibfnamefont {L.}~\bibnamefont
  {Fortnow}}\ and\ \bibinfo {author} {\bibfnamefont {J.}~\bibnamefont
  {Rogers}},\ }\href@noop {} {\bibfield  {journal} {\bibinfo  {journal}
  {Journal of Computer and System Sciences}\ }\textbf {\bibinfo {volume}
  {59}},\ \bibinfo {pages} {240} (\bibinfo {year} {1999})},\ \Eprint
  {https://arxiv.org/abs/cs/9811023} {arXiv:cs/9811023} \BibitemShut {NoStop}%
\bibitem [{\citenamefont {Fenner}\ \emph {et~al.}(1999)\citenamefont {Fenner},
  \citenamefont {Green}, \citenamefont {Homer},\ and\ \citenamefont
  {Pruim}}]{fenner1999determining}%
  \BibitemOpen
  \bibfield  {author} {\bibinfo {author} {\bibfnamefont {S.}~\bibnamefont
  {Fenner}}, \bibinfo {author} {\bibfnamefont {F.}~\bibnamefont {Green}},
  \bibinfo {author} {\bibfnamefont {S.}~\bibnamefont {Homer}},\ and\ \bibinfo
  {author} {\bibfnamefont {R.}~\bibnamefont {Pruim}},\ }\href@noop {}
  {\bibfield  {journal} {\bibinfo  {journal} {Proceedings of the Royal Society
  of London. Series A: Mathematical, Physical and Engineering Sciences}\
  }\textbf {\bibinfo {volume} {455}},\ \bibinfo {pages} {3953} (\bibinfo {year}
  {1999})},\ \Eprint {https://arxiv.org/abs/quant-ph/9812056}
  {arXiv:quant-ph/9812056} \BibitemShut {NoStop}%
\bibitem [{\citenamefont {Liu}(2021)}]{liu2021stoqma}%
  \BibitemOpen
  \bibfield  {author} {\bibinfo {author} {\bibfnamefont {Y.}~\bibnamefont
  {Liu}},\ }in\ \href {https://doi.org/10.4230/LIPIcs.TQC.2021.4} {\emph
  {\bibinfo {booktitle} {16th Conference on the Theory of Quantum Computation,
  Communication and Cryptography (TQC 2021)}}},\ Vol.\ \bibinfo {volume} {197}\
  (\bibinfo  {publisher} {Schloss Dagstuhl -- Leibniz-Zentrum f{\"u}r
  Informatik},\ \bibinfo {year} {2021})\ pp.\ \bibinfo {pages} {4:1--4:22},\
  \Eprint {https://arxiv.org/abs/2011.05733} {arXiv:2011.05733} \BibitemShut
  {NoStop}%
\bibitem [{\citenamefont {Jiang}(2025)}]{jiang2025local}%
  \BibitemOpen
  \bibfield  {author} {\bibinfo {author} {\bibfnamefont {J.}~\bibnamefont
  {Jiang}},\ }\href {https://doi.org/10.1103/PRXQuantum.6.020312} {\bibfield
  {journal} {\bibinfo  {journal} {PRX Quantum}\ }\textbf {\bibinfo {volume}
  {6}},\ \bibinfo {pages} {020312} (\bibinfo {year} {2025})},\ \Eprint
  {https://arxiv.org/abs/2309.10155} {arXiv:2309.10155} \BibitemShut {NoStop}%
\bibitem [{\citenamefont {Kempe}\ \emph {et~al.}(2006)\citenamefont {Kempe},
  \citenamefont {Kitaev},\ and\ \citenamefont {Regev}}]{kempe2006complexity}%
  \BibitemOpen
  \bibfield  {author} {\bibinfo {author} {\bibfnamefont {J.}~\bibnamefont
  {Kempe}}, \bibinfo {author} {\bibfnamefont {A.}~\bibnamefont {Kitaev}},\ and\
  \bibinfo {author} {\bibfnamefont {O.}~\bibnamefont {Regev}},\ }\href
  {https://doi.org/10.1137/S0097539704445226} {\bibfield  {journal} {\bibinfo
  {journal} {SIAM Journal on Computing}\ }\textbf {\bibinfo {volume} {35}},\
  \bibinfo {pages} {1070} (\bibinfo {year} {2006})},\ \Eprint
  {https://arxiv.org/abs/quant-ph/0406180} {arXiv:quant-ph/0406180}
  \BibitemShut {NoStop}%
\bibitem [{\citenamefont {Oliveira}\ and\ \citenamefont
  {Terhal}(2008)}]{oliveira2008complexity}%
  \BibitemOpen
  \bibfield  {author} {\bibinfo {author} {\bibfnamefont {R.}~\bibnamefont
  {Oliveira}}\ and\ \bibinfo {author} {\bibfnamefont {B.~M.}\ \bibnamefont
  {Terhal}},\ }\href@noop {} {\bibfield  {journal} {\bibinfo  {journal}
  {Quantum Info. Comput.}\ }\textbf {\bibinfo {volume} {8}},\ \bibinfo {pages}
  {900} (\bibinfo {year} {2008})},\ \Eprint
  {https://arxiv.org/abs/quant-ph/0504050} {arXiv:quant-ph/0504050}
  \BibitemShut {NoStop}%
\bibitem [{\citenamefont {Schuch}\ and\ \citenamefont
  {Verstraete}(2009)}]{schuch2009computational}%
  \BibitemOpen
  \bibfield  {author} {\bibinfo {author} {\bibfnamefont {N.}~\bibnamefont
  {Schuch}}\ and\ \bibinfo {author} {\bibfnamefont {F.}~\bibnamefont
  {Verstraete}},\ }\href {https://doi.org/10.1038/nphys1370} {\bibfield
  {journal} {\bibinfo  {journal} {Nature Physics}\ }\textbf {\bibinfo {volume}
  {5}},\ \bibinfo {pages} {732} (\bibinfo {year} {2009})},\ \Eprint
  {https://arxiv.org/abs/0712.0483} {arXiv:0712.0483} \BibitemShut {NoStop}%
\bibitem [{\citenamefont {Piddock}\ and\ \citenamefont
  {Montanaro}(2017)}]{piddock2017complexity}%
  \BibitemOpen
  \bibfield  {author} {\bibinfo {author} {\bibfnamefont {S.}~\bibnamefont
  {Piddock}}\ and\ \bibinfo {author} {\bibfnamefont {A.}~\bibnamefont
  {Montanaro}},\ }\href {https://doi.org/10.5555/3179553.3179559} {\bibfield
  {journal} {\bibinfo  {journal} {Quantum Info. Comput.}\ }\textbf {\bibinfo
  {volume} {17}},\ \bibinfo {pages} {636} (\bibinfo {year} {2017})},\ \Eprint
  {https://arxiv.org/abs/1506.04014} {arXiv:1506.04014} \BibitemShut {NoStop}%
\bibitem [{\citenamefont {Aaronson}\ \emph {et~al.}(2005)\citenamefont
  {Aaronson}, \citenamefont {Kuperberg},\ and\ \citenamefont
  {Granade}}]{aaronson2005complexity}%
  \BibitemOpen
  \bibfield  {author} {\bibinfo {author} {\bibfnamefont {S.}~\bibnamefont
  {Aaronson}}, \bibinfo {author} {\bibfnamefont {G.}~\bibnamefont
  {Kuperberg}},\ and\ \bibinfo {author} {\bibfnamefont {C.}~\bibnamefont
  {Granade}},\ }\href@noop {} {\bibinfo {title} {The complexity zoo}} (\bibinfo
  {year} {2005})\BibitemShut {NoStop}%
\bibitem [{\citenamefont {ten Haaf}\ \emph {et~al.}(1995)\citenamefont {ten
  Haaf}, \citenamefont {van Bemmel}, \citenamefont {van Leeuwen}, \citenamefont
  {van Saarloos},\ and\ \citenamefont {Ceperley}}]{tenHaaf1995proof}%
  \BibitemOpen
  \bibfield  {author} {\bibinfo {author} {\bibfnamefont {D.~F.~B.}\
  \bibnamefont {ten Haaf}}, \bibinfo {author} {\bibfnamefont {H.~J.~M.}\
  \bibnamefont {van Bemmel}}, \bibinfo {author} {\bibfnamefont {J.~M.~J.}\
  \bibnamefont {van Leeuwen}}, \bibinfo {author} {\bibfnamefont
  {W.}~\bibnamefont {van Saarloos}},\ and\ \bibinfo {author} {\bibfnamefont
  {D.~M.}\ \bibnamefont {Ceperley}},\ }\href
  {https://doi.org/10.1103/PhysRevB.51.13039} {\bibfield  {journal} {\bibinfo
  {journal} {Phys. Rev. B}\ }\textbf {\bibinfo {volume} {51}},\ \bibinfo
  {pages} {13039} (\bibinfo {year} {1995})}\BibitemShut {NoStop}%
\bibitem [{\citenamefont {Bravyi}\ \emph {et~al.}(2023)\citenamefont {Bravyi},
  \citenamefont {Carleo}, \citenamefont {Gosset},\ and\ \citenamefont
  {Liu}}]{bravyi2023rapidly}%
  \BibitemOpen
  \bibfield  {author} {\bibinfo {author} {\bibfnamefont {S.}~\bibnamefont
  {Bravyi}}, \bibinfo {author} {\bibfnamefont {G.}~\bibnamefont {Carleo}},
  \bibinfo {author} {\bibfnamefont {D.}~\bibnamefont {Gosset}},\ and\ \bibinfo
  {author} {\bibfnamefont {Y.}~\bibnamefont {Liu}},\ }\href
  {https://doi.org/10.22331/q-2023-11-07-1173} {\bibfield  {journal} {\bibinfo
  {journal} {{Quantum}}\ }\textbf {\bibinfo {volume} {7}},\ \bibinfo {pages}
  {1173} (\bibinfo {year} {2023})},\ \Eprint {https://arxiv.org/abs/2207.07044}
  {arXiv:2207.07044} \BibitemShut {NoStop}%
\bibitem [{\citenamefont {Bravyi}\ \emph {et~al.}(2006)\citenamefont {Bravyi},
  \citenamefont {DiVincenzo}, \citenamefont {Oliveira},\ and\ \citenamefont
  {Terhal}}]{bravyi2006complexity}%
  \BibitemOpen
  \bibfield  {author} {\bibinfo {author} {\bibfnamefont {S.}~\bibnamefont
  {Bravyi}}, \bibinfo {author} {\bibfnamefont {D.~P.}\ \bibnamefont
  {DiVincenzo}}, \bibinfo {author} {\bibfnamefont {R.~I.}\ \bibnamefont
  {Oliveira}},\ and\ \bibinfo {author} {\bibfnamefont {B.~M.}\ \bibnamefont
  {Terhal}},\ }\href@noop {} {\bibinfo {title} {The complexity of stoquastic
  local hamiltonian problems}} (\bibinfo {year} {2006}),\ \Eprint
  {https://arxiv.org/abs/quant-ph/0606140} {arXiv:quant-ph/0606140}
  \BibitemShut {NoStop}%
\bibitem [{\citenamefont {Kempe}\ and\ \citenamefont
  {Regev}(2003)}]{kempe2003local}%
  \BibitemOpen
  \bibfield  {author} {\bibinfo {author} {\bibfnamefont {J.}~\bibnamefont
  {Kempe}}\ and\ \bibinfo {author} {\bibfnamefont {O.}~\bibnamefont {Regev}},\
  }\href@noop {} {\bibinfo {title} {3-local hamiltonian is qma-complete}}
  (\bibinfo {year} {2003}),\ \Eprint {https://arxiv.org/abs/quant-ph/0302079}
  {arXiv:quant-ph/0302079} \BibitemShut {NoStop}%
\bibitem [{\citenamefont {Van~den Nest}(2011)}]{vandennest2011simulating}%
  \BibitemOpen
  \bibfield  {author} {\bibinfo {author} {\bibfnamefont {M.}~\bibnamefont
  {Van~den Nest}},\ }\href@noop {} {\bibfield  {journal} {\bibinfo  {journal}
  {Quantum Information \& Computation}\ }\textbf {\bibinfo {volume} {11}},\
  \bibinfo {pages} {784} (\bibinfo {year} {2011})},\ \Eprint
  {https://arxiv.org/abs/0911.1624} {arXiv:0911.1624} \BibitemShut {NoStop}%
\bibitem [{\citenamefont {Deshpande}\ \emph {et~al.}(2022)\citenamefont
  {Deshpande}, \citenamefont {Gorshkov},\ and\ \citenamefont
  {Fefferman}}]{deshpande2022importance}%
  \BibitemOpen
  \bibfield  {author} {\bibinfo {author} {\bibfnamefont {A.}~\bibnamefont
  {Deshpande}}, \bibinfo {author} {\bibfnamefont {A.~V.}\ \bibnamefont
  {Gorshkov}},\ and\ \bibinfo {author} {\bibfnamefont {B.}~\bibnamefont
  {Fefferman}},\ }\href {https://doi.org/10.1103/prxquantum.3.040327}
  {\bibfield  {journal} {\bibinfo  {journal} {PRX Quantum}\ }\textbf {\bibinfo
  {volume} {3}},\ \bibinfo {pages} {040327} (\bibinfo {year} {2022})},\ \Eprint
  {https://arxiv.org/abs/2007.11582} {arXiv:2007.11582} \BibitemShut {NoStop}%
\bibitem [{\citenamefont {Gharibian}\ and\ \citenamefont
  {Le~Gall}(2023)}]{gharibian2023dequantizing}%
  \BibitemOpen
  \bibfield  {author} {\bibinfo {author} {\bibfnamefont {S.}~\bibnamefont
  {Gharibian}}\ and\ \bibinfo {author} {\bibfnamefont {F.}~\bibnamefont
  {Le~Gall}},\ }\href {https://doi.org/10.1137/22M1513721} {\bibfield
  {journal} {\bibinfo  {journal} {SIAM Journal on Computing}\ }\textbf
  {\bibinfo {volume} {52}},\ \bibinfo {pages} {1009} (\bibinfo {year}
  {2023})},\ \Eprint {https://arxiv.org/abs/2111.09079} {arXiv:2111.09079}
  \BibitemShut {NoStop}%
\bibitem [{\citenamefont {Nielsen}\ and\ \citenamefont
  {Chuang}(2010)}]{nielsen2010quantum}%
  \BibitemOpen
  \bibfield  {author} {\bibinfo {author} {\bibfnamefont {M.~A.}\ \bibnamefont
  {Nielsen}}\ and\ \bibinfo {author} {\bibfnamefont {I.~L.}\ \bibnamefont
  {Chuang}},\ }\href@noop {} {\emph {\bibinfo {title} {Quantum Computation and
  Quantum Information: 10th Anniversary Edition}}}\ (\bibinfo  {publisher}
  {Cambridge University Press},\ \bibinfo {year} {2010})\BibitemShut {NoStop}%
\bibitem [{\citenamefont {Watrous}(2008)}]{watrous2008quantum}%
  \BibitemOpen
  \bibfield  {author} {\bibinfo {author} {\bibfnamefont {J.}~\bibnamefont
  {Watrous}},\ }\href@noop {} {\bibinfo {title} {{Quantum Computational
  Complexity}}} (\bibinfo {year} {2008}),\ \Eprint
  {https://arxiv.org/abs/0804.3401} {arXiv:0804.3401} \BibitemShut {NoStop}%
\bibitem [{\citenamefont {Gharibian}\ \emph {et~al.}(2015)\citenamefont
  {Gharibian}, \citenamefont {Huang}, \citenamefont {Landau},\ and\
  \citenamefont {Shin}}]{gharibian2015quantum}%
  \BibitemOpen
  \bibfield  {author} {\bibinfo {author} {\bibfnamefont {S.}~\bibnamefont
  {Gharibian}}, \bibinfo {author} {\bibfnamefont {Y.}~\bibnamefont {Huang}},
  \bibinfo {author} {\bibfnamefont {Z.}~\bibnamefont {Landau}},\ and\ \bibinfo
  {author} {\bibfnamefont {S.~W.}\ \bibnamefont {Shin}},\ }\href
  {https://doi.org/10.1561/0400000066} {\bibfield  {journal} {\bibinfo
  {journal} {Foundations and Trends® in Theoretical Computer Science}\
  }\textbf {\bibinfo {volume} {10}},\ \bibinfo {pages} {159} (\bibinfo {year}
  {2015})},\ \Eprint {https://arxiv.org/abs/1401.3916} {arXiv:1401.3916}
  \BibitemShut {NoStop}%
\bibitem [{\citenamefont {Waite}\ \emph {et~al.}(2023)\citenamefont {Waite},
  \citenamefont {Mann},\ and\ \citenamefont {Elman}}]{hamiltonianjungle2023}%
  \BibitemOpen
  \bibfield  {author} {\bibinfo {author} {\bibfnamefont {G.}~\bibnamefont
  {Waite}}, \bibinfo {author} {\bibfnamefont {R.~L.}\ \bibnamefont {Mann}},\
  and\ \bibinfo {author} {\bibfnamefont {S.~J.}\ \bibnamefont {Elman}},\ }\href
  {https://hamiltonianjungle.xyz} {\bibinfo {title} {The hamiltonian jungle}}
  (\bibinfo {year} {2023})\BibitemShut {NoStop}%
\bibitem [{\citenamefont {Giles}\ and\ \citenamefont
  {Selinger}(2013)}]{giles2013exact}%
  \BibitemOpen
  \bibfield  {author} {\bibinfo {author} {\bibfnamefont {B.}~\bibnamefont
  {Giles}}\ and\ \bibinfo {author} {\bibfnamefont {P.}~\bibnamefont
  {Selinger}},\ }\href {https://doi.org/10.1103/PhysRevA.87.032332} {\bibfield
  {journal} {\bibinfo  {journal} {Phys. Rev. A}\ }\textbf {\bibinfo {volume}
  {87}},\ \bibinfo {pages} {032332} (\bibinfo {year} {2013})},\ \Eprint
  {https://arxiv.org/abs/1212.0506} {arXiv:1212.0506} \BibitemShut {NoStop}%
\bibitem [{\citenamefont {Nam}\ \emph {et~al.}(2020)\citenamefont {Nam},
  \citenamefont {Su},\ and\ \citenamefont {Maslov}}]{nam2020approximate}%
  \BibitemOpen
  \bibfield  {author} {\bibinfo {author} {\bibfnamefont {Y.}~\bibnamefont
  {Nam}}, \bibinfo {author} {\bibfnamefont {Y.}~\bibnamefont {Su}},\ and\
  \bibinfo {author} {\bibfnamefont {D.}~\bibnamefont {Maslov}},\ }\href@noop {}
  {\bibfield  {journal} {\bibinfo  {journal} {NPJ Quantum Information}\
  }\textbf {\bibinfo {volume} {6}},\ \bibinfo {pages} {26} (\bibinfo {year}
  {2020})},\ \Eprint {https://arxiv.org/abs/1803.04933} {arXiv:1803.04933}
  \BibitemShut {NoStop}%
\bibitem [{\citenamefont {Rudolph}(2024)}]{rudolph2024towards}%
  \BibitemOpen
  \bibfield  {author} {\bibinfo {author} {\bibfnamefont {D.}~\bibnamefont
  {Rudolph}},\ }\href@noop {} {\bibinfo {title} {Towards a universal gateset
  for qma$_1$}} (\bibinfo {year} {2024}),\ \Eprint
  {https://arxiv.org/abs/2411.02681} {arXiv:2411.02681} \BibitemShut {NoStop}%
\bibitem [{\citenamefont {Waite}\ and\ \citenamefont
  {Bremner}(2025)}]{waite2025complexity}%
  \BibitemOpen
  \bibfield  {author} {\bibinfo {author} {\bibfnamefont {G.}~\bibnamefont
  {Waite}}\ and\ \bibinfo {author} {\bibfnamefont {M.~J.}\ \bibnamefont
  {Bremner}},\ }\href@noop {} {\bibinfo {title} {The complexity of local
  stoquastic hamiltonians on 2d lattices}} (\bibinfo {year} {2025}),\ \Eprint
  {https://arxiv.org/abs/2502.14244} {arXiv:2502.14244} \BibitemShut {NoStop}%
\bibitem [{\citenamefont {Zachos}\ and\ \citenamefont
  {Furer}(1987)}]{zachos1987probabilistic}%
  \BibitemOpen
  \bibfield  {author} {\bibinfo {author} {\bibfnamefont {S.}~\bibnamefont
  {Zachos}}\ and\ \bibinfo {author} {\bibfnamefont {M.}~\bibnamefont {Furer}},\
  }in\ \href@noop {} {\emph {\bibinfo {booktitle} {Proc. of the Seventh
  Conference on Foundations of Software Technology and Theoretical Computer
  Science}}}\ (\bibinfo  {publisher} {Springer-Verlag},\ \bibinfo {address}
  {Berlin, Heidelberg},\ \bibinfo {year} {1987})\ pp.\ \bibinfo {pages}
  {443--455}\BibitemShut {NoStop}%
\bibitem [{\citenamefont {Cade}\ \emph {et~al.}(2023)\citenamefont {Cade},
  \citenamefont {Folkertsma}, \citenamefont {Gharibian}, \citenamefont
  {Hayakawa}, \citenamefont {Le~Gall}, \citenamefont {Morimae},\ and\
  \citenamefont {Weggemans}}]{cade2023improved}%
  \BibitemOpen
  \bibfield  {author} {\bibinfo {author} {\bibfnamefont {C.}~\bibnamefont
  {Cade}}, \bibinfo {author} {\bibfnamefont {M.}~\bibnamefont {Folkertsma}},
  \bibinfo {author} {\bibfnamefont {S.}~\bibnamefont {Gharibian}}, \bibinfo
  {author} {\bibfnamefont {R.}~\bibnamefont {Hayakawa}}, \bibinfo {author}
  {\bibfnamefont {F.}~\bibnamefont {Le~Gall}}, \bibinfo {author} {\bibfnamefont
  {T.}~\bibnamefont {Morimae}},\ and\ \bibinfo {author} {\bibfnamefont
  {J.}~\bibnamefont {Weggemans}},\ }in\ \href
  {https://doi.org/10.4230/LIPICS.ICALP.2023.32} {\emph {\bibinfo {booktitle}
  {50th International Colloquium on Automata, Languages, and Programming (ICALP
  2023)}}},\ Vol.\ \bibinfo {volume} {261}\ (\bibinfo  {publisher} {Schloss
  Dagstuhl -- Leibniz-Zentrum f{\"u}r Informatik},\ \bibinfo {year} {2023})\
  pp.\ \bibinfo {pages} {32:1--32:19},\ \Eprint
  {https://arxiv.org/abs/2207.10250} {arXiv:2207.10250} \BibitemShut {NoStop}%
\bibitem [{\citenamefont {Bremner}\ \emph {et~al.}(2025)\citenamefont
  {Bremner}, \citenamefont {Ji}, \citenamefont {Li}, \citenamefont
  {Mathieson},\ and\ \citenamefont {Morales}}]{bremner2025parameterized}%
  \BibitemOpen
  \bibfield  {author} {\bibinfo {author} {\bibfnamefont {M.~J.}\ \bibnamefont
  {Bremner}}, \bibinfo {author} {\bibfnamefont {Z.}~\bibnamefont {Ji}},
  \bibinfo {author} {\bibfnamefont {X.}~\bibnamefont {Li}}, \bibinfo {author}
  {\bibfnamefont {L.}~\bibnamefont {Mathieson}},\ and\ \bibinfo {author}
  {\bibfnamefont {M.~E.~S.}\ \bibnamefont {Morales}},\ }\bibfield  {journal}
  {\bibinfo  {journal} {ACM Transactions on Quantum Computing}\ }\textbf
  {\bibinfo {volume} {6}},\ \href {https://doi.org/10.1145/3759156}
  {10.1145/3759156} (\bibinfo {year} {2025}),\ \Eprint
  {https://arxiv.org/abs/2211.05325} {arXiv:2211.05325} \BibitemShut {NoStop}%
\bibitem [{\citenamefont {Hastings}(2004)}]{hastings2004locality}%
  \BibitemOpen
  \bibfield  {author} {\bibinfo {author} {\bibfnamefont {M.~B.}\ \bibnamefont
  {Hastings}},\ }\href {https://doi.org/10.1103/physrevlett.93.140402}
  {\bibfield  {journal} {\bibinfo  {journal} {Physical Review Letters}\
  }\textbf {\bibinfo {volume} {93}},\ \bibinfo {pages} {140402} (\bibinfo
  {year} {2004})},\ \Eprint {https://arxiv.org/abs/cond-mat/0405587}
  {arXiv:cond-mat/0405587} \BibitemShut {NoStop}%
\bibitem [{\citenamefont {Verstraete}\ and\ \citenamefont
  {Cirac}(2006)}]{verstraete2006matrix}%
  \BibitemOpen
  \bibfield  {author} {\bibinfo {author} {\bibfnamefont {F.}~\bibnamefont
  {Verstraete}}\ and\ \bibinfo {author} {\bibfnamefont {J.~I.}\ \bibnamefont
  {Cirac}},\ }\bibfield  {journal} {\bibinfo  {journal} {Physical Review B}\
  }\textbf {\bibinfo {volume} {73}},\ \href
  {https://doi.org/10.1103/physrevb.73.094423} {10.1103/physrevb.73.094423}
  (\bibinfo {year} {2006}),\ \Eprint {https://arxiv.org/abs/cond-mat/0505140}
  {arXiv:cond-mat/0505140} \BibitemShut {NoStop}%
\bibitem [{\citenamefont {Anshu}\ \emph {et~al.}(2022)\citenamefont {Anshu},
  \citenamefont {Arad},\ and\ \citenamefont {Gosset}}]{anshu2022area}%
  \BibitemOpen
  \bibfield  {author} {\bibinfo {author} {\bibfnamefont {A.}~\bibnamefont
  {Anshu}}, \bibinfo {author} {\bibfnamefont {I.}~\bibnamefont {Arad}},\ and\
  \bibinfo {author} {\bibfnamefont {D.}~\bibnamefont {Gosset}},\ }in\ \href
  {https://doi.org/10.1145/3519935.3519962} {\emph {\bibinfo {booktitle}
  {Proceedings of the 54th Annual ACM SIGACT Symposium on Theory of
  Computing}}},\ \bibinfo {series and number} {STOC '22}\ (\bibinfo
  {publisher} {ACM},\ \bibinfo {year} {2022})\ \Eprint
  {https://arxiv.org/abs/2103.02492} {arXiv:2103.02492} \BibitemShut {NoStop}%
\bibitem [{\citenamefont {Ge}\ and\ \citenamefont {Eisert}(2015)}]{ge2015area}%
  \BibitemOpen
  \bibfield  {author} {\bibinfo {author} {\bibfnamefont {Y.}~\bibnamefont
  {Ge}}\ and\ \bibinfo {author} {\bibfnamefont {J.}~\bibnamefont {Eisert}},\
  }\bibfield  {journal} {\bibinfo  {journal} {New Journal of Physics}\ }\textbf
  {\bibinfo {volume} {18}},\ \href
  {https://doi.org/10.1088/1367-2630/18/8/083026}
  {10.1088/1367-2630/18/8/083026} (\bibinfo {year} {2015}),\ \Eprint
  {https://arxiv.org/abs/1411.2995} {arXiv:1411.2995} \BibitemShut {NoStop}%
\bibitem [{\citenamefont {Eisert}\ \emph {et~al.}(2010)\citenamefont {Eisert},
  \citenamefont {Cramer},\ and\ \citenamefont {Plenio}}]{eisert2010colloquium}%
  \BibitemOpen
  \bibfield  {author} {\bibinfo {author} {\bibfnamefont {J.}~\bibnamefont
  {Eisert}}, \bibinfo {author} {\bibfnamefont {M.}~\bibnamefont {Cramer}},\
  and\ \bibinfo {author} {\bibfnamefont {M.~B.}\ \bibnamefont {Plenio}},\
  }\href {https://doi.org/10.1103/revmodphys.82.277} {\bibfield  {journal}
  {\bibinfo  {journal} {Reviews of Modern Physics}\ }\textbf {\bibinfo {volume}
  {82}},\ \bibinfo {pages} {277} (\bibinfo {year} {2010})},\ \Eprint
  {https://arxiv.org/abs/0808.3773} {arXiv:0808.3773} \BibitemShut {NoStop}%
\bibitem [{\citenamefont {Huang}(2020)}]{huang2020local}%
  \BibitemOpen
  \bibfield  {author} {\bibinfo {author} {\bibfnamefont {Y.}~\bibnamefont
  {Huang}},\ }in\ \href {https://doi.org/10.1109/isit44484.2020.9174029} {\emph
  {\bibinfo {booktitle} {2020 IEEE International Symposium on Information
  Theory (ISIT)}}}\ (\bibinfo  {publisher} {IEEE},\ \bibinfo {year} {2020})\
  pp.\ \bibinfo {pages} {1927--1932},\ \Eprint
  {https://arxiv.org/abs/1411.6614} {arXiv:1411.6614} \BibitemShut {NoStop}%
\bibitem [{\citenamefont {Kuperberg}(2015)}]{kuperberg2015hard}%
  \BibitemOpen
  \bibfield  {author} {\bibinfo {author} {\bibfnamefont {G.}~\bibnamefont
  {Kuperberg}},\ }\href@noop {} {\bibfield  {journal} {\bibinfo  {journal}
  {Theory OF Computing}\ }\textbf {\bibinfo {volume} {11}},\ \bibinfo {pages}
  {183} (\bibinfo {year} {2015})},\ \Eprint {https://arxiv.org/abs/0908.0512}
  {arXiv:0908.0512} \BibitemShut {NoStop}%
\bibitem [{\citenamefont {Milne}(2003)}]{milne2003fields}%
  \BibitemOpen
  \bibfield  {author} {\bibinfo {author} {\bibfnamefont {J.~S.}\ \bibnamefont
  {Milne}},\ }\href {http://www.galois-group.net/theory/math594fS.pdf}
  {\bibinfo {title} {Fields and galois theory}} (\bibinfo {year}
  {2003})\BibitemShut {NoStop}%
\bibitem [{\citenamefont {Cotler}\ \emph {et~al.}(2021)\citenamefont {Cotler},
  \citenamefont {Huang},\ and\ \citenamefont {McClean}}]{cotler2021revisiting}%
  \BibitemOpen
  \bibfield  {author} {\bibinfo {author} {\bibfnamefont {J.}~\bibnamefont
  {Cotler}}, \bibinfo {author} {\bibfnamefont {H.-Y.}\ \bibnamefont {Huang}},\
  and\ \bibinfo {author} {\bibfnamefont {J.~R.}\ \bibnamefont {McClean}},\
  }\href@noop {} {\bibinfo {title} {Revisiting dequantization and quantum
  advantage in learning tasks}} (\bibinfo {year} {2021}),\ \Eprint
  {https://arxiv.org/abs/2112.00811} {arXiv:2112.00811 [quant-ph]} \BibitemShut
  {NoStop}%
\bibitem [{\citenamefont {{Froese Fischer}}(1987)}]{froese1987general}%
  \BibitemOpen
  \bibfield  {author} {\bibinfo {author} {\bibfnamefont {C.}~\bibnamefont
  {{Froese Fischer}}},\ }\href {https://doi.org/10.1016/0010-4655(87)90053-1}
  {\bibfield  {journal} {\bibinfo  {journal} {Computer Physics Communications}\
  }\textbf {\bibinfo {volume} {43}},\ \bibinfo {pages} {355} (\bibinfo {year}
  {1987})}\BibitemShut {NoStop}%
\bibitem [{\citenamefont {White}(1992)}]{white1992density}%
  \BibitemOpen
  \bibfield  {author} {\bibinfo {author} {\bibfnamefont {S.~R.}\ \bibnamefont
  {White}},\ }\href {https://doi.org/10.1103/PhysRevLett.69.2863} {\bibfield
  {journal} {\bibinfo  {journal} {Physical Review Letters}\ }\textbf {\bibinfo
  {volume} {69}},\ \bibinfo {pages} {2863} (\bibinfo {year}
  {1992})}\BibitemShut {NoStop}%
\bibitem [{\citenamefont {Vidal}(2003)}]{vidal2003efficient}%
  \BibitemOpen
  \bibfield  {author} {\bibinfo {author} {\bibfnamefont {G.}~\bibnamefont
  {Vidal}},\ }\href {https://doi.org/10.1103/PhysRevLett.91.147902} {\bibfield
  {journal} {\bibinfo  {journal} {Phys. Rev. Lett.}\ }\textbf {\bibinfo
  {volume} {91}},\ \bibinfo {pages} {147902} (\bibinfo {year} {2003})},\
  \Eprint {https://arxiv.org/abs/quant-ph/0301063} {arXiv:quant-ph/0301063}
  \BibitemShut {NoStop}%
\bibitem [{\citenamefont {Verstraete}\ \emph {et~al.}(2008)\citenamefont
  {Verstraete}, \citenamefont {Murg},\ and\ \citenamefont
  {Cirac}}]{verstraete2008matrix}%
  \BibitemOpen
  \bibfield  {author} {\bibinfo {author} {\bibfnamefont {F.}~\bibnamefont
  {Verstraete}}, \bibinfo {author} {\bibfnamefont {V.}~\bibnamefont {Murg}},\
  and\ \bibinfo {author} {\bibfnamefont {J.}~\bibnamefont {Cirac}},\ }\href
  {https://doi.org/10.1080/14789940801912366} {\bibfield  {journal} {\bibinfo
  {journal} {Advances in Physics}\ }\textbf {\bibinfo {volume} {57}},\ \bibinfo
  {pages} {143} (\bibinfo {year} {2008})},\ \Eprint
  {https://arxiv.org/abs/0907.2796} {arXiv:0907.2796} \BibitemShut {NoStop}%
\bibitem [{\citenamefont {Carleo}\ and\ \citenamefont
  {Troyer}(2017)}]{carleo2017solving}%
  \BibitemOpen
  \bibfield  {author} {\bibinfo {author} {\bibfnamefont {G.}~\bibnamefont
  {Carleo}}\ and\ \bibinfo {author} {\bibfnamefont {M.}~\bibnamefont
  {Troyer}},\ }\href {https://doi.org/10.1126/science.aag2302} {\bibfield
  {journal} {\bibinfo  {journal} {Science}\ }\textbf {\bibinfo {volume}
  {355}},\ \bibinfo {pages} {602} (\bibinfo {year} {2017})},\ \Eprint
  {https://arxiv.org/abs/1606.02318} {arXiv:1606.02318} \BibitemShut {NoStop}%
\bibitem [{\citenamefont {Huang}\ \emph {et~al.}(2025)\citenamefont {Huang},
  \citenamefont {Preskill},\ and\ \citenamefont
  {Soleimanifar}}]{huang2025certifying}%
  \BibitemOpen
  \bibfield  {author} {\bibinfo {author} {\bibfnamefont {H.-Y.}\ \bibnamefont
  {Huang}}, \bibinfo {author} {\bibfnamefont {J.}~\bibnamefont {Preskill}},\
  and\ \bibinfo {author} {\bibfnamefont {M.}~\bibnamefont {Soleimanifar}},\
  }\href@noop {} {\bibfield  {journal} {\bibinfo  {journal} {Nature Physics}\
  ,\ \bibinfo {pages} {1}} (\bibinfo {year} {2025})},\ \Eprint
  {https://arxiv.org/abs/2404.07281} {arXiv:2404.07281} \BibitemShut {NoStop}%
\bibitem [{\citenamefont {Gillespie}(1977)}]{gillespie1977exact}%
  \BibitemOpen
  \bibfield  {author} {\bibinfo {author} {\bibfnamefont {D.~T.}\ \bibnamefont
  {Gillespie}},\ }\href {https://doi.org/10.1021/j100540a008} {\bibfield
  {journal} {\bibinfo  {journal} {The Journal of Physical Chemistry}\ }\textbf
  {\bibinfo {volume} {81}},\ \bibinfo {pages} {2340} (\bibinfo {year}
  {1977})}\BibitemShut {NoStop}%
\bibitem [{\citenamefont {Bravyi}\ \emph {et~al.}(2022)\citenamefont {Bravyi},
  \citenamefont {Gosset},\ and\ \citenamefont {Liu}}]{bravyi2022simulate}%
  \BibitemOpen
  \bibfield  {author} {\bibinfo {author} {\bibfnamefont {S.}~\bibnamefont
  {Bravyi}}, \bibinfo {author} {\bibfnamefont {D.}~\bibnamefont {Gosset}},\
  and\ \bibinfo {author} {\bibfnamefont {Y.}~\bibnamefont {Liu}},\ }\href@noop
  {} {\bibfield  {journal} {\bibinfo  {journal} {Physical Review Letters}\
  }\textbf {\bibinfo {volume} {128}},\ \bibinfo {pages} {220503} (\bibinfo
  {year} {2022})},\ \Eprint {https://arxiv.org/abs/2112.08499}
  {arXiv:2112.08499} \BibitemShut {NoStop}%
\bibitem [{\citenamefont {Tsuneda}(2014)}]{tsuneda2014density}%
  \BibitemOpen
  \bibfield  {author} {\bibinfo {author} {\bibfnamefont {T.}~\bibnamefont
  {Tsuneda}},\ }\href {https://doi.org/10.1007/978-4-431-54825-6} {\emph
  {\bibinfo {title} {Density Functional Theory in Quantum Chemistry}}}\
  (\bibinfo  {publisher} {Springer Japan},\ \bibinfo {year} {2014})\BibitemShut
  {NoStop}%
\bibitem [{\citenamefont {Arute}\ \emph {et~al.}(2020)\citenamefont {Arute},
  \citenamefont {Arya}, \citenamefont {Babbush}, \citenamefont {Bacon},
  \citenamefont {Bardin}, \citenamefont {Barends} \emph
  {et~al.}}]{arute2020hartree}%
  \BibitemOpen
  \bibfield  {author} {\bibinfo {author} {\bibfnamefont {F.}~\bibnamefont
  {Arute}}, \bibinfo {author} {\bibfnamefont {K.}~\bibnamefont {Arya}},
  \bibinfo {author} {\bibfnamefont {R.}~\bibnamefont {Babbush}}, \bibinfo
  {author} {\bibfnamefont {D.}~\bibnamefont {Bacon}}, \bibinfo {author}
  {\bibfnamefont {J.~C.}\ \bibnamefont {Bardin}}, \bibinfo {author}
  {\bibfnamefont {R.}~\bibnamefont {Barends}}, \emph {et~al.},\ }\href
  {https://doi.org/10.1126/science.abb9811} {\bibfield  {journal} {\bibinfo
  {journal} {Science}\ }\textbf {\bibinfo {volume} {369}},\ \bibinfo {pages}
  {1084} (\bibinfo {year} {2020})}\BibitemShut {NoStop}%
\bibitem [{\citenamefont {Impagliazzo}\ \emph {et~al.}(2002)\citenamefont
  {Impagliazzo}, \citenamefont {Kabanets},\ and\ \citenamefont
  {Wigderson}}]{impagliazzo2002search}%
  \BibitemOpen
  \bibfield  {author} {\bibinfo {author} {\bibfnamefont {R.}~\bibnamefont
  {Impagliazzo}}, \bibinfo {author} {\bibfnamefont {V.}~\bibnamefont
  {Kabanets}},\ and\ \bibinfo {author} {\bibfnamefont {A.}~\bibnamefont
  {Wigderson}},\ }\href {https://doi.org/10.1016/S0022-0000(02)00024-7}
  {\bibfield  {journal} {\bibinfo  {journal} {Journal of Computer and System
  Sciences}\ }\textbf {\bibinfo {volume} {65}},\ \bibinfo {pages} {672}
  (\bibinfo {year} {2002})},\ \bibinfo {note} {special Issue on Complexity
  2001}\BibitemShut {NoStop}%
\bibitem [{\citenamefont {Canonne}\ and\ \citenamefont
  {Rubinfeld}(2014)}]{canonne2014testing}%
  \BibitemOpen
  \bibfield  {author} {\bibinfo {author} {\bibfnamefont {C.~L.}\ \bibnamefont
  {Canonne}}\ and\ \bibinfo {author} {\bibfnamefont {R.}~\bibnamefont
  {Rubinfeld}},\ }\href@noop {} {\bibfield  {journal} {\bibinfo  {journal}
  {Electron. Colloquium Comput. Complex.}\ }\textbf {\bibinfo {volume} {TR14}}
  (\bibinfo {year} {2014})},\ \Eprint {https://arxiv.org/abs/1402.3835}
  {arXiv:1402.3835} \BibitemShut {NoStop}%
\bibitem [{\citenamefont {Alicki}\ \emph {et~al.}(2009)\citenamefont {Alicki},
  \citenamefont {Fannes},\ and\ \citenamefont
  {Horodecki}}]{alicki2009thermalization}%
  \BibitemOpen
  \bibfield  {author} {\bibinfo {author} {\bibfnamefont {R.}~\bibnamefont
  {Alicki}}, \bibinfo {author} {\bibfnamefont {M.}~\bibnamefont {Fannes}},\
  and\ \bibinfo {author} {\bibfnamefont {M.}~\bibnamefont {Horodecki}},\
  }\href@noop {} {\bibfield  {journal} {\bibinfo  {journal} {Journal of Physics
  A: Mathematical and Theoretical}\ }\textbf {\bibinfo {volume} {42}},\
  \bibinfo {pages} {065303} (\bibinfo {year} {2009})},\ \Eprint
  {https://arxiv.org/abs/0810.4584} {arXiv:0810.4584} \BibitemShut {NoStop}%
\end{thebibliography}%
